\documentclass[reqno]{amsart}
\usepackage[foot]{amsaddr}
\usepackage[centertableaux]{ytableau}
\usepackage{setspace,tikz,xcolor,mathrsfs,listings,multicol}
\usepackage{rotating}
\usepackage{booktabs, afterpage}
\usepackage[vcentermath]{youngtab}
\usepackage[inline]{enumitem}
\usepackage[all,cmtip]{xy}
\usetikzlibrary{arrows,matrix}
\usepackage[margin=1.13in]{geometry}
\usetikzlibrary{shapes.misc}
\tikzset{cross/.style={cross out, draw=black, minimum size=2*(#1-\pgflinewidth), inner sep=0pt, outer sep=0pt},
	cross/.default={0.4 cm},
	dynkinnode/.style={circle, draw, fill=white, inner sep = 0pt, minimum width=0.35cm}
}
\usepackage[colorlinks=true, pdfstartview=FitV, linkcolor=blue, citecolor=blue, urlcolor=blue]{hyperref}

\newcommand{\gl}{\mathfrak{gl}}
\newcommand{\fsl}{\mathfrak{sl}}
\newcommand{\agl}{\widehat{\gl}}
\newcommand{\asl}{\widehat{\fsl}}
\newcommand{\dcrys}{\mathcal{B}_{\bze,\overline{m-r}}} 
\newcommand{\tab}{\mathcal{T}}	
\newcommand{\up}[1]{{#1_\uparrow}}
\newcommand{\down}[1]{{#1_\downarrow}}
\newcommand{\inlinetab}[1]{\begin{array}{|c|} \hline #1 \\\hline \end{array}} 
\newcommand{\fundrep}{\mathbf{V}} 

\DeclareMathOperator{\id}{id} 

\DeclareMathOperator{\col}{col} 
\DeclareMathOperator{\sg}{sg} 
\DeclareMathOperator{\rsg}{rsg} 
\newcommand{\Ieven}{I_{\text{even}}}
\newcommand{\Iodd}{I_{\text{odd}}}

\newcommand{\ZZ}{\mathbb{Z}}

\newcommand{\bze}{\overline{0}}
\newcommand{\bon}{\overline{1}}
\newcommand{\btw}{\overline{2}}
\newcommand{\bth}{\overline{3}}
\newcommand{\bfo}{\overline{4}}

\newcommand{\bk}{\overline{k}}

\newcommand{\bm}{\overline{m}}

\newcommand{\bfze}{\mathbf{0}}
\newcommand{\bfon}{\mathbf{1}}
\newcommand{\bftw}{\mathbf{2}}
\newcommand{\bfth}{\mathbf{3}}
\newcommand{\bfX}{\mathbf{X}}

\definecolor{darkred}{rgb}{0.7,0,0} 
\newcommand{\defn}[1]{{\color{darkred}\emph{#1}}} 

\theoremstyle{plain}
\newtheorem{thm}{Theorem}[section]
\newtheorem{lemma}[thm]{Lemma}
\newtheorem{conj}[thm]{Conjecture}
\newtheorem{prop}[thm]{Proposition}

\theoremstyle{definition}
\newtheorem{dfn}[thm]{Definition}
\newtheorem{ex}[thm]{Example}
\newtheorem{remark}[thm]{Remark}
\newtheorem{obs}[thm]{Observation}
\numberwithin{equation}{section}

\usepackage[colorinlistoftodos]{todonotes}
\begin{document}
	\title[SCA for $\agl(m|n)$]{Soliton cellular automata for the affine general linear Lie superalgebra}
	
	\author[M.~Ryan]{Mitchell Ryan}
	\author[B.~Solomon]{Benjamin Solomon}
	\address[M.~Ryan and B.~Solomon]{School of Mathematics and Physics, University of Queensland, St.\ Lucia, QLD 4072, Australia}
	\email{research@mitchell-ryan.com}
	\email{hart.loire@gmail.com}
	
	\keywords{soliton, crystal, cellular automaton}
	\subjclass[2010]{17B37, 05E10, 82B23, 37B15}

	\begin{abstract}
		The box-ball system (BBS) is a cellular automaton that is an ultradiscrete
		analogue of the Korteweg--de Vries equation, a non-linear PDE used to
		model water waves. 
		In 2001, Hikami and Inoue generalised the BBS to the general linear Lie superalgebra $\mathfrak{gl}(m|n)$.
		We further generalise the Hikami--Inoue BBS to column tableaux
		using the Kirillov--Reshetikhin crystals for $\widehat{\mathfrak{gl}}{(m|n)}$ devised by Kwon and Okado in 2021,
		where we find similar solitonic behaviour under certain conditions.
	\end{abstract}
	
	\maketitle
	
	\setcounter{tocdepth}{1}
	\tableofcontents

	\section{Introduction}
	The box-ball system (BBS) is an integrable nonlinear dynamical system that has connections to both classical and quantum integrable systems. The Korteweg--De Vries (KdV) equation is one such classical system which describes shallow water waves in a one-dimensional channel~\cite{KdV1895}. It was shown by Kruskal and Zabusky that solutions to the KdV equation separate into solitonic waves that move with speed proportional to their amplitude and maintain their shape under collision with other solitons~\cite{KZ64}. It was later discovered that the ultradiscretisation of these soliton solutions produces the BBS~\cite{TTMS96}. In addition to its connection to classical systems, the BBS also emerges from quantum integrable systems such as the six-vertex lattice model from statistical mechanics~\cite{B89}. The symmetries of this particular system are governed by the quantum group $U_q'(\fsl_2)$. Under crystallisation (the $q \rightarrow 0$ limit) the system is frozen to the ground state and produces the BBS~\cite{HI99}. The position of BBS within the realm of classical and quantum integrable systems opens its analysis to a variety of methods. Moreover, the discrete nature of the system provides an important connection to combinatorics.
	
	An important development in the analysis of the BBS comes from the construction of crystal bases by Kashiwara~\cite{K91,K90}. This allowed for the reformulation of crystallisation in terms of the crystal theory of quantum affine algebras~\cite{HKT00}, leading to a crystal theoretic formulation of the BBS. This formulation utilises the `classical'
	crystal $B^{1,s}$, which is the crystal base of an $s$-fold symmetric tensor representation of $U_q(\fsl_n)$ promoted to the Kirillov--Reshetikhin (KR) crystal of $U_q'(\asl_n)$~\cite{KKMMNN92} by adding additional crystal operators.
	States of the system are then defined as elements of $ (B^{1,1})^{\otimes \infty}$.
	The time evolution of these states utilises the existence of the combinatorial $R$-matrix, a unique isomorphism between the tensor product of KR crystals,
	$R\colon B \otimes B' \to B' \otimes B$~\cite{KKMMNN92}.
	The time evolution is given by repeated applications of this $R$-matrix together with a carrier.
	Below is an example of a BBS constructed from the KR crystal $B^{1,1}$ in $A_1^{(1)}$.
	
	\begin{ex} \label{example:coloured_bbs}
		For $B^{1,1}$ in $A_1^{(1)}$ with carrier $B^{1,3}$, the time evolution of the state
		\[
		\inlinetab{3}\otimes\inlinetab{3}\otimes\inlinetab{2}\otimes \inlinetab{1} \otimes \inlinetab{1} \otimes \inlinetab{1} \otimes \inlinetab{1} \otimes \cdots
		\]
		is the state
		\[
		\inlinetab{1} \otimes \inlinetab{1} \otimes \inlinetab{1} \otimes\inlinetab{3}\otimes\inlinetab{3}\otimes\inlinetab{2}\otimes  \inlinetab{1} \otimes \cdots .
		\]
		Each $ \inlinetab{1} $ represents a vacuum element (i.e.\ an empty box).
		The computation of this time evolution can be represented in the following diagram
		
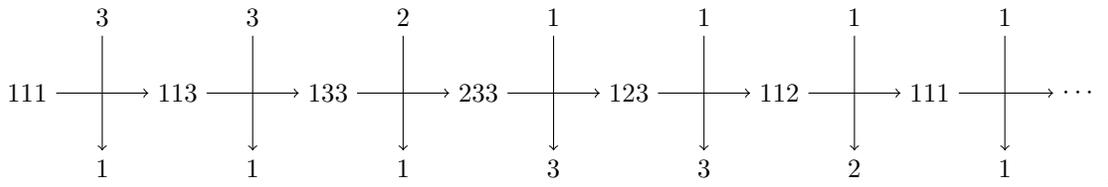
\[
		\begin{tikzpicture}
			\node (carr0) at (-1,0) {$111
				$};
			
			\node (carr1) at (1,0) {$ 113
				$};
			
			\node (carr2) at (3,0) {$133
				$};
			
			\node (carr3) at (5,0) {$ 233
				$};
			
			\node (carr4) at (7,0) {$ 123
				$};
			
			\node (carr5) at (9,0) {$112
				$};
			
			\node (carr6) at (11,0) {$ 111 $};
			
			\node (carr7) at (13,0) {$ \cdots $};
			
			\foreach \i[evaluate=\i as \si using int(\i+1)] in {0,...,6} {
				\draw [->] (carr\i) -- (carr\si);
			}
			
			\node (b0) at (0,1) {$ 3 $};
			
			\node (b1) at (2,1) {$ 3 $};
			
			\node (b2) at (4,1) {$ 2 $};
			
			\foreach \i in {3,...,5} {
				\node (b\i) at ({2*\i},1) {$1$};
			}
			
			\foreach \i in {0,2,1} {
				\node (bt\i) at ({2*\i},-1) {$ 1
					$};
			}
			
			\node (bt3) at (6,-1) {$3$};
			\node (bt4) at (8,-1) {$3$};
			\node (bt5) at (10,-1) {$2$};
			\foreach \i in {0,...,5} {
				\draw [->] (b\i) -- (bt\i);
			}
			
			\node (b6) at (12,1) {$ 1 $};
			\node (bt6) at (12,-1) {$ 1 $};
			\draw [->] (b6) -- (bt6);
			
		\end{tikzpicture}
		,
		\]
		where each crossing represents the application of the R-matrix, $R\colon B^{1,3} \otimes B^{1,1} \rightarrow B^{1,1} \otimes B^{1,3}$. The top row represents the initial state, the bottom row represents the state after one time evolution, and the middle represents how the carrier changes during the evolution. 
	\end{ex}
	
	Within the BBS there exist states exhibiting solitonic behaviour; that is, states containing elements of $(B^{1,1})^{\otimes d}$ within $(B^{1,1})^{\otimes \infty}$ that move with speed corresponding to their length and are stable under collisions (this stability is called scattering). Such elements are called solitons. These solitons are the ultradiscrete analogue of the KdV solitons. In Example~\ref{example:coloured_bbs}, $ \inlinetab{3}\otimes\inlinetab{3}\otimes\inlinetab{2} $ is a soliton. For more detail we refer the reader to~\cite{FOY00,HHIKTT01}.
	
	In 2001, Hikami and Inoue generalised the BBS using crystals for the general linear Lie superalgebra $\gl(m|n)$ and showed that similar solitonic behaviour existed in this supersymmetric system~\cite{HI00}. Work by Yamada~\cite{Yamada04} generalised the system in a different manner by considering the crystal $B^{r,s}$ of $U'_q(\asl_n)$, producing a system with $ r $ rows. This paper uses the KR crystals for $\agl(m|n)$ devised by Kwon and Okado~\cite{KO18} to generalise the Hikami--Inoue BBS analogously to Yamada's generalisation of the $ U_q'(\asl_n) $ BBS. Each $ \agl(m|n) $ KR crystal is parameterised by a Young diagram $Y$, and the crystal is identified with the set of semistandard Young tableaux (SSYT) of shape $ Y $. We are primarily interested in $ B^{r,s} $; the crystal of rectangular SSYT of height $ r $ and width $ s $. In our generalised BBS, we define states as elements of $\bigl( B^{r,1} \bigr)^{\otimes \infty} $.  We similarly have an $R$-matrix giving a bijection of tensor products of crystals $ B^{r_1,s_1}\otimes B^{r_2,s_2} \to B^{r_2,s_2}\otimes B^{r_1,s_1} $. The $R$-matrix can be explicitly calculated with the RSK algorithm, using the modified Schensted's bumping algorithm outlined in Section~\ref{sec:comb_R}.  This allows us to define the time evolution of the system analogously to the classical case. Taking $r =1$ reduces our system to the Hikami--Inoue BBS~\cite{HI00}. The following example gives a two-soliton state within our generalised system.
	\begin{ex}\label{ex:glmnscattering}
		Consider the $ U_q(\agl(3|1)) $ crystal $ B^{2,1} $.
		In the following diagram is a state (in $ (B^{2,1})^{\otimes \infty} $) evolved over four time steps starting at time $ t=0 $.
		The maximal weight element (
		$
		\begin{smallmatrix}
			\bth\\
			\btw
		\end{smallmatrix}
		$
		) is represented as a dot.
		\[
		\begin{array}{cc}
			t = 0 &
			\begin{tikzpicture}[scale = 0.9, baseline = -2]
				\def\s{0.5};
				\foreach \x in {-7,...,-4,-1,0,2,3,...,11} {
					\draw[fill=black] (\x*\s,0) circle (0.025);
				}
				\foreach \count/\valone/\valtwo/\valthree in {0/\bon/\bon/\bon,1/\bth/\bth/\btw}{
					\node (a) at (-3*\s,0 + \count*\s) {$\valone$};
					\node (b) at (-2*\s,0 + \count*\s) {$\valtwo$};
					\node (c) at (1*\s,0 + \count*\s) {$\valthree$};
				}
			\end{tikzpicture}
			\\[7pt]
			
			t = 1 &
			
			\begin{tikzpicture}[scale = 0.9, baseline = -2]
				\def\s{0.5};
				\foreach \x in {-7,...,-2,1,3,4,...,11} {
					\draw[fill=black] (\x*\s,0) circle (0.025);
				}
				\foreach \count/\valone/\valtwo/\valthree in {0/\bon/\bon/\bon,1/\bth/\bth/\btw}{
					\node (a) at (-1*\s,0 + \count*\s) {$\valone$};
					\node (b) at (0*\s,0 + \count*\s) {$\valtwo$};
					\node (c) at (2*\s,0 + \count*\s) {$\valthree$};
				}
			\end{tikzpicture}
			\\[7pt]
			
			t = 2 &
			
			\begin{tikzpicture}[scale = 0.9, baseline = -2]
				\def\s{0.5};
				\foreach \x in {-7,...,0,4,5,...,11} {
					\draw[fill=black] (\x*\s,0) circle (0.025);
				}
				\foreach \count/\valone/\valtwo/\valthree in {0/\bon/\bon/\bon,1/\bth/\bth/\btw}{
					\node (a) at (1*\s,0 + \count*\s) {$\valone$};
					\node (b) at (2*\s,0 + \count*\s) {$\valtwo$};
					\node (c) at (3*\s,0 + \count*\s) {$\valthree$};
				}
			\end{tikzpicture}
			\\[7pt]
			
			t = 3 &
			\begin{tikzpicture}[scale = 0.9, baseline = -2]
				\def\s{0.5};
				\foreach \x in {-7,...,2,6,7,...,11} {
					\draw[fill=black] (\x*\s,0) circle (0.025);
				}
				\foreach \count/\valone/\valtwo/\valthree in {0/\bon/\bon/\bon,1/\bth/\btw/\bth}{
					\node (a) at (3*\s,0 + \count*\s) {$\valone$};
					\node (b) at (4*\s,0 + \count*\s) {$\valtwo$};
					\node (c) at (5*\s,0 + \count*\s) {$\valthree$};
				}
			\end{tikzpicture}
			\\[7pt]
			
			t = 4 &
			\begin{tikzpicture}[scale = 0.9, baseline = -2]
				\def\s{0.5};
				\foreach \x in {-7,...,3,5,8,9,...,11} {
					\draw[fill=black] (\x*\s,0) circle (0.025);
				}
				\foreach \count/\valone/\valtwo/\valthree in {0/\bon/\bon/\bon,1/\bth/\btw/\bth}{
					\node (a) at (4*\s,0 + \count*\s) {$\valone$};
					\node (b) at (6*\s,0 + \count*\s) {$\valtwo$};
					\node (c) at (7*\s,0 + \count*\s) {$\valthree$};
				}
			\end{tikzpicture}
		\end{array}
		\]
		At $ t=1 $, we observe both 
		$ 
		\begin{smallmatrix}
			\bth & \bth\\
			\bon & \bon
		\end{smallmatrix}
		$ and $ 
		\begin{smallmatrix}
			\btw\\
			\bon
		\end{smallmatrix}
		$
		have moved with speed proportional to their length.
		They collide at times $ t=2 $ and $ t=3 $,
		before separating back into two solitons at $t = 4$
		(stability under collisions).
		This demonstrates solitonic behaviour in our generalised system.
		
		Note that, at $ t=4 $, $ 
		\begin{smallmatrix}
			\btw & \bth\\
			\bon & \bon
		\end{smallmatrix}
		$ is one step ahead (to the right) of where $ 
		\begin{smallmatrix}
			\bth & \bth\\
			\bon & \bon
		\end{smallmatrix}
		$ would be if there had been no collision.
		Similarly, $ 
		\begin{smallmatrix}
			\bth\\
			\bon
		\end{smallmatrix}
		$ is one step behind where $ 
		\begin{smallmatrix}
			\btw\\
			\bon
		\end{smallmatrix}
		$ would be.
		This phenomenon is called the phase shift and is a shadow of the non-linearity. The phase shift is governed by the integer valued energy function.
	\end{ex}
	
	Example~\ref{ex:glmnscattering} shows that there are objects within our system exhibiting solitonic behaviour. 
	Additionally, we can also find coupled solitons, which move with constant speed but contain many overlapping uncoupled solitons. Conjecture~\ref{conj:uncoupling} proposes that these coupled solitons can be uncoupled upon collision. Conjecture~\ref{conj:separate} says that, given sufficient time, every state will separate into (potentially coupled) solitons.
	
	We prove two main theorems.
	Theorem~\ref{thm:speed} provides sufficient conditions for a soliton in a BBS state to move with speed corresponding to its length. Conjecture~\ref{conj:allspeedlength} claims that these conditions are also necessary. 
	In addition, Conjecture~\ref{conj:dividingline} provides a relationship between the conditions of Theorem~\ref{thm:speed} and the number of uncoupled solitons within a soliton.
	Theorem~\ref{thm:scattering} provides a class of solitons which satisfy the conditions of Theorem~\ref{thm:speed} and also maintain their shape under collision (see Example~\ref{ex:glmnscattering}). In addition, we characterise the phase shift of these solitons in terms of the energy function, and show that the solitons after collision can be computed using the combinatorial $ R $-matrix without needing to compute the entire sequence of time evolutions. In most cases, the proof of Theorem~\ref{thm:scattering} reduces the behaviour of the system to modified version of the height $1$ system in~\cite{HI00}. An interesting case occurs when, in a $ U_q(\gl(m|n)) $ system, the height of the solitons is equal to $ m $;
	in this case the highest weight states have a form that is not analogous to any of the highest weight states in the non-super symmetric system of~\cite{Yamada04}.
	
	The paper is organised as follows. In Section~\ref{sec:background} we quickly review the crystal base theory required for our purposes, including the computation of crystal operators, the combinatorial $ R $-matrix and the energy function. In Section~\ref{sec:SBBS} we present the explicit structure of our generalised system and outline the process of time evolution. In Section~\ref{sec:solitons}, we present our two main theorems and our conjectures. The proofs of the theorems can be found in Appendices~\ref{appendix:single_soliton_proof} and~\ref{sec:mainthmproof}, respectively.
	Appendices~\ref{appendix:classRmatprf} and~\ref{appendix:m=rclassRmatprf} contain the proofs of some technical lemmas used in Appendix~\ref{sec:mainthmproof}.
	
	\section{Background}\label{sec:background}
	\subsection{The affine general linear Lie superalgebra}
	The original BBS can be derived from the affine Lie algebra $ \asl_n $.
	Our supersymmetric BBS
	is instead derived from the affine general linear Lie superalgebra $\agl(m|n)$ 
	and its corresponding quantum group $U_q(\agl(m|n))$ (in the sense of~\cite{KO18}).
	Let $I = \Ieven \sqcup \Iodd $ be the indexing set of simple roots,
	where $ \Ieven = \{\overline{m-1}, \ldots ,\overline{1}, 1,\ldots , n-1 \}$ 
	and $ \Iodd = \{0,\bze\}$.
	It is useful to set $ I_- = \{\overline{m-1},\ldots,\bon\} $ and $ I_+ = \{1,\ldots,n-1\} $,
	so that $ \Ieven = I_-\sqcup I_+ $.
	
	The Dynkin diagram for $ \agl(m|n) $ is:
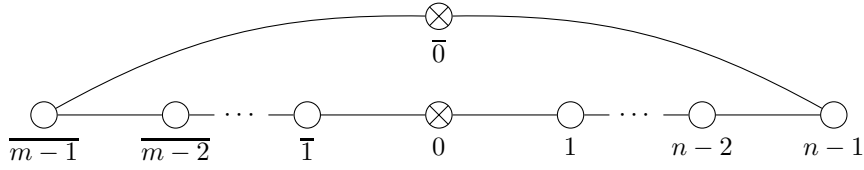
	\begin{center}
		\begin{tikzpicture}[scale=1.75]
			\node[dynkinnode, label={below:$\overline{m-1}$}] (bm-1) at (-3,0) {};
			\node[dynkinnode, label={below:$\overline{m-2}$}] (bm-2) at (-2,0) {};
			\node (dotsleft) at (-1.5,0) {$ \cdots $};
			\node[dynkinnode, label={below:$\bon$}] (bon) at (-1,0) {};
			\node[dynkinnode, label={below:$0$}] (0) at (0,0) {};
			\node[cross=0.14cm] (0cross) at (0) {};
			\node[dynkinnode, label={below:$1$}] (1) at (1,0) {};
			\node (dotsright) at (1.5,0) {$ \cdots $};
			\node[dynkinnode, label={below:$n-2$}] (n-2) at (2,0) {};
			\node[dynkinnode, label={below:$n-1$}] (n-1) at (3,0) {};
			
			\node[dynkinnode, label={below:$\bze$}] (bze) at (0,0.75) {};
			\node[cross=0.14cm] (bzecross) at (bze) {};
			
			\draw (bm-1) 
			-- (bm-2)
			-- (dotsleft)
			-- (bon)
			-- (0)
			-- (1)
			-- (dotsright)
			-- (n-2)
			-- (n-1);
			\draw (n-1) to[bend  right=15] (bze);
			\draw (bze) to[bend  right=15] (bm-1);
		\end{tikzpicture}
	\end{center}
	where 
	\begin{tikzpicture}
		\node[dynkinnode] (iso) at (0,0){};
		\node[cross=0.14cm] (isocross) at (iso){};
	\end{tikzpicture}
	denotes an isotropic simple root. 
	The Dynkin diagram for the finite dimensional $ \gl(m|n) $
	can be obtained from the above Dynkin diagram by removing the $ \bze $ node.
	
	The fundamental representation of $U_q(\agl(m|n))$
	is an $(m+n)$-dimensional super vector space $ \fundrep = \fundrep_+ \oplus \fundrep_- $.
	The fundamental representation admits a \defn{crystal base} $ \{v_b \mid b\in B\} $
	with $ B=B_-\sqcup B_+ $  where $ B_- = \{\bm,\overline{m-1},\ldots, \bon \} $ and $ B_+ = \{1,\ldots,n-1,n\} $,
	which gives rise to the following \defn{crystal graph}:
	
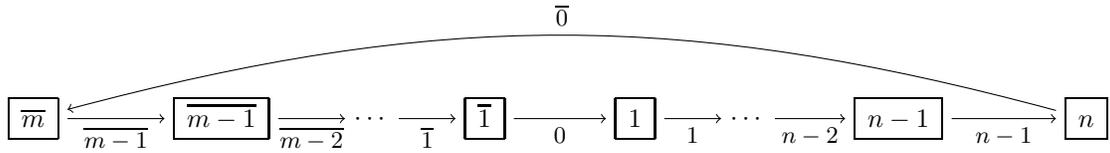
\begin{center}
		\begin{tikzpicture}
			\node (bm) at (-7,0) {$ \inlinetab{\bm} $};
			\node (bm-1) at (-4.5,0) {$ \inlinetab{\overline{m-1}} $};
			\node (dotsleft) at (-2.5,0) {$ \cdots $};
			\node (bon) at (-1,0) {$ \inlinetab{\bon} $};
			\node (1) at (1,0) {$ \inlinetab{1} $};
			\node (dotsright) at (2.5,0) {$ \cdots $};
			\node (n-1) at (4.5,0) {$ \inlinetab{n-1} $};
			\node (n) at (7,0) {$ \inlinetab{n} $};
			
			\draw [->] (bm) --  node[below]{\small $ \overline{m-1} $} (bm-1);
			\draw [->] (bm-1) -- node[below]{\small $ \overline{m-2} $} (dotsleft);
			\draw [->] (dotsleft) -- node[below]{\small $ \bon $} (bon);
			\draw [->] (bon) -- node[below]{\small $ 0 $} (1);
			\draw [->] (1) -- node[below]{\small $ 1 $} (dotsright);
			\draw [->] (dotsright) -- node[below]{\small $ n-2 $} (n-1);
			\draw [->] (n-1) -- node[below]{\small $ n-1 $} (n);
			\draw [->] (n) to[bend right=15] node[above]{\small $ \bze $} (bm);
		\end{tikzpicture}
	\end{center}
	where $ \inlinetab{b'}\xrightarrow{i}\inlinetab{b} $
	if and only if $ f_iv_{b'}=v_b $ 
	(equivalently, $ e_i v_b = v_{b'} $). 
	If we instead consider the fundamental representation of the finite type $ U_q(\gl(m|n)) $,
	the crystal graph the same as above but without the $ \bze $ arrow. 
	We can interpret the finite type crystal graph (with arrow labels removed) as a total ordering;
	explicitly, $ \bm < \cdots < \bon < 1 < \cdots < n $.
	For a more detailed explanation of crystals for $ U_q(\agl(m|n)) $,
	see~\cite{KO18}. 
	
	\subsection{Finite type crystals and tableaux}\label{sec:operators}
	Now we restrict our attention to the finite type crystal.
	Let $\fundrep^{\otimes N}$ be the $N$-th tensor power of the fundamental representation
	of $ U_q(\gl(m|n)) $.
	It can be shown that all tensor powers with $N\geq 1$ are completely reducible.
	Moreover, the irreducible subrepresentations (up to isomorphism) are in bijection 
	with Young diagrams of $ (m|n) $-hook shape~\cite{BR87, BKK00}.
	This bijection is derived using a map from crystal base elements to semistandard Young tableaux.
	In this context, a tableau is called \defn{semistandard} if 
	the rows are weakly (resp. strictly) increasing for letters
	in $ B_- $ (resp. $ B_+ $)
	and the columns are weakly (resp. strictly) increasing for letters in $ B_+ $ (resp. $ B_- $).
	
	We map crystal base elements 
	to Young diagrams using a modified version of Schensted's bumping algorithm.
	For inserting $ i\in B $ into a tableau $ \tab $, which we will denote
	$ i\rightarrow \tab $,
	the bumping algorithm is as follows:
	\begin{enumerate}
		\item\label{enum:bump1} For $ i\in B_{+} $, (resp. $ i\in B_{-} $):
		if none of the boxes in the first column of $ \tab $ are strictly larger than $ i $
		(resp. larger than or equal to $ i $)
		then add a box containing $i$ at the bottom of the column.
		\item\label{enum:bump2} Otherwise,
		for the topmost $ \inlinetab{j} $ with $ j>i $ (resp. $ j\ge i $) in the first column,
		replace $ \inlinetab{j} $  with $ \inlinetab{i} $.
		Then, insert $ j $ into the second column 
		following analogous steps~\ref{enum:bump1} and~\ref{enum:bump2}.
		\item Repeat until the bumped number can be put in a new box.
	\end{enumerate}
	
	\begin{ex}
		The following is an example computation of the bumping algorithm:
		\begin{align*}
			&\btw \rightarrow
			\begin{array}{|c|c|c|c|}
				\hline
				\raisebox{-1pt}{$ \bth $} & \raisebox{-1pt}{$ \bth $} & \raisebox{-1pt}{$ 1 $} & \raisebox{-1pt}{$ 3 $} \\
				\hline
				\raisebox{-1pt}{$ \btw $} & \raisebox{-1pt}{$ 1 $} & \raisebox{-1pt}{$ 2 $} & \raisebox{-1pt}{$ 5 $}\\
				\hline
			\end{array}\\
			&=
			\begin{array}{|c|c|c|c|}
				\multicolumn{1}{c}{\btw}\\
				\multicolumn{1}{c}{\downarrow}\\
				\hline
				\raisebox{-1pt}{$ \bth $} & \raisebox{-1pt}{$ \bth $} & \raisebox{-1pt}{$ 1 $} & \raisebox{-1pt}{$ 3 $} \\
				\hline
				\raisebox{-1pt}{$ {\bf\btw} $} & \raisebox{-1pt}{$ 1 $} & \raisebox{-1pt}{$ 2 $} & \raisebox{-1pt}{$ 5 $}\\
				\hline
			\end{array}
			=
			\begin{array}{|c|c|c|c|}
				\multicolumn{1}{c}{}&\multicolumn{1}{c}{\btw}\\
				\multicolumn{1}{c}{}&\multicolumn{1}{c}{\downarrow}\\
				\hline
				\raisebox{-1pt}{$ \bth $} & \raisebox{-1pt}{$ \bth $} & \raisebox{-1pt}{$ 1 $} & \raisebox{-1pt}{$ 3 $} \\
				\hline
				\raisebox{-1pt}{$ \btw $} & \raisebox{-1pt}{$ {\bf1} $} & \raisebox{-1pt}{$ 2 $} & \raisebox{-1pt}{$ 5 $}\\
				\hline
			\end{array}
			=
			\begin{array}{|c|c|c|c|}
				\multicolumn{1}{c}{}&\multicolumn{1}{c}{}&\multicolumn{1}{c}{1}\\
				\multicolumn{1}{c}{}&\multicolumn{1}{c}{}&\multicolumn{1}{c}{\downarrow}\\
				\hline
				\raisebox{-1pt}{$ \bth $} & \raisebox{-1pt}{$ \bth $} & \raisebox{-1pt}{$ 1 $} & \raisebox{-1pt}{$ 3 $} \\
				\hline
				\raisebox{-1pt}{$ \btw $} & \raisebox{-1pt}{$ \btw $} & \raisebox{-1pt}{$ {\bf2} $} & \raisebox{-1pt}{$ 5 $}\\
				\hline
			\end{array}
			=
			\begin{array}{|c|c|c|c|}
				\multicolumn{1}{c}{}&\multicolumn{1}{c}{}&\multicolumn{1}{c}{}&\multicolumn{1}{c}{2}\\
				\multicolumn{1}{c}{}&\multicolumn{1}{c}{}&\multicolumn{1}{c}{}&\multicolumn{1}{c}{\downarrow}\\
				\hline
				\raisebox{-1pt}{$ \bth $} & \raisebox{-1pt}{$ \bth $} & \raisebox{-1pt}{$ 1 $} & \raisebox{-1pt}{$ {\bf3} $} \\
				\hline
				\raisebox{-1pt}{$ \btw $} & \raisebox{-1pt}{$ \btw $} & \raisebox{-1pt}{$ 1 $} & \raisebox{-1pt}{$ 5 $}\\
				\hline
			\end{array}
			=
			\begin{array}{|c|c|c|c|c}
				\multicolumn{1}{c}{}&\multicolumn{1}{c}{}&\multicolumn{1}{c}{}&\multicolumn{1}{c}{}&\multicolumn{1}{c}{3}\\
				\multicolumn{1}{c}{}&\multicolumn{1}{c}{}&\multicolumn{1}{c}{}&\multicolumn{1}{c}{}&\multicolumn{1}{c}{\downarrow}\\
				\cline{1-4}
				\raisebox{-1pt}{$ \bth $} & \raisebox{-1pt}{$ \bth $} & \raisebox{-1pt}{$ 1 $} & \raisebox{-1pt}{$ {2} $} \\
				\cline{1-4}
				\raisebox{-1pt}{$ \btw $} & \raisebox{-1pt}{$ \btw $} & \raisebox{-1pt}{$ 1 $} & \raisebox{-1pt}{$ 5 $}\\
				\cline{1-4}
			\end{array}\\
			&=
			\begin{array}{|c|c|c|c|c}
				\hline
				\raisebox{-1pt}{$ \bth $} & \raisebox{-1pt}{$ \bth $} & \raisebox{-1pt}{$ 1 $} & \raisebox{-1pt}{$ 2 $} & \multicolumn{1}{|c|}{\raisebox{-1pt}{3}} \\
				\hline
				\raisebox{-1pt}{$ \btw $} & \raisebox{-1pt}{$ \btw $} & \raisebox{-1pt}{$ 1 $} & \raisebox{-1pt}{$ 5 $}\\
				\cline{1-4}
			\end{array}\,.
		\end{align*}
	\end{ex}
	
	Let $ v_{b_1}\otimes v_{b_2} \otimes \cdots \otimes v_{b_N}\in\fundrep^{\otimes N} $
	be a crystal base element with $ b_1,\ldots, b_N\in B $.
	The SSYT associated with this crystal base element is the insertion
	\[
	b_N\rightarrow(\cdots \rightarrow(b_3\rightarrow(b_2\rightarrow\inlinetab{b_1}))\cdots)
	\]
	which, for brevity, we will denote $ b_2 \cdots b_{N}\rightarrow \inlinetab{b_1} $.
	
	\begin{ex}\label{ex:crystalSSYT}
		For $ \fundrep $ the fundamental representation of $ U_q(\gl(3|5)) $,
		the crystal base element 
		$ v_3\otimes v_5 \otimes v_1 \otimes v_{\bth} \otimes v_2 \otimes v_{\bth} \otimes v_1 \otimes v_{\btw} \otimes v_{\btw} \in \fundrep^{\otimes 9}$
		is mapped to the tableau
		\begin{align*}
			5 1\bth 2\bth 1 \btw \btw \rightarrow 
			\begin{array}{|c|}
				\hline
				3\\
				\hline
			\end{array}
			=
			1\bth 2\bth 1 \btw \btw \rightarrow 
			\begin{array}{|c|}
				\hline
				3\\
				\hline
				5\\
				\hline
			\end{array}
			=
			\bth 2\bth 1 \btw \btw \rightarrow 
			\begin{array}{|c|c|}
				\hline
				1 & 3\\
				\hline
				5\\
				\cline{1-1}
			\end{array}
			=
			2\bth 1 \btw \btw \rightarrow 
			\begin{array}{|c|c|c|}
				\hline
				\bth & 1 & 3\\
				\hline
				5\\
				\cline{1-1}
			\end{array}
			=
			\bth 1 \btw \btw \rightarrow 
			\begin{array}{|c|c|c|}
				\hline
				\bth & 1 & 3\\
				\hline
				2 & 5\\
				\cline{1-2}
			\end{array}\\
			=
			1 \btw \btw \rightarrow 
			\begin{array}{|c|c|c|c|}
				\hline
				\bth & \bth & 1 & 3\\
				\hline
				2 & 5\\
				\cline{1-2}
			\end{array}
			=
			\btw \btw \rightarrow 
			\begin{array}{|c|c|c|c|}
				\hline
				\bth & \bth & 1 & 3\\
				\hline
				1 & 2 & 5\\
				\cline{1-3}
			\end{array}
			=
			\btw \rightarrow 
			\begin{array}{|c|c|c|c|}
				\hline
				\bth & \bth & 1 & 3\\
				\hline
				\btw & 1 & 2 & 5\\
				\hline
			\end{array}
			=
			\begin{array}{|c|c|c|c|c}
				\hline
				\raisebox{-1pt}{$ \bth $} & \raisebox{-1pt}{$ \bth $} & \raisebox{-1pt}{$ 1 $} & \raisebox{-1pt}{$ 2 $} & \multicolumn{1}{|c|}{\raisebox{-1pt}{3}} \\
				\hline
				\raisebox{-1pt}{$ \btw $} & \raisebox{-1pt}{$ \btw $} & \raisebox{-1pt}{$ 1 $} & \raisebox{-1pt}{$ 5 $}\\
				\cline{1-4}
			\end{array}.	
		\end{align*}
	\end{ex}
	
	Note that the map from crystal base elements to SSYT is not injective
	(for example, 
	$ v_3\otimes v_2\otimes v_5 \otimes v_1\otimes v_1 \otimes v_{\bth}\otimes v_{\btw} \otimes v_{\bth} \otimes v_{\btw} $ is mapped to the same tableau in Example~\ref{ex:crystalSSYT}).
	However, this map sends crystal base elements of isomorphic irreducible subrepresentations 
	to the same set of SSYT
	(in particular, this map gives a bijection between irreducible subrepresentations (up to isomorphism) and Young diagrams).
	
	Note also that it is possible to construct a bijection between crystal base elements and ordered pairs of tableaux using the RSK algorithm (which makes use of Schensted's bumping algorithm)~\cite{BR87}.

	Using SSYT allows us to give combinatorial descriptions of $ U_q(\gl(m|n)) $ crystals.
	We will now restrict our attention to rectangular tableaux,
	but much of the discussion in this section applies more generally.
	
	Let $ B^{r,s} $ be the set of rectangular SSYT with height $ r $ and width $ s $.
	Take, an arbitrary tableau,
	\[
	\tab=
	\begin{array}{|c|c|c|c|}
		\hline
		t_{11} & t_{12} & \cdots  & t_{1s}\\
		\hline
		t_{21} & t_{22} & \cdots  & t_{2s}\\
		\hline
		\vdots  & \vdots  & \ddots  & \vdots \\
		\hline
		t_{r1} & t_{r2} & \cdots  & t_{rs}\\
		\hline
	\end{array}
	\in B^{r,s}.
	\]
	
	We define a function, $ \col $ by reading the tableau from top-to-bottom, right-to-left;
	explicitly,
	\[
	\col(\tab) = \underbrace{t_{1s}\ldots t_{rs}}_{t_{*s}} \cdots \underbrace{t_{12}\ldots t_{r2}}_{t_{*2}}
	\underbrace{t_{11}\ldots t_{r1}}_{t_{*1}}.
	\]
	
	Moreover, for $ \tab_1,\tab_2\in B^{r,s} $,
	we define $ \col(\tab_1\otimes \tab_2)=\col(\tab_1)\col(\tab_2) $.
	
	For $ i\in \Ieven $,
	the action of the crystal operators $ e_i $ and $ f_i $ can be computed 
	by a \defn{signature rule} similar to that for $ U_q'(\asl_n) $-crystals~\cite{Yamada04}.
	
	\begin{dfn}
		For some positive integer $ d $,
		let $ \tab\in (B^{r,s})^{\otimes d} $ and let $ i\in \Ieven $.
		If $ i=\bk\in I_- $, we denote $ i+1 = \overline{k+1} $.
		We define the \defn{$ i $-signature}, denoted $ \sg_i(\tab) $,
		to be the sequence of $ + $ and $ - $
		obtained by deleting all letters in $ \col(\tab) $ that are not $ i $ or $ i+1 $,
		and then replacing all $ i $ with a $ - $ symbol
		and replacing all $ i+1 $ with a $ + $ symbol.
		
		We define the \defn{reduced $ i $-signature}, denoted $ \rsg_i(\tab) $,
		to be equal to the $ i $-signature,
		except with $ +- $ pairs (in that order) successively deleted,
		so that $ \rsg_i(\tab) $ is of the form
		\[
		\underbrace{-\cdots -}_{a}\underbrace{+\cdots +}_{b}
		\]
		(where $ a $ or $ b $ can be zero).
	\end{dfn}
	
	For a tableau $ \tab\in B^{r,s} $ and for $ i\in \Ieven $ where $ i\in I_+ $
	(resp. $ i\in I_- $):
	\begin{itemize}
		\item To evaluate $ f_i(\tab) $ (resp. $ e_i(\tab) $),
		find the rightmost $ - $ symbol in $ \rsg_i(\tab) $ and change the corresponding
		$ \begin{array}{|c|}
			\hline \raisebox{-1pt}{$ i $}\\\hline
		\end{array} $
		in $ \tab $ to 
		$ \begin{array}{|c|}
			\hline \raisebox{-1pt}{$ i+1 $} \\\hline
		\end{array} $.
		If there are no $ - $ symbols, then $ f_i(\tab)=0 $ (resp. $ e_i(\tab)=0 $).
		\item To evaluate $ e_i(\tab) $ (resp. $ f_i(\tab) $),
		find the leftmost $ + $ symbol in $ \rsg_i(\tab) $ and change the corresponding
		$ \begin{array}{|c|}
			\hline \raisebox{-1pt}{$ i+1 $} \\\hline
		\end{array} $
		in $ \tab $ to 
		$ \begin{array}{|c|}
			\hline \raisebox{-1pt}{$ i $} \\\hline
		\end{array} $.
		If there are no $ + $ symbols, then $ e_i(\tab)=0 $ (resp. $ f_i(\tab)=0 $).
	\end{itemize}
	
	The $ f_0 $ and $ e_0 $ operators have a different algorithm:
	\begin{itemize}
		\item If the first occurrence of $ \bon $ in $ \col(\tab) $ is before the first occurrence of $ 1 $,
		then $ e_0(\tab)=0 $ and $ f_0(\tab) $ replaces the corresponding 
		$\begin{array}{|c|}
			\hline \raisebox{-1pt}{$ \bon $} \\\hline
		\end{array}$
		in $ \tab $ with 
		$\begin{array}{|c|}
			\hline \raisebox{-1pt}{$ 1 $} \\\hline
		\end{array}$.
		
		\item If the first occurrence of $ 1 $ in $ \col(\tab) $ is before the first occurrence of $ \bon $,
		then $ f_0(\tab)=0 $ and $ e_0(\tab) $ replaces the corresponding
		$\begin{array}{|c|}
			\hline \raisebox{-1pt}{$ 1 $} \\\hline
		\end{array}$
		in $ \tab $ with
		$\begin{array}{|c|}
			\hline \raisebox{-1pt}{$ \bon $} \\\hline
		\end{array}$.
	\end{itemize}
	
	
	\begin{ex}\label{ex:operator}
		We will compute $ e_{\bth}(\tab) $ for
		\[
		\tab=
		\begin{array}{|c|c|c|}
			\hline
			\raisebox{-1pt}{$ \bfo $} & \raisebox{-1pt}{$ \bth $} & \raisebox{-1pt}{$ \bth $}\\
			\hline
			\raisebox{-1pt}{$ \bth $} & \raisebox{-1pt}{$ 1 $} & \raisebox{-1pt}{$ 3 $} \\
			\hline
			\raisebox{-1pt}{$ 1 $} & \raisebox{-1pt}{$ 2 $} & \raisebox{-1pt}{$ 3 $}\\
			\hline
		\end{array}
		.
		\]
		We have that
		\[
		\begin{array}{rcccccccccc}
			\col(\tab)\phantom{)} & = & \bth & 3 & 3 & \bth & 1 & 2 & \bfo & \bth & 1\\
			\sg_{\bth}(\tab) & = & - &&& - &&& + & - &\\
			\rsg_{\bth}(\tab) & = & - &&& - &&&&&
		\end{array}
		.
		\]
		
		The rightmost $ - $ corresponds to the bolded number below, 
		\[
		\begin{array}{rcccccccccc}
			\col(\tab)\phantom{)} & = & \bth & 3 & 3 & {\bf\bth} & 1 & 2 & \bfo & \bth & 1\\
			\rsg_{\bth}(\tab) & = & - &&& - &&&&&
		\end{array}
		\quad \rightsquigarrow  \quad
		\begin{array}{|c|c|c|}
			\hline
			\raisebox{-1pt}{$ \bfo $} & \raisebox{-1pt}{$ {\bf\bth} $} & \raisebox{-1pt}{$ \bth $} \\
			\hline
			\raisebox{-1pt}{$ \bth $} & \raisebox{-1pt}{$ 1 $} & \raisebox{-1pt}{$ 3 $} \\
			\hline
			\raisebox{-1pt}{$ 1 $} & \raisebox{-1pt}{$ 2 $} & \raisebox{-1pt}{$ 3 $}\\
			\hline
		\end{array}
		\]
		so we replace this 
		$ \begin{array}{|c|}
			\hline \raisebox{-1pt}{$ \bth $} \\\hline
		\end{array} $
		with 
		$ \begin{array}{|c|}
			\hline \raisebox{-1pt}{$ \bfo $} \\\hline
		\end{array} $
		to get
		\[
		e_{\bth}(\tab) = 
		\begin{array}{|c|c|c|}
			\hline
			\raisebox{-1pt}{$ \bfo $} & \raisebox{-1pt}{$ \bfo $} & \raisebox{-1pt}{$ \bth $}\\
			\hline
			\raisebox{-1pt}{$ \bth $} & \raisebox{-1pt}{$ 1 $} & \raisebox{-1pt}{$ 3 $} \\
			\hline
			\raisebox{-1pt}{$ 1 $} & \raisebox{-1pt}{$ 2 $} & \raisebox{-1pt}{$ 3 $} \\
			\hline
		\end{array}
		.
		\]
	\end{ex}
	
	We can also use SSYT to describe the weights (in the representation theoretic sense)
	of the crystal elements.
	Weights are linear combinations in the set 
	$ \bigoplus_{b\in B}\ZZ\varepsilon_b $ 
	(for our purposes, $ \varepsilon_b $ can be treated as formal symbols).
	In the weight of a SSYT $ \tab $,
	the coefficient corresponding to $ \varepsilon_b $ 
	is equal to the number of appearances of $ b $ in $ \tab $~\cite{BKK00}.
	
	\begin{ex}
		Let $ \tab $ be as in Example~\ref{ex:operator}.
		Then, the weight of $ \tab $ is 
		$ \varepsilon_{\bfo}+3\varepsilon_{\bth}+2\varepsilon_1+\varepsilon_2+2\varepsilon_3 $.
	\end{ex}
	We define arbitrary weights $\mu$ and $\nu$ as follows:
	\begin{align*}
		\mu &= \mu_1\varepsilon_{\bm}+\cdots+\mu_m\varepsilon_{\bon}+\mu_{m+1}\varepsilon_1+\cdots+\mu_{m+n}\varepsilon_n\,, \\
		\nu &= \nu_1\varepsilon_{\bm}+\cdots+\nu_m\varepsilon_{\bon}+\nu_{m+1}\varepsilon_1+\cdots+\nu_{m+n}\varepsilon_n \,.
	\end{align*}
	We can define a partial ordering on the set of weights
	by saying 
	$ \mu\ge\nu $ if the following hold:
	\begin{align*}
		\mu_1+\cdots+\mu_{m+n}&=\nu_1+\cdots+\nu_{m+n} \\
		\mu_1+\cdots+\mu_j&\ge\nu_1+\cdots+\nu_j && \text{for all}~j=1,\ldots,m+n.
	\end{align*}
	Note that the operators $ e_i $ ($ i\in I\setminus\{\bze\} $)
	raise the weight and the operators $ f_i $ ($ i\in I\setminus\{\bze\} $) lower the weight. We say that $ \tab $ is a \defn{highest weight element} if $ e_i \tab = 0 $ for all $ i\in I\setminus\{\bze\} $.
	
	\begin{dfn}
		A crystal element $ \tab $ with weight $ \lambda $ is a \defn{genuine highest weight element}
		if
		\begin{enumerate}[label=(\roman*)]
			\item given any other crystal element with some weight $ \mu $,
			the expression $ \lambda-\mu $ has only positive coefficients; and
			\item no other crystal element has weight $ \lambda $.
		\end{enumerate}
	\end{dfn}
	Every genuine highest weight element is a highest weight element,
	but not every highest weight element is a genuine highest weight element~\cite{BKK00}.
	
	For crystals whose elements are the SSYT of the same shape,
	the genuine highest weight element exists and is unique~\cite{BKK00}.
	Each connected component of $ B^{r_1,s_1}\otimes B^{r_2,s_2} $
	is isomorphic to such a crystal
	(this isomorphism is given by $ (\tab_1\otimes \tab_2) \mapsto (\col(\tab_2)\rightarrow \tab_1) $
	for $ \tab_1\otimes \tab_2 $ in the connected component of interest).
	Thus, each connected component of $ B^{r_1,s_1}\otimes B^{r_1,s_1} $ 
	has a unique genuine highest weight element.
	This property is of great utility in the proofs of the main theorems.
	
	\subsection{Combinatorial \texorpdfstring{$ R $}{R}-matrix} \label{sec:comb_R}
	Consider two $U_q(\agl(m|n))$-crystals $B^{r_1,s_1}$ and $B^{r_2,s_2}$. Then there exists~\cite{KO18} a unique isomorphism that commutes with $ e_i $ and $ f_i $ (for all $ i\in I $) called the \defn{combinatorial R-matrix}:
	\[
	R\colon B^{r_1,s_1} \otimes B^{r_2,s_2} \to B^{r_2,s_2} \otimes B^{r_1,s_1}.
	\]
	
	We can describe the action of the combinatorial $ R $-matrix
	using Schensted's bumping algorithm.
	\begin{thm}[{\cite[Theorem~7.9]{KO18}}]
		The combinatorial $ R $-matrix maps $ \tab_1\otimes \tab_2 $ to $ \widetilde{\tab}_2\otimes \widetilde{\tab}_1 $ 
		if and only if
		$
		\col(\tab_2)\rightarrow \tab_1 = \col(\widetilde{\tab}_1)\rightarrow\widetilde{\tab}_2
		$
	\end{thm}
	
	\begin{ex}\label{ex:Rmatrix}
		Set
		\[
		\tab_1 = 
		\begin{array}{|c|c|c|}
			\hline
			\raisebox{-1pt}{$ \bfo $} & \raisebox{-1pt}{$ \bfo $} & \raisebox{-1pt}{$ \bth $}\\
			\hline
			\raisebox{-1pt}{$ \bth $} & \raisebox{-1pt}{$ 1 $} & \raisebox{-1pt}{$ 3 $} \\
			\hline
			\raisebox{-1pt}{$ 1 $} & \raisebox{-1pt}{$ 2 $} & \raisebox{-1pt}{$ 3 $} \\
			\hline
		\end{array}
		\,, \quad \tab_2=
		\begin{array}{|c|}
			\hline
			\raisebox{-1pt}{$ \bth $} \\
			\hline
			\raisebox{-1pt}{$ 1 $} \\
			\hline
			\raisebox{-1pt}{$ 2 $} \\
			\hline
		\end{array}
		\,, \quad \widetilde{\tab}_2=
		\begin{array}{|c|}
			\hline
			\raisebox{-1pt}{$ \bth $} \\
			\hline
			\raisebox{-1pt}{$ 1 $} \\
			\hline
			\raisebox{-1pt}{$ 3 $} \\
			\hline
		\end{array}
		\,,\quad \widetilde{\tab}_1=
		\begin{array}{|c|c|c|}
			\hline
			\raisebox{-1pt}{$ \bfo $} & \raisebox{-1pt}{$ \bfo $} & \raisebox{-1pt}{$ 1 $}\\
			\hline
			\raisebox{-1pt}{$ \bth $} & \raisebox{-1pt}{$ \bth $} & \raisebox{-1pt}{$ 2 $} \\
			\hline
			\raisebox{-1pt}{$ 1 $} & \raisebox{-1pt}{$ 2 $} & \raisebox{-1pt}{$ 3 $} \\
			\hline
		\end{array}
		\,.
		\]
		Then, $ R(\tab_1\otimes \tab_2)=\widetilde{\tab}_2\otimes \widetilde{\tab}_1 $.
		Indeed, let us first compute 
		\[ 
		\col(\tab_2)\rightarrow \tab_1
		=
		\bth12\rightarrow
		\begin{array}{ccc}
			\hline
			\multicolumn{1}{|c}{\raisebox{-1pt}{$ \bfo $}} 
			& \multicolumn{1}{|c}{\raisebox{-1pt}{$ \bfo $}} 
			& \multicolumn{1}{|c|}{\raisebox{-1pt}{$ \bth $}}\\
			\hline
			\multicolumn{1}{|c}{\raisebox{-1pt}{$ \bth $}} 
			& \multicolumn{1}{|c}{\raisebox{-1pt}{$ 1 $}} 
			& \multicolumn{1}{|c|}{\raisebox{-1pt}{$ 3 $}} \\
			\hline
			\multicolumn{1}{|c}{\raisebox{-1pt}{$ 1 $}}
			& \multicolumn{1}{|c}{\raisebox{-1pt}{$ 2 $}}
			& \multicolumn{1}{|c|}{\raisebox{-1pt}{$ 3 $}} \\
			\hline
		\end{array}
		=
		12\rightarrow
		\begin{array}{cccc}
			\hline
			\multicolumn{1}{|c}{\raisebox{-1pt}{$ \bfo $}} 
			& \multicolumn{1}{|c}{\raisebox{-1pt}{$ \bfo $}} 
			& \multicolumn{1}{|c}{\raisebox{-1pt}{$ \bth $}}
			& \multicolumn{1}{|c|}{\raisebox{-1pt}{$ 3 $}}\\
			\hline
			\multicolumn{1}{|c}{\raisebox{-1pt}{$ \bth $}} 
			& \multicolumn{1}{|c}{\raisebox{-1pt}{$ \bth $}} 
			& \multicolumn{1}{|c|}{\raisebox{-1pt}{$ 1 $}} \\
			\cline{1-3}
			\multicolumn{1}{|c}{\raisebox{-1pt}{$ 1 $}}
			& \multicolumn{1}{|c}{\raisebox{-1pt}{$ 2 $}}
			& \multicolumn{1}{|c|}{\raisebox{-1pt}{$ 3 $}} \\
			\cline{1-3}
		\end{array}
		=
		2\rightarrow
		\begin{array}{cccc}
			\hline
			\multicolumn{1}{|c}{\raisebox{-1pt}{$ \bfo $}} 
			& \multicolumn{1}{|c}{\raisebox{-1pt}{$ \bfo $}} 
			& \multicolumn{1}{|c}{\raisebox{-1pt}{$ \bth $}}
			& \multicolumn{1}{|c|}{\raisebox{-1pt}{$ 3 $}}\\
			\hline
			\multicolumn{1}{|c}{\raisebox{-1pt}{$ \bth $}} 
			& \multicolumn{1}{|c}{\raisebox{-1pt}{$ \bth $}} 
			& \multicolumn{1}{|c|}{\raisebox{-1pt}{$ 1 $}} \\
			\cline{1-3}
			\multicolumn{1}{|c}{\raisebox{-1pt}{$ 1 $}}
			& \multicolumn{1}{|c}{\raisebox{-1pt}{$ 2 $}}
			& \multicolumn{1}{|c|}{\raisebox{-1pt}{$ 3 $}} \\
			\cline{1-3}
			\multicolumn{1}{|c|}{\raisebox{-1pt}{$ 1 $}}\\
			\cline{1-1}
		\end{array}
		=
		\begin{array}{cccc}
			\hline
			\multicolumn{1}{|c}{\raisebox{-1pt}{$ \bfo $}} 
			& \multicolumn{1}{|c}{\raisebox{-1pt}{$ \bfo $}} 
			& \multicolumn{1}{|c}{\raisebox{-1pt}{$ \bth $}}
			& \multicolumn{1}{|c|}{\raisebox{-1pt}{$ 3 $}}\\
			\hline
			\multicolumn{1}{|c}{\raisebox{-1pt}{$ \bth $}} 
			& \multicolumn{1}{|c}{\raisebox{-1pt}{$ \bth $}} 
			& \multicolumn{1}{|c|}{\raisebox{-1pt}{$ 1 $}} \\
			\cline{1-3}
			\multicolumn{1}{|c}{\raisebox{-1pt}{$ 1 $}}
			& \multicolumn{1}{|c}{\raisebox{-1pt}{$ 2 $}}
			& \multicolumn{1}{|c|}{\raisebox{-1pt}{$ 3 $}} \\
			\cline{1-3}
			\multicolumn{1}{|c|}{\raisebox{-1pt}{$ 1 $}}\\
			\cline{1-1}
			\multicolumn{1}{|c|}{\raisebox{-1pt}{$ 2 $}}\\
			\cline{1-1}
		\end{array}
		.
		\]
		We similarly find that
		\[
		\col(\widetilde{\tab}_1)\rightarrow \widetilde{\tab}_2
		=
		\begin{array}{cccc}
			\hline
			\multicolumn{1}{|c}{\raisebox{-1pt}{$ \bfo $}} 
			& \multicolumn{1}{|c}{\raisebox{-1pt}{$ \bfo $}} 
			& \multicolumn{1}{|c}{\raisebox{-1pt}{$ \bth $}}
			& \multicolumn{1}{|c|}{\raisebox{-1pt}{$ 3 $}}\\
			\hline
			\multicolumn{1}{|c}{\raisebox{-1pt}{$ \bth $}} 
			& \multicolumn{1}{|c}{\raisebox{-1pt}{$ \bth $}} 
			& \multicolumn{1}{|c|}{\raisebox{-1pt}{$ 1 $}} \\
			\cline{1-3}
			\multicolumn{1}{|c}{\raisebox{-1pt}{$ 1 $}}
			& \multicolumn{1}{|c}{\raisebox{-1pt}{$ 2 $}}
			& \multicolumn{1}{|c|}{\raisebox{-1pt}{$ 3 $}} \\
			\cline{1-3}
			\multicolumn{1}{|c|}{\raisebox{-1pt}{$ 1 $}}\\
			\cline{1-1}
			\multicolumn{1}{|c|}{\raisebox{-1pt}{$ 2 $}}\\
			\cline{1-1}
		\end{array}
		.
		\]
	\end{ex}
	
	There is a more explicit method of determining the $R$-matrix of two tableaux. This involves an inversion of the modified bumping algorithm outlined previously.
	Let $\tab_1 \in B^{r_1,s_1}$ and $\tab_2 \in B^{r_2,s_2}$.
	Then $R(\tab_1 \otimes \tab_2)$ is determined by the following process. 
	Begin with $ P = \col(\tab_2) \rightarrow \tab_1$.
	Let $\widehat{Q}$ be a rectangular reverse semi-standard tableau of height $ r_1 $ and width $ s_1 $. We construct $\widehat{Q}$ using the weight vector given by $\mu = (\widetilde{d}_i - d_i, \cdots, \widetilde{d}_s - d_s)$ where $\widetilde{d_i}$ and $d_i$ are the heights of the $i$th column of P and $\tab_2$ respectively, and $s$ is the width of $P$. Then $\widehat{Q}$ is given by the unique reverse conjugate semi-standard tableau of shape $\tab_1$ and weight $\mu$.
	
	We then perform a reverse insertion on P by reading $\widehat{Q}$ bottom-to-top, left-to-right. Each element in the reading of $\widehat{Q}$ gives the next column on which we perform the bumping algorithm in reverse.
	The elements removed from $P$ are placed bottom-to-top and left-to-right into a rectangular tableau of height $ r_1 $ and width $ s_1 $.
	This tableau is $\widetilde{\tab}_1$.
	Continuing until $\widehat{Q}$ is empty, we obtain the resultant tableaux as $P = \widetilde{\tab}_2$ and $\widetilde{\tab}_1$.
	\begin{ex}\label{ex:1explicit_R}
		This example will demonstrate the explicit $R$-matrix computation for $\tab_1$ and $\tab_2$ defined as follows
		\[
		\tab_1 = 
		\begin{array}{|c|c|c|}
			\hline
			\raisebox{-1pt}{\( \bfo \)} & \raisebox{-1pt}{\( \bfo \)} & \raisebox{-1pt}{\( \bth \)}\\
			\hline
			\raisebox{-1pt}{\( \bth \)} & \raisebox{-1pt}{\( 1 \)} & \raisebox{-1pt}{\( 3 \)} \\
			\hline
			\raisebox{-1pt}{\( 1 \)} & \raisebox{-1pt}{\( 2 \)} & \raisebox{-1pt}{\( 3 \)} \\
			\hline
		\end{array}
		, \quad \tab_2=
		\begin{array}{|c|}
			\hline
			\raisebox{-1pt}{\( \bth \)} \\
			\hline
			\raisebox{-1pt}{\( 1 \)} \\
			\hline
			\raisebox{-1pt}{\( 2 \)} \\
			\hline
		\end{array}\,.
		\]
		From Example~\ref{ex:Rmatrix} we find that
		\[
		P=	\col(\tab_2) \rightarrow \tab_1
		=
		\begin{array}{cccc}
			\hline
			\multicolumn{1}{|c}{\raisebox{-1pt}{\( \bfo \)}} 
			& \multicolumn{1}{|c}{\raisebox{-1pt}{\( \bfo \)}} 
			& \multicolumn{1}{|c}{\raisebox{-1pt}{\( \bth \)}}
			& \multicolumn{1}{|c|}{\raisebox{-1pt}{\( 3 \)}}\\
			\hline
			\multicolumn{1}{|c}{\raisebox{-1pt}{\( \bth \)}} 
			& \multicolumn{1}{|c}{\raisebox{-1pt}{\( \bth \)}} 
			& \multicolumn{1}{|c|}{\raisebox{-1pt}{\( 1 \)}} \\
			\cline{1-3}
			\multicolumn{1}{|c}{\raisebox{-1pt}{\( 1 \)}}
			& \multicolumn{1}{|c}{\raisebox{-1pt}{\( 2 \)}}
			& \multicolumn{1}{|c|}{\raisebox{-1pt}{\( 3 \)}} \\
			\cline{1-3}
			\multicolumn{1}{|c|}{\raisebox{-1pt}{\( 1 \)}}\\
			\cline{1-1}
			\multicolumn{1}{|c|}{\raisebox{-1pt}{\( 2 \)}}\\
			\cline{1-1}
		\end{array} \quad .
		\]
		The weight vector associated with $\widehat{Q}$ is then given by $(2,3,3,1)$. There exists a unique reverse conjugate semi-standard tableau of shape $\tab_1$ and weight $(2,3,3,1)$ given as follows,
		\[
		\widehat{Q} = 
		\begin{array}{ccc}
			\hline
			\multicolumn{1}{|c|}{\raisebox{-1pt}{\( 4 \)}} &\multicolumn{1}{|c|}{\raisebox{-1pt}{\( 3 \)}}&\multicolumn{1}{|c|}{\raisebox{-1pt}{\( 2 \)}}  \\
			\hline
			\multicolumn{1}{|c|}{\raisebox{-1pt}{\( 3 \)}}&\multicolumn{1}{|c|}{\raisebox{-1pt}{\( 2 \)}}&\multicolumn{1}{|c|}{\raisebox{-1pt}{\( 1 \)}} \\
			\hline
			\multicolumn{1}{|c|}{\raisebox{-1pt}{\( 3 \)}}&\multicolumn{1}{|c|}{\raisebox{-1pt}{\( 2 \)}}&\multicolumn{1}{|c|}{\raisebox{-1pt}{\( 1 \)}}\\
			\hline
		\end{array}
		\,.
		\]
		Now reading $\widehat{Q}$ from bottom-to-top, left-to-right.
		We first read the element $3$, beginning in column $3$ of $P$ we pop the final row element $3$, then performing reverse insertion from column $2$, this $3$ switches with the $2$ which then subsequently switches with $1$ in column $1$.
		The $1$ left over then begins filling a tableau $\widetilde{\tab}_1$ bottom-to-top, left-to-right.
		We are then left with the following $P$ and $\widehat{Q}$, and $\widetilde{\tab}_1$:
		\[
		P = 
		\begin{array}{cccc}
			\hline
			\multicolumn{1}{|c}{\raisebox{-1pt}{\( \bfo \)}} 
			& \multicolumn{1}{|c}{\raisebox{-1pt}{\( \bfo \)}} 
			& \multicolumn{1}{|c}{\raisebox{-1pt}{\( \bth \)}}
			& \multicolumn{1}{|c|}{\raisebox{-1pt}{\( 3 \)}}\\
			\hline
			\multicolumn{1}{|c}{\raisebox{-1pt}{\( \bth \)}} 
			& \multicolumn{1}{|c}{\raisebox{-1pt}{\( \bth \)}} 
			& \multicolumn{1}{|c|}{\raisebox{-1pt}{\( 1 \)}} \\
			\cline{1-3}
			\multicolumn{1}{|c}{\raisebox{-1pt}{\( 1 \)}}
			& \multicolumn{1}{|c}{\raisebox{-1pt}{\( 3 \)}}
			& \multicolumn{1}{|c}{\raisebox{-1pt}{\(  \)}} \\
			\cline{1-2}
			\multicolumn{1}{|c|}{\raisebox{-1pt}{\( 2 \)}}\\
			\cline{1-1}
			\multicolumn{1}{|c|}{\raisebox{-1pt}{\( 2 \)}}\\
			\cline{1-1}
		\end{array}
		, \quad \widehat{Q} =
		\begin{array}{ccc}
			\hline
			\multicolumn{1}{|c|}{\raisebox{-1pt}{\( 4 \)}} &\multicolumn{1}{|c|}{\raisebox{-1pt}{\( 3 \)}}&\multicolumn{1}{|c|}{\raisebox{-1pt}{\( 2 \)}}  \\
			\hline
			\multicolumn{1}{|c|}{\raisebox{-1pt}{\( 3 \)}}&\multicolumn{1}{|c|}{\raisebox{-1pt}{\( 2 \)}}&\multicolumn{1}{|c|}{\raisebox{-1pt}{\( 1 \)}} \\
			\hline
			&\multicolumn{1}{|c|}{\raisebox{-1pt}{\( 2 \)}}&\multicolumn{1}{|c|}{\raisebox{-1pt}{\( 1 \)}}\\
			\cline{2-3}
		\end{array}
		, \quad \widetilde{\tab}_1 = \begin{array}{ccc}
			\hline
			\multicolumn{1}{|c|}{\raisebox{-1pt}{\(  \)}}& \multicolumn{1}{|c|}{\raisebox{-1pt}{\(  \)}}& \multicolumn{1}{c|}{\raisebox{-1pt}{\(  \)}}  \\
			\hline
			\multicolumn{1}{|c|}{\raisebox{-1pt}{\(  \)}}& \multicolumn{1}{|c|}{\raisebox{-1pt}{\(  \)}}& \multicolumn{1}{|c|}{\raisebox{-1pt}{\(  \)}}  \\
			\cline{1-3}
			\multicolumn{1}{|c|}{\raisebox{-1pt}{\( 1 \)}}& \multicolumn{1}{|c|}{\raisebox{-1pt}{\( \hphantom{3} \)}}& \multicolumn{1}{|c|}{\raisebox{-1pt}{\( \hphantom{3} \)}}  \\
			\hline
		\end{array}\,.
		\]
		Continuing this process until $\widehat{Q}$ is empty we obtain the following tableaux
		\[
		\widetilde{\tab}_2=P = 
		\begin{array}{|c|}
			\hline
			\raisebox{-1pt}{$ \bth $} \\
			\hline
			\raisebox{-1pt}{$ 1 $} \\
			\hline
			\raisebox{-1pt}{$ 3 $} \\
			\hline
		\end{array}
		\,,\quad \widetilde{\tab}_1=
		\begin{array}{|c|c|c|}
			\hline
			\raisebox{-1pt}{$ \bfo $} & \raisebox{-1pt}{$ \bfo $} & \raisebox{-1pt}{$ 1 $}\\
			\hline
			\raisebox{-1pt}{$ \bth $} & \raisebox{-1pt}{$ \bth $} & \raisebox{-1pt}{$ 2 $} \\
			\hline
			\raisebox{-1pt}{$ 1 $} & \raisebox{-1pt}{$ 2 $} & \raisebox{-1pt}{$ 3 $} \\
			\hline
		\end{array}
		\,.
		\]
		These are indeed the resultant tableaux satisfying the $R$-matrix as shown in Example~\ref{ex:Rmatrix}.
	\end{ex}
	
	\subsection{Energy function}
	
	\begin{dfn}
		We call a function $ H\colon B^{r_1,s_1}\otimes B^{r_2,s_2} \to \ZZ $ 
		an \defn{energy function} if, for all $ b=x\otimes y \in B^{r_1,s_1}\otimes B^{r_2,s_2} $,
		we have $ H(f_i b)=H(b) $ and $ H(e_i b)=H(b) $ 
		for $ i\in I\setminus\{\bze\} $, and if  $ \widetilde{y}\otimes\widetilde{x} = R(b) $ then
		\[
		H(e_{\bze}b) = H(b) + 
		\begin{cases}
			1 & \text{if}~e_{\bze}b = (e_{\bze}x)\otimes y~
			\text{and}~e_{\bze}R(b)=(e_{\bze}\widetilde{y})\otimes\widetilde{x},\\
			-1 & \text{if}~e_{\bze}b = x\otimes (e_{\bze}y)~
			\text{and}~e_{\bze}R(b)=\widetilde{y}\otimes(e_{\bze}\widetilde{x}),\\
			0 & \text{otherwise.}
		\end{cases}
		\]
	\end{dfn}
	
	The energy function exists and is unique up to additive constant~\cite{KO18}.
	Moreover, we can compute the energy function using the bumping algorithm.
	\begin{prop}[{\cite[Theorem~7.9]{KO18}}]\label{prop:energyboxes}
		Up to additive constant,
		$ H(x\otimes y) $ 
		is given by the number of boxes in $ \col(y)\rightarrow x $
		that are strictly to the right of the $ \max(s_1,s_2) $-th column.
	\end{prop}
	
	By convention, we will choose the additive constant so that the maximum value of $ H $ is zero.
	Explicitly, if $ \widetilde{H}(x\otimes y) $ is given by the number of boxes 
	as in Proposition~\ref{prop:energyboxes}, with additive constant equal to $ 0 $,
	then we define
	$
	H(x\otimes y) = \widetilde{H}(x\otimes y) - \min(r_1,r_2)\min(s_1,s_2).
	$
	
	\begin{ex}
		Set $ x $ and $ y $ as in Example~\ref{ex:Rmatrix}.
		We know that
		\[
		\col(y)\rightarrow x
		=
		\begin{array}{cccc}
			\hline
			\multicolumn{1}{|c}{\raisebox{-1pt}{$ \bfo $}} 
			& \multicolumn{1}{|c}{\raisebox{-1pt}{$ \bfo $}} 
			& \multicolumn{1}{|c}{\raisebox{-1pt}{$ \bth $}}
			& \multicolumn{1}{|c|}{\raisebox{-1pt}{$ 3 $}}\\
			\hline
			\multicolumn{1}{|c}{\raisebox{-1pt}{$ \bth $}} 
			& \multicolumn{1}{|c}{\raisebox{-1pt}{$ \bth $}} 
			& \multicolumn{1}{|c|}{\raisebox{-1pt}{$ 1 $}} \\
			\cline{1-3}
			\multicolumn{1}{|c}{\raisebox{-1pt}{$ 1 $}}
			& \multicolumn{1}{|c}{\raisebox{-1pt}{$ 2 $}}
			& \multicolumn{1}{|c|}{\raisebox{-1pt}{$ 3 $}} \\
			\cline{1-3}
			\multicolumn{1}{|c|}{\raisebox{-1pt}{$ 1 $}}\\
			\cline{1-1}
			\multicolumn{1}{|c|}{\raisebox{-1pt}{$ 2 $}}\\
			\cline{1-1}
		\end{array}.
		\]
		We have that $ \max(s_1,s_2)=\max(3,1)=3 $,
		and the number of boxes to the right of the third column is $ 1 $.
		So, $ H(x\otimes y) = 1 - \min(r_1,r_2)\min(s_1,s_2) = -2 $.
	\end{ex}
	
	\section{Super Box-Ball System} \label{sec:SBBS}
	\subsection{Box-ball system definition} \label{sec:SBBSdef}
	
	A box-ball system possesses a \defn{vacuum element} representing the absence of a ball.
	We require that the combinatorial $ R $-matrix acts as an identity on
	the vacuum element; that is,
	if $ u $ is the vacuum element then $R(u \otimes u) = u\otimes u$.
	
	We define the vacuum element to be the genuine highest weight element of $ B^{r,1} $
	as a finite-type $ U_q(\gl(m|n)) $-crystal where $ r\leq m $ (see~\cite{BKK00}).
	Such a vacuum element will have the desired property.
	More generally, the genuine highest weight element for $ B^{r,s} $ has the form
	\[
	u_s = 
	\underbrace{
		\begin{array}{|c|c|c|}
			\hline
			\raisebox{-1pt}{$ \bm $} & \cdots  & \raisebox{-1pt}{$ \bm $}  \\
			\hline
			\raisebox{-1pt}{$ \overline{m-1} $} & \cdots  & \raisebox{-1pt}{$ \overline{m-1} $} \\
			\hline
			\vdots  & \ddots  & \vdots \\
			\hline
			\raisebox{-1pt}{$ \overline{m-r+1} $} & \cdots  & \raisebox{-1pt}{$ \overline{m-r+1} $} \\
			\hline
		\end{array}
	}_{s} \;.
	\]
	The vacuum element is then denoted by $ u_1 $.
	
	We can think of the elements of $ B^{r,1} \setminus \{u_1\} $ as representing 
	different balls in the system.
	Within the super box-ball system, a \defn{state}
	consists of $B^{r,1}$ elements in a one dimensional lattice with only finitely many non-vacuum elements.
	More precisely, a state
	is of the form
	\[
	b_0 \otimes  b_1 \otimes \cdots \otimes b_K \otimes  (u_1)^{\otimes \infty }\in (B^{r,1})^{\otimes \infty }
	\]
	where $ b_j \in B^{r,1} $ can be any element (including $ u_1 $).
	
	The state evolves in time via a \defn{carrier} element
	which `picks up' and `puts down' $ B^{r,1} $ elements.
	The carrier is an element of $ B^{r,\ell} $
	which changes based on its location in the state.
	It is initialised as the genuine highest weight element $ u_{\ell} $.
	The action of moving the carrier through the state is performed by the combinatorial $ R $-matrix.
	In particular, this is performed by functions $ R_a $ where
	\[
	R_a  = \underbrace{\id \otimes  \cdots  \otimes \id}_a\otimes R\otimes \id\otimes \id\otimes \cdots .
	\]
	
	We can then define the \defn{time evolution operator}, $ T_{\ell} $, by
	\[
	T_{\ell}(b)\otimes u_{\ell}=\cdots R_3R_2R_1R_0(u_{\ell}\otimes b)
	\]
	for any state $ b $.
	This is well-defined because there are finitely many non-vacuum elements in the state,
	so we eventually have $ R(u_{\ell}\otimes u_1)=u_1\otimes u_{\ell} $.
	The time evolution operator computes the state for the next time step.
	For convenience, we will write $ T_{\infty}=\lim_{\ell\to\infty}T_{\ell} $.
	
	\begin{prop}\label{prop:Tinfwelldefined}
		$ T_{\infty} $ is well-defined
	\end{prop}
	\begin{proof}[Proof sketch]
		A simple insertion argument shows that
		the action of the $ R $-matrices in the definition of $ T_{K} $ and $ T_{\ell} $ are the same
		(under the inclusion $ B^{r,K} \hookrightarrow B^{r,\ell},\, 
		\inlinetab{v}
		\mapsto 
		\begin{array}{|c|c|}
			\hline
			u_{\ell - K} & v \\
			\hline
		\end{array}
		$).
		This shows that $ T_{\ell}=T_{\ell'} $ for $ \ell,\ell'\ge K $
		and hence $ T_{\infty} $ is well-defined.
	\end{proof}
	
	Pictorially,
	we can represent the computation of the time evolution 
	$ T_\ell (b_1\otimes\cdots\otimes b_K\otimes (u_1)^{\otimes\infty}) 
	= \bigotimes_{j=1}^\infty \widetilde{b}_j $
	as follows:
	\[
	\begin{tikzpicture}
		\node (b1) at (1,1) {$b_1$};
		\node (b2) at (3,1) {$b_2$};
		\node (bK) at (7,1) {$b_K$};
		\node (bK1) at (9,1) {$u_1$};
		\node (bK2) at (11,1) {$u_1$};
		\node (bK3) at (13,1) {$u_1$};
		
		\node (bt1) at (1,-1.1) {$\widetilde{b}_1$};
		\node (bt2) at (3,-1.1) {$\widetilde{b}_2$};
		\node (btK) at (7,-1.1) {$\widetilde{b}_K$};
		\node (btK1) at (9,-1.1) {$\widetilde{b}_{K+1}$};
		\node (btK2) at (11,-1.1) {$\widetilde{b}_{K+2}$};
		\node (btK3) at (13,-1.1) {$\widetilde{b}_{K+3}$};
		
		\node (carr0) at (0,0) {$u_\ell$};
		\node (carr1) at (2,0) {$u_\ell^{(1)}$};
		\node (carr2) at (4,0) {$u_\ell^{(2)}$};
		\node (carrK-1) at (6,0) {$ u_{\ell}^{(K-1)}$};
		\node (carrK0) at (8,0) {$u_\ell^{(K)}$};
		\node (carrK1) at (10,0) {$u_\ell^{(K+1)}$};
		\node (carrK2) at (12,0) {$u_\ell^{(K+2)}$};
		\node (carrK3) at (14.6,0) {$\cdots$};
		
		\foreach \i in {1,2,K,{K1},{K2},{K3}}
		{
			\draw[->] (b\i) -- (bt\i);
		}
		\foreach \i[evaluate=\i as \j using int(\i+1)] in {0,1}
		{
			\draw[->] (carr\i) -- (carr\j);
		}
		\node (carrdots) at (4.9,0) {$\cdots$};
		\foreach \i[evaluate=\i as \j using int(\i+1)] in {-1,...,2}
		{
			\draw[->] (carrK\i) -- (carrK\j);
		}
	\end{tikzpicture}
	\]
	where $ R(u_\ell^{(j)}\otimes b_{j+1}) = \widetilde{b}_{j+1}\otimes u_\ell^{(j+1)} $.
	
	\begin{ex}
		For $ U_q(\agl(3|3)) $ crystals,
		\[
		\begin{tikzpicture}
			\node (carr0) at (-1,0) {$ 
				\begin{matrix}
					\bth & \bth\\
					\btw & \btw
				\end{matrix}
				$};
			
			\node (carr1) at (1,0) {$ 
				\begin{matrix}
					\bth & \btw\\
					\btw & 3
				\end{matrix}
				$};
			
			\node (carr2) at (3,0) {$ 
				\begin{matrix}
					\bth & \btw\\
					\bon & 3
				\end{matrix}
				$};
			
			\node (carr3) at (5,0) {$ 
				\begin{matrix}
					\bth & \bth \\
					\btw & \bon
				\end{matrix}
				$};
			
			\node (carr4) at (7,0) {$ 
				\begin{matrix}
					\bth & \bth\\
					\btw & \btw
				\end{matrix}
				$};
			
			\node (carr5) at (9,0) {$ 
				\begin{matrix}
					\bth & \bth\\
					\btw & \btw
				\end{matrix}
				$};
			
			\node (carr6) at (11,0) {$ \cdots $};
			
			\foreach \i[evaluate=\i as \si using int(\i+1)] in {0,...,5} {
				\draw [->] (carr\i) -- (carr\si);
			}
			
			\node (b0) at (0,1.25) {$ 
				\begin{matrix}
					\btw\\
					3
				\end{matrix}
				$};
			
			\node (b1) at (2,1.25) {$ 
				\begin{matrix}
					\bth\\
					\bon
				\end{matrix}
				$};
			\foreach \i in {2,...,4} {
				\node (b\i) at ({2*\i},1.25) {$
					\begin{matrix}
						\bth\\
						\btw
					\end{matrix}
					$};
			}
			
			\foreach \i in {0,1,4} {
				\node (bt\i) at ({2*\i},-1.25) {$ 
					\begin{matrix}
						\bth\\
						\btw
					\end{matrix}
					$};
			}
			\node (bt2) at (4,-1.25) {$ 
				\begin{matrix}
					\btw\\
					3
				\end{matrix}
				$};
			\node (bt3) at (6,-1.25) {$ 
				\begin{matrix}
					\bth\\
					\bon
				\end{matrix}
				$};
			\foreach \i in {0,...,4} {
				\draw [->] (b\i) -- (bt\i);
			}
			
		\end{tikzpicture}
		.
		\]
		That is,
		\[
		p = 
		\begin{array}{|c|}
			\hline
			\raisebox{-1pt}{$ \btw $}\\
			\hline
			\raisebox{-1pt}{$ 3 $}\\
			\hline
		\end{array}
		\otimes
		\begin{array}{|c|}
			\hline
			\raisebox{-1pt}{$ \bth $}\\
			\hline
			\raisebox{-1pt}{$ \bon $}\\
			\hline
		\end{array}
		\otimes u_1 \otimes u_1\otimes u_1 \otimes \cdots
		\qquad\implies\qquad
		T_2(p) = u_1\otimes u_1\otimes
		\begin{array}{|c|}
			\hline
			\raisebox{-1pt}{$ \btw $}\\
			\hline
			\raisebox{-1pt}{$ 3 $}\\
			\hline
		\end{array}
		\otimes
		\begin{array}{|c|}
			\hline
			\raisebox{-1pt}{$ \bth $}\\
			\hline
			\raisebox{-1pt}{$ \bon $}\\
			\hline
		\end{array}
		\otimes
		u_1\otimes\cdots
		.
		\]
	\end{ex}
	
	\begin{remark}\label{rmk:r>m}
		We only consider BBSs with $ r\le m $.
		This is because we encounter difficulties if we consider $ r>m $.
		For instance, the empty carrier $ u_\ell $ will contain fermionic boxes
		which will increase in value horizontally.
		So, we can no longer think of $ u_\ell $ as containing $ \ell $ vacuum elements.
		More concerning, the time evolution operator may no longer be well-defined.
		Consider a BBS defined from a $ U_q(\agl(1|3)) $-crystal with $ r=2 $ and 
		\[
		p= 
		\begin{array}{|c|}
			\hline
			\raisebox{-1pt}{$\bon$}\\
			\hline
			\raisebox{-1pt}{$3$}\\
			\hline
		\end{array}
		\otimes u_1 \otimes u_1 \otimes\cdots
		\qquad \text{where} \qquad u_1=
		\begin{array}{|c|}
			\hline
			\raisebox{-1pt}{$\bon$}\\
			\hline
			\raisebox{-1pt}{$1$}\\
			\hline
		\end{array}.
		\]
		Then,
		\[
		\cdots R_3R_2R_1R_0\left(
		\begin{array}{|c|c|}
			\hline
			\raisebox{-1pt}{$\bon$} & \raisebox{-1pt}{$\bon$}\\
			\hline
			\raisebox{-1pt}{$1$} & \raisebox{-1pt}{$2$} \\
			\hline
		\end{array}
		\otimes p\right) = 
		\begin{array}{|c|}
			\hline
			\bon\\
			\hline
			2\\
			\hline
		\end{array}
		\otimes u_1\otimes u_1 \otimes\cdots
		\otimes 
		\begin{array}{|c|c|}
			\hline
			\bon & \bon\\
			\hline
			1 & 3\\
			\hline
		\end{array}
		\]
		which is not of the form
		\[
		\widetilde{p}\otimes u_2
		= \widetilde{p}\otimes 
		\begin{array}{|c|c|}
			\hline
			\raisebox{-1pt}{$\bon$} & \raisebox{-1pt}{$\bon$}\\
			\hline
			\raisebox{-1pt}{$1$} & \raisebox{-1pt}{$2$} \\
			\hline
		\end{array}
		\]
		that is required for our above definition of $ T_2(p)$.
		Additionally, we do not want to define $ T_2(p) $ to be
		\[
		\begin{array}{|c|}
			\hline
			\bon\\
			\hline
			2\\
			\hline
		\end{array}
		\otimes u_1\otimes u_1\otimes \cdots
		\]
		because the boxes inside $ T_2(p) $ are different from $ p $;
		intuitively, the `mass' is no longer a conserved quantity.
		This indicates that other conserved quantities may not be present,
		which impacts the integrability of the system.
		Nevertheless, such a system with $ r>m $ may exhibit interesting behaviour or applications,
		but is beyond the scope of this article.
	\end{remark}
	
	\begin{remark}\label{rmk:timereversible}
		Many authors will define a state so that the tensor product extends infinitely in both directions
		(with only finitely many non-vacuum states).
		If we define a state in this way
		then the system is time reversible,
		since the uniqueness of the combinatorial $ R $-matrix implies $ R^{-1}_a = R_a $.
	\end{remark}
	\subsection{Properties of the time evolution operator}
	\begin{prop}\label{prop:Tlcomm}
		Time evolution operators commute: $ T_{\ell}T_{\ell'}(p) = T_{\ell'}T_{\ell}(p) $.
	\end{prop}
	The proof of this fact is identical to \cite[Theorem~3.1]{FOY00},
	and relies on the Yang--Baxter equation:
	$
	(R\otimes 1)(1\otimes R)(R\otimes 1) = (1\otimes R)(R\otimes 1)(1\otimes R)
	$.
	The Yang--Baxter equation is proved for $ U_q(\agl(m|n)) $-crystals in~\cite[Theorem~7.11]{KO18}.

	The time evolution operator also respects the crystal structure,
	i.e., $ T_{\ell} $ commutes with some of the crystal operators,
	as outlined in the following Lemma:
	\begin{lemma}\label{lem:efcom}
		For all $ i\in I\setminus \{\overline{0},\overline{m-r}\} $, and for a state $ p $,
		we have that $ T_{\ell}(e_i(p))=e_i(T_{\ell}(p)) $ and  $ T_{\ell}(f_i(p))=f_i(T_{\ell}(p)) $.
	\end{lemma}
	The proof is similar to \cite[Lemma~2.8]{Yamada04}.
	Let $ \dcrys $ be the $ U_q(\agl(m|n)) $-crystal of BBS states
	where the operators $ f_{\bze},\, e_{\bze},\,f_{\overline{m-r}} $ and $ e_{\overline{m-r}} $ 
	have been removed.
	Note that $ \dcrys $
	is isomorphic to a $ U_q(\gl(r))\otimes U_q(\gl(m-r|n)) $-crystal.
	Lemma~\ref{lem:efcom} allows us to prove results by only considering 
	a single element from each connected component of $ \dcrys $.
	In practice, this means that it is sufficient to consider genuine highest weight elements of $ \dcrys $.
	
	\section{Solitons} \label{sec:solitons}
	\subsection{Properties of solitons and coupled solitons}\label{sec:coupled}
	In this paper, a soliton is an element of $ (B^{r,1} \setminus \{u_1\})^{\otimes d} $ that moves with constant speed (not necessarily equal to its length).
	
	\begin{dfn}\label{dfn:solitons}
		We call an element $ v\in (B^{r,1} \setminus \{u_1\})^{\otimes d} $ a \defn{soliton} if 
		$ T_{\infty}(v\otimes u_1^{\otimes\infty}) = u_1^{\otimes c} \otimes v \otimes u_1^{\otimes\infty} $ for some positive integer $ c $.
		We call $ c $ the \defn{speed} of the soliton.
	\end{dfn}
	
	This definition is very broad, and these solitons do not satisfy many of the properties of we want.
	However, this broad definition is convenient for our purposes.
	Many of the important properties will be satisfied by a specific type of soliton (uncoupled solitons) which is defined by conserved quantities $ N_{\ell} $.
	
	Let $ p = p_1\otimes p_2 \otimes p_3 \otimes\cdots $ be a state.
	Let $ u_{\ell}^{(j)} $ be the carrier after applying the $ R $-matrix $ j $ times; that is
	\[
	R_{j-1}\cdots R_1R_0(u_{\ell}\otimes p) = \widetilde{p}_1\otimes\cdots\otimes \widetilde{p}_j \otimes u_{\ell}^{(j)}\otimes p_{j+1}\otimes\cdots.
	\]
	Define a function $ E_{\ell} $ by
	\[
	E_{\ell}(p) = -\sum_{j=1}^{\infty}H(u_{\ell}^{(j-1)}\otimes p_j)
	\]
	where $ H $ is the energy function
	(note that the above sum is finite because we chose $ H $ such that $ H(u_{\ell}\otimes u_1) = 0 $).
	
	\begin{prop}
		For each $ \ell $, the number $ E_{\ell}(p) $ is a \defn{conserved quantity}: $ E_{\ell}(T_{\ell'}(p)) = E_{\ell}(p) $ for every positive integer $ \ell' $.
	\end{prop}
	
	The proof of this proposition is the same as \cite[Theorem~3.1]{FOY00}.
	Define $ N_{\ell} $ by
	\[
	N_{\ell} = - E_{\ell - 1} + 2E_{\ell} - E_{\ell + 1}
	\]
	with $ E_0 = 0 $.
	We can now use $ N_{\ell} $ to define uncoupled solitons.
	
	\begin{dfn}
		Let $ v $ be a soliton.
		If there exists a positive integer $ s $ for which $ N_{s}(v\otimes u_1^{\otimes \infty}) = 1 $ and $ N_{j}(v\otimes u_1^{\otimes \infty}) = 0 $ for all $ j\neq s $,
		then we call $ v $ an \defn{uncoupled soliton}. Otherwise, we call $ v $ a \defn{coupled soliton}.
	\end{dfn}
	
	Note that $ s \leq d $.
	Indeed, the proof of Proposition~\ref{prop:Tinfwelldefined} shows that $ T_{\ell} = T_{\ell'} $ for $ \ell,\ell' \geq d $. In fact, the same insertion argument shows that $ E_{\ell} = E_{\ell'} $ for $ \ell,\ell'\geq d $ and hence $ N_{j} = 0 $ for $ j > d $.
	Moreover,
	we believe that $ s=d $ for every uncoupled soliton (see the discussion after Conjecture~\ref{conj:dividingline}) so we can interpret $ s $ as the length of the uncoupled soliton.
	
	Intuitively, a coupled soliton contains overlapping uncoupled solitons that are not interacting with one another.
	With this intuition in mind, we can interpret $ N_{\ell} $ as the total number of uncoupled solitons of length $ \ell $ in a state. 
	With this interpretation, we should expect that a coupled soliton has $ N_s\neq 0 $ for exactly one positive integer $ s $.
	Indeed, were this not the case, then a coupled soliton would contain overlapping uncoupled solitons of different speeds, and we would expect these uncoupled solitons to separate given enough time (contradicting the fact that the coupled soliton is a soliton --- Definition~\ref{dfn:solitons}).

	We formalise this intuitive interpretation of $ N_{\ell} $ in the following conjecture, which claims that a coupled soliton $ v $ can be split into uncoupled solitons after collision.
	
	\begin{conj}\label{conj:uncoupling}
		Let $ v $ be a coupled soliton of speed $ s $. 
		There exist positive integers $ \widetilde{t}, A, c_1,\ldots, c_A $ and uncoupled solitons $ w_1,\ldots, w_A $ of speeds $ d_1\le\ldots\le d_A $ (respectively) greater than $ s $
		such that if $ t>\widetilde{t} $ then
		\begin{align*}
			&(T_{\infty})^t(w_1\otimes u_1^{\otimes c_1}\otimes\cdots\otimes w_A\otimes u_1^{\otimes c_A} \otimes v \otimes u_1^{\otimes \infty})\\
			&= u_1^{\otimes \widetilde{c}_1+st} \otimes \widetilde{v}_1 \otimes u_1^{\otimes \widetilde{c}_2} \otimes \widetilde{v}_2 \otimes \cdots \otimes  u_1^{\otimes \widetilde{c}_B} \otimes \widetilde{v}_B\\
			&\qquad
			\otimes u_1^{\otimes \widetilde{c}_{B+1}+(d_1-s)t} \otimes \widetilde{w}_1 \otimes u_1^{\otimes \widetilde{c}_{B+2}+(d_2-d_1)t} \otimes \widetilde{w}_2 \otimes \cdots \otimes  u_1^{\otimes \widetilde{c}_{B+A}+(d_A-d_{A-1})t} \otimes \widetilde{w}_A \otimes u_1^{\otimes\infty}
		\end{align*}
		for some uncoupled solitons $ \widetilde{v}_1,\widetilde{v}_2,\ldots,\widetilde{v}_B,\widetilde{w}_1,\ldots,\widetilde{w}_A $,
		where $ B = N_{s}(v \otimes u_1^{\otimes \infty}) $,
		and for some integers $ \widetilde{c}_1,\widetilde{c}_2,\ldots,\widetilde{c}_{B+A} $.
		Note that each $ \widetilde{v}_j $ has speed $ s $ and that each $ \widetilde{w}_j $ has speed $ d_j $.
	\end{conj}
	
	\begin{ex} \label{ex:uncoupling}
		Consider the $ U_q(\agl(3|1)) $ coupled soliton $
		v=
		\begin{smallmatrix}
			\bon\\
			1
		\end{smallmatrix}
		$.
		Note that $ N_1(v \otimes u_1^{\otimes \infty})=2 $ and $ N_j(v \otimes u_1^{\otimes \infty})=0 $ for $ j\neq1 $.
		So we expect that $ v $ contains two overlapping uncoupled solitons of length one.
		In the following diagram, we pass the uncoupled soliton
		$
		\begin{smallmatrix}
			\bth & \bth & \bth\\
			\bon & \bon & \bon
		\end{smallmatrix}
		$
		through a state containing $ v $.
		The maximal weight element
		$
		\begin{smallmatrix}
			\bth\\
			\btw
		\end{smallmatrix}
		$
		is represented as a dot.
		We find that $ v $ splits into two copies of the uncoupled soliton
		$
		\begin{smallmatrix}
			\bth\\
			\bon 
		\end{smallmatrix}
		$
		after collision.
		\[
		\begin{array}{cc}
			t = 0 &
			\begin{tikzpicture}[scale = 0.9, baseline = -2,every node/.style={scale=1}]
				\def\s{0.5};
				\foreach \x in {-5,-1,0,2,3,...,12} {
					\draw[fill=black] (\x*\s,0) circle (0.025);
				}
				\foreach \count/\valone/\valtwo/\valthree/\valfour in {0/\bon/\bon/\bon/1,1/\bth/\bth/\bth/\bon}{
					\node (d) at (-4*\s,0 + \count*\s) {$\valone$};
					\node (a) at (-3*\s,0 + \count*\s) {$\valtwo$};
					\node (b) at (-2*\s,0 + \count*\s) {$\valthree$};
					\node (c) at (1*\s,0 + \count*\s) {$\valfour$};
				}
			\end{tikzpicture}
			\\[7pt]
			t = 1 &
			
			\begin{tikzpicture}[scale = 0.9, baseline = -2,every node/.style={scale=1}]
				\def\s{0.5};
				\foreach \x in {-5,...,-2,1,4,5,6,7,8,9,10,11,12} {
					\draw[fill=black] (\x*\s,0) circle (0.025);
				}
				\foreach \count/\valone/\valtwo/\valthree/\valfour in {0/\bon/\bon/1/\bon,1/\bth/\bth/\bon/\bth}{
					\node (a) at (-1*\s,0 + \count*\s) {$\valone$};
					\node (b) at (0*\s,0 + \count*\s) {$\valtwo$};
					\node (c) at (2*\s,0 + \count*\s) {$\valthree$};
					\node (d) at (3*\s,0 + \count*\s) {$\valfour$};
				}
			\end{tikzpicture}
			
			\\[7pt]
			t = 2 &
			
			\begin{tikzpicture}[scale = 0.9, baseline = -2,every node/.style={scale=1}]
				\def\s{0.5};
				\foreach \x in {-5,...,0,2,7,8,9,10,11,12} {
					\draw[fill=black] (\x*\s,0) circle (0.025);
				}
				\foreach \count/\valone/\valtwo/\valthree/\valfour/\valfive in {0/\bon/\bon/1/\bon/\bon,1/\bth/\bth/\btw/\bth/\bth}{
					\node (a) at (1*\s,0 + \count*\s) {$\valone$};
					\node (b) at (3*\s,0 + \count*\s) {$\valtwo$};
					\node (c) at (4*\s,0 + \count*\s) {$\valthree$};
					\node (d) at (5*\s,0 + \count*\s) {$\valfour$};
					\node (e) at (6*\s,0 + \count*\s) {$\valfive$};
				}
			\end{tikzpicture}
			\\[7pt]
			t = 3 &
			
			\begin{tikzpicture}[scale = 0.9, baseline = -2,every node/.style={scale=1}]
				\def\s{0.5};
				\foreach \x in {-5,...,1,3,5,6,10,11,12} {
					\draw[fill=black] (\x*\s,0) circle (0.025);
				}
				\foreach \count/\valone/\valtwo/\valthree/\valfour/\valfive in {0/\bon/\bon/1/\bon/\bon,1/\bth/\bth/\btw/\bth/\bth}{
					\node (a) at (2*\s,0 + \count*\s) {$\valone$};
					\node (d) at (4*\s,0 + \count*\s) {$\valtwo$};
					\node (b) at (7*\s,0 + \count*\s) {$\valthree$};
					\node (c) at (8*\s,0 + \count*\s) {$\valfour$};
					\node (e) at (9*\s,0 + \count*\s) {$\valfive$};
				}
			\end{tikzpicture}
			\\[7pt]
			t = 4 &
			
			\begin{tikzpicture}[scale = 0.9, baseline = -2,every node/.style={scale=1}]
				\def\s{0.5};
				\foreach \x in {-5,...,1,2,4,6,7,8,9} {
					\draw[fill=black] (\x*\s,0) circle (0.025);
				}
				\foreach \count/\valone/\valtwo/\valthree/\valfour/\valfive in {0/\bon/\bon/1/\bon/\bon,1/\bth/\bth/\btw/\bth/\bth}{
					\node (a) at (3*\s,0 + \count*\s) {$\valone$};
					\node (d) at (5*\s,0 + \count*\s) {$\valtwo$};
					\node (b) at (10*\s,0 + \count*\s) {$\valthree$};
					\node (c) at (11*\s,0 + \count*\s) {$\valfour$};
					\node (e) at (12*\s,0 + \count*\s) {$\valfive$};
				}
			\end{tikzpicture}
			\\[7pt]
		\end{array}\]
	\end{ex}
	
	\begin{ex} \label{ex:uncoupling2}
		Consider the $ U_q(\agl(4|1)) $ coupled soliton $
		v=
		\begin{smallmatrix}
			\bon\\
			1\\
			1
		\end{smallmatrix}
		$.
		Note that $ N_1(v \otimes u_1^{\otimes \infty})=3 $ and $ N_j(v \otimes u_1^{\otimes \infty})=0 $ for $ j\neq1 $.
		So we expect that $ v $ contains three overlapping uncoupled solitons of length one.
		In the following diagram, we pass two copies of the uncoupled soliton
		$
		\begin{smallmatrix}
			\bfo & \bfo\\
			\bth & \bth\\
			\bon & \bon
		\end{smallmatrix}
		$
		through a state containing $ v $.
		The maximal weight element
		$
		\begin{smallmatrix}
			\bfo\\
			\bth\\
			\btw
		\end{smallmatrix}
		$
		is represented as a dot.
		We find that $ v $ splits into three copies of the uncoupled soliton
		$
		\begin{smallmatrix}
			\bfo\\
			\bth\\
			\bon 
		\end{smallmatrix}
		$
		after collision.
		\begin{align*}
			t = 0 & \quad
			\begin{tikzpicture}[scale = 0.9, baseline = -2,every node/.style={scale=1}]
				\def\s{0.5};
				\foreach \x in {3,4,5,8,9,11,12,...,25} {
					\draw[fill=black] (\x*\s,0) circle (0.025);
				}
				\foreach \count/\valone/\valtwo/\valthree/\valfive/\valsix in 
				{0/\bon/\bon/\bon/\bon/1,
					1/\bth/\bth/\bth/\bth/1,
					2/\bfo/\bfo/\bfo/\bfo/\bon}{
					\node (a) at (1*\s,0 + \count*\s) {$\valone$};
					\node (d) at (2*\s,0 + \count*\s) {$\valtwo$};
					\node (b) at (6*\s,0 + \count*\s) {$\valthree$};
					\node (e) at (7*\s,0 + \count*\s) {$\valfive$};
					\node (e) at (10*\s,0 + \count*\s) {$\valsix$};
				}
			\end{tikzpicture}
			\allowdisplaybreaks
			\\[7pt]
			t = 1 & \quad
			\begin{tikzpicture}[scale = 0.9, baseline = -2,every node/.style={scale=1}]
				\def\s{0.5};
				\foreach \x in {1,2,5,6,7,10,12,13,...,25} {
					\draw[fill=black] (\x*\s,0) circle (0.025);
				}
				\foreach \count/\valone/\valtwo/\valthree/\valfive/\valsix in 
				{0/\bon/\bon/\bon/\bon/1,
					1/\bth/\bth/\bth/\bth/1,
					2/\bfo/\bfo/\bfo/\bfo/\bon}{
					\node (a) at (3*\s,0 + \count*\s) {$\valone$};
					\node (d) at (4*\s,0 + \count*\s) {$\valtwo$};
					\node (b) at (8*\s,0 + \count*\s) {$\valthree$};
					\node (e) at (9*\s,0 + \count*\s) {$\valfive$};
					\node (e) at (11*\s,0 + \count*\s) {$\valsix$};
				}
			\end{tikzpicture}
			\allowdisplaybreaks
			\\[7pt]
			t = 2 & \quad
			\begin{tikzpicture}[scale = 0.9, baseline = -2,every node/.style={scale=1}]
				\def\s{0.5};
				\foreach \x in {1,...,4,7,8,9,11,14,15,...,25} {
					\draw[fill=black] (\x*\s,0) circle (0.025);
				}
				\foreach \count/\valone/\valtwo/\valthree/\valfive/\valsix in 
				{0/\bon/\bon/\bon/1/\bon,
					1/\bth/\bth/\bth/1/\bth,
					2/\bfo/\bfo/\bfo/\bon/\bfo}{
					\node (a) at (5*\s,0 + \count*\s) {$\valone$};
					\node (d) at (6*\s,0 + \count*\s) {$\valtwo$};
					\node (b) at (10*\s,0 + \count*\s) {$\valthree$};
					\node (e) at (12*\s,0 + \count*\s) {$\valfive$};
					\node (e) at (13*\s,0 + \count*\s) {$\valsix$};
				}
			\end{tikzpicture}
			\allowdisplaybreaks
			\\[7pt]
			t = 3 & \quad
			\begin{tikzpicture}[scale = 0.9, baseline = -2,every node/.style={scale=1}]
				\def\s{0.5};
				\foreach \x in {1,...,6,9,10,12,16,17,...,25} {
					\draw[fill=black] (\x*\s,0) circle (0.025);
				}
				\foreach \count/\valone/\valtwo/\valthree/\valfive/\valsix/\valseven in 
				{0/\bon/\bon/\bon/1/1/\bon,
					1/\bth/\bth/\bth/\bon/\btw/\bth,
					2/\bfo/\bfo/\bfo/\bfo/\bth/\bfo}{
					\node (a) at (7*\s,0 + \count*\s) {$\valone$};
					\node (d) at (8*\s,0 + \count*\s) {$\valtwo$};
					\node (b) at (11*\s,0 + \count*\s) {$\valthree$};
					\node (e) at (13*\s,0 + \count*\s) {$\valfive$};
					\node (e) at (14*\s,0 + \count*\s) {$\valsix$};
					\node (e) at (15*\s,0 + \count*\s) {$\valseven$};
				}
			\end{tikzpicture}
			\allowdisplaybreaks
			\\[7pt]
			t = 4 & \quad
			\begin{tikzpicture}[scale = 0.9, baseline = -2,every node/.style={scale=1}]
				\def\s{0.5};
				\foreach \x in {1,...,8,11,13,15,18,19,...,25} {
					\draw[fill=black] (\x*\s,0) circle (0.025);
				}
				\foreach \count/\valone/\valtwo/\valthree/\valfive/\valsix/\valseven in 
				{0/\bon/\bon/\bon/1/1/\bon,
					1/\bth/\bth/\bth/\bon/\btw/\bth,
					2/\bfo/\bfo/\bfo/\bfo/\bth/\bfo}{
					\node (a) at (9*\s,0 + \count*\s) {$\valone$};
					\node (d) at (10*\s,0 + \count*\s) {$\valtwo$};
					\node (b) at (12*\s,0 + \count*\s) {$\valthree$};
					\node (e) at (14*\s,0 + \count*\s) {$\valfive$};
					\node (e) at (16*\s,0 + \count*\s) {$\valsix$};
					\node (e) at (17*\s,0 + \count*\s) {$\valseven$};
				}
			\end{tikzpicture}
			\allowdisplaybreaks
			\\[7pt]
			t = 5 & \quad
			\begin{tikzpicture}[scale = 0.9, baseline = -2,every node/.style={scale=1}]
				\def\s{0.5};
				\foreach \x in {1,...,10,12,14,17,20,21,...,25} {
					\draw[fill=black] (\x*\s,0) circle (0.025);
				}
				\foreach \count/\valone/\valtwo/\valthree/\valfive/\valsix/\valseven in 
				{0/\bon/\bon/1/\bon/1/\bon,
					1/\bth/\bth/\bon/\bth/\btw/\bth,
					2/\bfo/\bfo/\bfo/\bfo/\bth/\bfo}{
					\node (a) at (11*\s,0 + \count*\s) {$\valone$};
					\node (d) at (13*\s,0 + \count*\s) {$\valtwo$};
					\node (b) at (15*\s,0 + \count*\s) {$\valthree$};
					\node (e) at (16*\s,0 + \count*\s) {$\valfive$};
					\node (e) at (18*\s,0 + \count*\s) {$\valsix$};
					\node (e) at (19*\s,0 + \count*\s) {$\valseven$};
				}
			\end{tikzpicture}
			\allowdisplaybreaks
			\\[7pt]
			t = 6 & \quad
			\begin{tikzpicture}[scale = 0.9, baseline = -2,every node/.style={scale=1}]
				\def\s{0.5};
				\foreach \x in {1,...,11,13,15,19,22,23,...,25} {
					\draw[fill=black] (\x*\s,0) circle (0.025);
				}
				\foreach \count/\valone/\valtwo/\valthree/\valfour/\valfive/\valsix/\valseven in 
				{0/\bon/\bon/\bon/1/\bon/1/\bon,
					1/\bth/\bth/\bth/\btw/\bth/\btw/\bth,
					2/\bfo/\bfo/\bfo/\bfo/\bfo/\bth/\bfo}{
					\node (a) at (12*\s,0 + \count*\s) {$\valone$};
					\node (d) at (14*\s,0 + \count*\s) {$\valtwo$};
					\node (b) at (16*\s,0 + \count*\s) {$\valthree$};
					\node (c) at (17*\s,0 + \count*\s) {$\valfour$};
					\node (e) at (18*\s,0 + \count*\s) {$\valfive$};
					\node (e) at (20*\s,0 + \count*\s) {$\valsix$};
					\node (e) at (21*\s,0 + \count*\s) {$\valseven$};
				}
			\end{tikzpicture}
			\\[7pt]
			t = 7 & \quad
			\begin{tikzpicture}[scale = 0.9, baseline = -2,every node/.style={scale=1}]
				\def\s{0.5};
				\foreach \x in {1,...,12,14,16,18,21,24,25} {
					\draw[fill=black] (\x*\s,0) circle (0.025);
				}
				\foreach \count/\valone/\valtwo/\valthree/\valfour/\valfive/\valsix/\valseven in 
				{0/\bon/\bon/\bon/1/\bon/1/\bon,
					1/\bth/\bth/\bth/\btw/\bth/\btw/\bth,
					2/\bfo/\bfo/\bfo/\bfo/\bfo/\bth/\bfo}{
					\node (a) at (13*\s,0 + \count*\s) {$\valone$};
					\node (d) at (15*\s,0 + \count*\s) {$\valtwo$};
					\node (b) at (17*\s,0 + \count*\s) {$\valthree$};
					\node (c) at (19*\s,0 + \count*\s) {$\valfour$};
					\node (e) at (20*\s,0 + \count*\s) {$\valfive$};
					\node (e) at (22*\s,0 + \count*\s) {$\valsix$};
					\node (e) at (23*\s,0 + \count*\s) {$\valseven$};
				}
			\end{tikzpicture}
			\\[7pt]
			t = 8 & \quad
			\begin{tikzpicture}[scale = 0.9, baseline = -2,every node/.style={scale=1}]
				\def\s{0.5};
				\foreach \x in {1,...,13,15,17,19,20,23} {
					\draw[fill=black] (\x*\s,0) circle (0.025);
				}
				\foreach \count/\valone/\valtwo/\valthree/\valfour/\valfive/\valsix/\valseven in 
				{0/\bon/\bon/\bon/1/\bon/1/\bon,
					1/\bth/\bth/\bth/\btw/\bth/\btw/\bth,
					2/\bfo/\bfo/\bfo/\bfo/\bfo/\bth/\bfo}{
					\node (a) at (14*\s,0 + \count*\s) {$\valone$};
					\node (d) at (16*\s,0 + \count*\s) {$\valtwo$};
					\node (b) at (18*\s,0 + \count*\s) {$\valthree$};
					\node (c) at (21*\s,0 + \count*\s) {$\valfour$};
					\node (e) at (22*\s,0 + \count*\s) {$\valfive$};
					\node (e) at (24*\s,0 + \count*\s) {$\valsix$};
					\node (e) at (25*\s,0 + \count*\s) {$\valseven$};
				}
			\end{tikzpicture}
		\end{align*}
	\end{ex}
	
	We know that $ N_{\ell} $ is a conserved quantity (since $ E_{\ell} $ is). We can interpret this fact as a form of stability under collision, which is one of the important properties of solitons.
	
	In the height $ 1 $ BBS, every state separates into solitons given enough time.
	We conjecture that this is also true in our system (though the solitons might be coupled).
	
	\begin{conj}\label{conj:separate}
		Let $ p $ be any state.
		There exists some positive integer $ \widetilde{t} $,
		some solitons $ v_1,v_2,\ldots,v_D $ of speeds $ d_1,d_2,\ldots,d_D $ (respectively) and some positive integers $ c_1,c_2,\ldots,c_D $,
		such that for any $ t>\widetilde{t} $,
		\[
		(T_{\infty})^t(p) = u_1^{\otimes c_1+d_1(t-\widetilde{t})} \otimes v_1 \otimes u_1^{\otimes c_2+(d_2-d_1)(t-\widetilde{t})} \otimes v_2 \otimes \cdots \otimes u_1^{\otimes c_D+(d_D-d_{D-1})(t-\widetilde{t})} \otimes v_D \otimes u_1^{\otimes \infty}.
		\]
	\end{conj}	
	We have verified this conjecture for all $ U_q(\agl(2|2)) $ and $ U_q(\agl(3|3)) $ states with only the first five factors being non-vacuum elements.
	
	There are two important properties that the solitons of the KdV equation and of height 1 BBSs satisfy:
	they move with speed corresponding to their length and are stable under collision.
	In general, the solitons of Definition~\ref{dfn:solitons} do not satisfy these two properties (even though $ N_{\ell} $ is conserved).
	However, we will show that \emph{uncoupled} solitons do satisfy these two properties.
	
	\subsection{Solitons with speed equal to their length}
	One of the properties of the height 1 BBS is that the speed of the solitons is equal to their length.
	In general, this is not true of the solitons in our system.
	In this section, we provide a large class of solitons which move with speed corresponding to their length.
	We conjecture that this is the largest such class and that it contains all of the uncoupled solitons.
	
	\begin{thm}\label{thm:speed}
		Let $ B^{r,1} $ be the $ U_q(\gl(m|n)) $-crystal (with $ r\le m $) 
		of rectangular semistandard Young tableaux (SSYT) with height $ r $ and width $ 1 $.
		Let 
		\[
		x = 
		\begin{array}{|c|}
			\hline
			\raisebox{-1pt}{$ x_{11} $} \\
			\hline
			\raisebox{-1pt}{$ x_{21} $} \\
			\hline
			\vdots \\
			\hline
			\raisebox{-1pt}{$ x_{r1} $} \\
			\hline
		\end{array}
		\otimes 
		\begin{array}{|c|}
			\hline
			\raisebox{-1pt}{$ x_{12} $} \\
			\hline
			\raisebox{-1pt}{$ x_{22} $} \\
			\hline
			\vdots \\
			\hline
			\raisebox{-1pt}{$ x_{r2} $} \\
			\hline
		\end{array}
		\otimes \cdots \otimes 
		\begin{array}{|c|}
			\hline
			\raisebox{-1pt}{$ x_{1s} $} \\
			\hline
			\raisebox{-1pt}{$ x_{2s} $} \\
			\hline
			\vdots \\
			\hline
			\raisebox{-1pt}{$ x_{rs} $} \\
			\hline
		\end{array}
		\in (B^{r,1})^{\otimes s}.
		\]
		Suppose the factors of the tensor product in reverse order:
		\[
		\begin{array}{|c|c|c|c|}
			\hline
			\raisebox{-1pt}{$ x_{1s} $} 
			& \raisebox{-1pt}{$ \cdots  $} 
			& \raisebox{-1pt}{$ x_{12} $} 
			& \raisebox{-1pt}{$ x_{11} $}\\
			\hline
			\raisebox{-1pt}{$ x_{2s} $} 
			& \raisebox{-1pt}{$ \cdots  $} 
			& \raisebox{-1pt}{$ x_{22} $} 
			& \raisebox{-1pt}{$ x_{21} $}\\
			\hline
			\raisebox{-1pt}{$ \vdots  $} 
			& \raisebox{-1pt}{$ \ddots  $} 
			& \raisebox{-1pt}{$ \vdots  $}
			& \raisebox{-1pt}{$ \vdots  $} \\
			\hline
			\raisebox{-1pt}{$ x_{rs} $} 
			& \raisebox{-1pt}{$ \cdots  $} 
			& \raisebox{-1pt}{$ x_{r2} $}
			& \raisebox{-1pt}{$ x_{r1} $} \\
			\hline
		\end{array}	
		\]
		form a SSYT and that there is a row number $ k $ 
		($ 1\le k\le r $) such that
		\begin{align*}
			x_{ij}&<\overline{m-r} \qquad \text{for all } j \text{ and for } i<k\\
			x_{ij}&\ge \overline{m-r} \qquad \text{for all } j \text{ and for } i\ge k.
		\end{align*}
		Then,
		$
		(T_{\ell})^t(u_1^{\otimes c}\otimes x\otimes u_1^{\otimes \infty }) = u_1^{\otimes (c+t\min\{s,\ell\})}\otimes x\otimes u_1^{\otimes \infty }
		$ for all positive integers $ t $.
	\end{thm}
	We first use the $ R $-matrix insertion algorithm to prove the theorem for a genuine highest weight state,
	and then generalise using Lemma~\ref{lem:efcom}.
	The proof is given in Appendix~\ref{appendix:single_soliton_proof}.
	
	\begin{conj}\label{conj:allspeedlength}
		The subset of $ \bigcup_{d=1}^{\infty}(B^{r,1})^{\otimes d} $ defined by Theorem~\ref{thm:speed} 
		is the largest such subset of solitons which move with speed equal to their length.
	\end{conj}
	We verified this conjecture experimentally for solitons 
	with $ r=2,\, s=2 $, with $ r=2,\, s=3 $ and with $ r=3,\,s=2 $
	for $ m=n=3 $ and $ m=n=4 $.
	
	The value of $ k $ relates the structure of these solitons to the value of $ N_{s} $.
	
	\begin{prop}\label{prop:Nlisgood}
		Let $ x,\, k $ be as in Theorem~\ref{thm:speed}.
		Then $ N_{s}(u_1^{\otimes c}\otimes x \otimes u_1^{\otimes\infty}) = r+1-k $
		and $ N_{\ell} = 0 $ for $ \ell \neq s $.
	\end{prop}
	\begin{proof}
		Let $ p_1\otimes p_2 \otimes \cdots =  u_1^{\otimes c}\otimes x \otimes u_1^{\otimes\infty} $.
		We can compute $ H(u_{\ell}^{(j-1)}\otimes u_1) = 0 $.
		Applying the bumping algorithm while moving the carrier through the soliton (c.f.\ Appendix~\ref{appendix:single_soliton_proof}),
		we find that $ H(u_{\ell}^{(j-1)}\otimes p_j) = r+1-k $
		unless the carrier $ u_{\ell}^{(j-1)} $ is full,
		in which case $ H(u_{\ell}^{(j-1)}\otimes p_j) = 0 $.
		Hence, $ E_{\ell}(p) = (r+1-k)\min(\ell,s) $.
		It is then easily verified that $ N_{\ell}(p) = r+1-k $ if $ \ell=s $ and $ N_{\ell}(p) = 0 $ otherwise.
	\end{proof}
	
	In particular, note that if $ k=r $ then $ x $ is an uncoupled soliton.
	We propose the following generalisation of the above proposition.
	
	\begin{conj}\label{conj:dividingline}
		Let $ v = v_1\otimes v_2\otimes \cdots \otimes v_{d} \in (B^{r,1})^{\otimes d} $ be a soliton of speed s.
		Let $ k_j $ be the value of $ k $ (defined in Theorem~\ref{thm:speed}) for each $ v_j $.
		Then
		\begin{equation}\label{eq:conjdividingline}
			\sum_{j=1}^{d} (r+1- k_j) =  s N_{s}(v\otimes u^{\otimes \infty}).
		\end{equation}
	\end{conj}
	Intuitively, we interpret $ r+1-k_j $ as the number of overlapping solitons at the $ j $-th position of the soliton.
	
	\begin{prop}
		If we assume Conjecture~\ref{conj:dividingline}, then the solitons of Theorem~\ref{thm:speed} with $ k=r $ are the only uncoupled solitons. 
	\end{prop}
	\begin{proof}
		Since $ r+1-k_j \geq 1 $,
		the left-hand side of \eqref{eq:conjdividingline}
		is at least $ d $,
		and for an uncoupled soliton,
		the right-hand side of \eqref{eq:conjdividingline} is $ s $.
		But we already know $ s \leq d $.
		We deduce that $ s=d $ and hence $ k_j = r $ for all $ j $.
	\end{proof}
	
	\subsection{Scattering of two solitons}\label{sec:scattering}
	One of the main properties of the height 1 BBS is that solitons are stable under collision.
	This behaviour is also called scattering.
	
	\begin{dfn}\label{dfn:scattering}
		Let $ v,w $ be uncoupled solitons of lengths $ \mathfrak{d}_1,\,\mathfrak{d}_2 $ respectively, with $ \mathfrak{d}_1>\mathfrak{d}_2 $.
		We say that $ v $ and $ w $ \defn{scatter} 
		if there exist non-negative integers $ \mathfrak{c}_2,\widetilde{t} $ such that
		for any $ t>\widetilde{t} $ and $ \mathfrak{c}_1\in\ZZ_{\ge 0} $,
		\[
		(T_{\infty})^t(u_1^{\otimes \mathfrak{c}_1}\otimes v \otimes u_1^{\otimes\mathfrak{c}_2}\otimes w \otimes u_1^{\otimes\infty})
		= u_1^{\otimes \mathfrak{c}_3} \otimes \widetilde{w} \otimes u_1^{\otimes\mathfrak{c}_4}\otimes \widetilde{v} \otimes u_1^{\otimes\infty}
		\]
		for some non-negative integers $ \mathfrak{c}_3,\mathfrak{c}_4 $
		(dependent on $ t $)
		and some uncoupled solitons $ \widetilde{w},\widetilde{v} $
		of lengths $ \mathfrak{d}_2,\mathfrak{d}_1 $ respectively.
	\end{dfn}
	
	We can interpret this definition as saying that the longer soliton, $ v $,
	eventually collides and interacts with the shorter soliton, $ w $,
	after which the states separate into two solitons again.
	However, it is important to note that $ \widetilde{w} $ and $ \widetilde{v} $
	are generally different from $ w $ and $ v $, respectively.
	We have already seen an example of scattering for $ \agl(m|n) $ in Example~\ref{ex:glmnscattering}.
	
	Using the same notation in Definition~\ref{dfn:scattering},
	let $ j_v=\mathfrak{c}_1+1,\, j_w=\mathfrak{c}_1+\mathfrak{d}_1+\mathfrak{c}_2+1,\,
	j_{\widetilde{w}}=\mathfrak{c}_3+1,\,j_{\widetilde{v}}=\mathfrak{c}_3+\mathfrak{d}_2+\mathfrak{c}_4+1$ be the positions of $ v,w,\widetilde{w},\widetilde{v} $ respectively.
	If there exists an integer $ \delta $ such that
	$ j_{\widetilde{v}}=j_v+t\mathfrak{d}_1+\delta $
	and $ j_{\widetilde{w}}=j_w+t\mathfrak{d}_2-\delta $
	then we call $ \delta $ the \defn{phase shift}.
	
	Let $ V $ be a SSYT.
	Let $ \down{V} $ denote the bottom row of $ V $,
	and $ \up{V} $ denote the other rows of $ V $.
	We will just consider the case where
	$ \down{V} $ only has entries greater than or equal to $ \overline{m-r} $,
	and $ \up{V} $ only has entries strictly less than $ \overline{m-r} $
	(where $ r $ is the height of $ V $).
	In the notation from Theorem~\ref{thm:speed},
	we are only considering the case where $ k=r $.
	
	\begin{thm}\label{thm:scattering}
		Consider uncoupled solitons composed of elements of $ U_q(\agl(m|n)) $-crystals with height $ r\le m $.
		Then any two uncoupled solitons of the form given in Theorem~\ref{thm:speed} (i.e.\ with $ k=r $) scatter.
		
		Moreover, let $ v,w $ be uncoupled solitons
		and let $ \widetilde{w},\widetilde{v} $ be obtained from $ v,w $ as in Definition~\ref{dfn:scattering}.
		The elements $ v $ and $ w $ are related to $ \widetilde{w} $ and $ \widetilde{v} $
		via their semistandard Young tableaux (SSYT).
		Let $ V,W,\widetilde{W},\widetilde{V} $ be the SSYT corresponding to $ v,w,\widetilde{w},\widetilde{v} $,
		respectively.
		Then,
		\[
		\up{\widetilde{W}}\otimes \up{\widetilde{V}} = R(\up{V}\otimes \up{W})
		\qquad \text{and} \qquad
		\down{\widetilde{W}}\otimes \down{\widetilde{V}} = R(\down{V}\otimes \down{W}).
		\]
		The phase shift is given by
		$
		\delta = 2\mathfrak{d}_2+H(\down{V}\otimes \down{W})+H(\up{V}\otimes \up{W}).
		$
	\end{thm}
	
	The proof of the above theorem is given in Appendix~\ref{sec:mainthmproof}
	
	Note that the assumption of uncoupled is not a necessary condition,
	and some other coupled solitons also scatter as in Defintion~\ref{dfn:scattering}.
	
	\begin{ex}\label{ex:non_thm_soliton}
		Consider the following time evolution of a BBS composed of elements from the $U_q(\widehat{\mathfrak{gl}}(4|1))$-crystal with $ r=2 $.
		\[
		\begin{array}{cc}
			t = 0 &
			\begin{tikzpicture}[scale = 0.9, baseline = -2,every node/.style={scale=1}]
				\def\s{0.5};
				\foreach \x in {-7,...,-4,-1,0,2,3,...,10} {
					\draw[fill=black] (\x*\s,0) circle (0.025);
				}
				\foreach \count/\valone/\valtwo/\valthree in {0/\bon/\bon/\bon,1/\btw/\btw/\btw}{
					\node (a) at (-3*\s,0 + \count*\s) {$\valone$};
					\node (b) at (-2*\s,0 + \count*\s) {$\valtwo$};
					\node (c) at (1*\s,0 + \count*\s) {$\valthree$};
				}
			\end{tikzpicture}
			
			\\[7pt]
			t = 1 &
			
			\begin{tikzpicture}[scale = 0.9, baseline = -2,every node/.style={scale=1}]
				\def\s{0.5};
				\foreach \x in {-7,...,-2,1,3,4,...,10} {
					\draw[fill=black] (\x*\s,0) circle (0.025);
				}
				\foreach \count/\valone/\valtwo/\valthree in {0/\bon/\bon/\bon,1/\btw/\btw/\btw}{
					\node (a) at (-1*\s,0 + \count*\s) {$\valone$};
					\node (b) at (0*\s,0 + \count*\s) {$\valtwo$};
					\node (c) at (2*\s,0 + \count*\s) {$\valthree$};
				}
			\end{tikzpicture}
			
			\\[7pt]
			
			t = 2 &
			
			\begin{tikzpicture}[scale = 0.9, baseline = -2,every node/.style={scale=1}]
				\def\s{0.5};
				\foreach \x in {-7,...,0,2,5,6,...,10} {
					\draw[fill=black] (\x*\s,0) circle (0.025);
				}
				\foreach \count/\valone/\valtwo/\valthree in {0/\bon/\bon/\bon,1/\btw/\btw/\btw}{
					\node (a) at (1*\s,0 + \count*\s) {$\valone$};
					\node (b) at (3*\s,0 + \count*\s) {$\valtwo$};
					\node (c) at (4*\s,0 + \count*\s) {$\valthree$};
				}
			\end{tikzpicture}
			
			\\[7pt]
			
			t = 3 &
			
			\begin{tikzpicture}[scale = 0.9, baseline = -2,every node/.style={scale=1}]
				\def\s{0.5};
				\foreach \x in {-7,...,1,3,4,7,8,...,10} {
					\draw[fill=black] (\x*\s,0) circle (0.025);
				}
				\foreach \count/\valone/\valtwo/\valthree in {0/\bon/\bon/\bon,1/\btw/\btw/\btw}{
					\node (a) at (2*\s,0 + \count*\s) {$\valone$};
					\node (b) at (5*\s,0 + \count*\s) {$\valtwo$};
					\node (c) at (6*\s,0 + \count*\s) {$\valthree$};
				}
			\end{tikzpicture}
			\\[7pt]
		\end{array}
		\]
		We observe that the two objects $ 
		\begin{smallmatrix}
			\btw & \btw\\
			\bon & \bon
		\end{smallmatrix}
		$ and $ 
		\begin{smallmatrix}
			\btw\\
			\bon
		\end{smallmatrix}
		$ satisfy Theorem~\ref{thm:speed} with $k = 1$ and are stable upon collision.
		However, $ 
		\begin{smallmatrix}
			\btw & \btw\\
			\bon & \bon
		\end{smallmatrix}
		$ and $ 
		\begin{smallmatrix}
			\btw\\
			\bon
		\end{smallmatrix}
		$ are not uncoupled.
		Calculating $ N_{\ell} $, we find that $ N_1 = 2,\, N_2 = 2 $ and $ N_{\ell} = 0 $ for $ \ell>2 $.
	\end{ex}
	
	However, in general, the collisions of coupled solitons are more complicated.
	\begin{ex}
		\label{ex:non_solitonic}
		Consider the following time evolution of a BBS 
		composed of elements from the $U_q(\widehat{\mathfrak{gl}}(3|3))$-crystal with $ r=2 $. 
		\[
		\begin{array}{cc}
			t = 0 &
			\begin{tikzpicture}[scale = 0.9, baseline = -2,every node/.style={scale=1}]
				\def\s{0.5};
				\foreach \x in {-5,-4,-1,0,2,3,...,12} {
					\draw[fill=black] (\x*\s,0) circle (0.025);
				}
				\foreach \count/\valone/\valtwo/\valthree in {0/2/1/1,1/\bon/\bon/\bon}{
					\node (a) at (-3*\s,0 + \count*\s) {$\valone$};
					\node (b) at (-2*\s,0 + \count*\s) {$\valtwo$};
					\node (c) at (1*\s,0 + \count*\s) {$\valthree$};
				}
			\end{tikzpicture}
			
			\\[7pt]
			t = 1 &
			
			\begin{tikzpicture}[scale = 0.9, baseline = -2,every node/.style={scale=1}]
				\def\s{0.5};
				\foreach \x in {-5,...,-2,1,3,4,...,12} {
					\draw[fill=black] (\x*\s,0) circle (0.025);
				}
				\foreach \count/\valone/\valtwo/\valthree in {0/2/1/1,1/\bon/\bon/\bon}{
					\node (a) at (-1*\s,0 + \count*\s) {$\valone$};
					\node (b) at (0*\s,0 + \count*\s) {$\valtwo$};
					\node (c) at (2*\s,0 + \count*\s) {$\valthree$};
				}
			\end{tikzpicture}
			
			\\[7pt]
			t = 2 &
			
			\begin{tikzpicture}[scale = 0.9, baseline = -2,every node/.style={scale=1}]
				\def\s{0.5};
				\foreach \x in {-5,...,0,5,6,...,12} {
					\draw[fill=black] (\x*\s,0) circle (0.025);
				}
				\foreach \count/\valone/\valtwo/\valthree/\valfour in {0/2/\bon/1/\bon,1/\bon/\bth/1/\btw}{
					\node (a) at (1*\s,0 + \count*\s) {$\valone$};
					\node (d) at (2*\s,0 + \count*\s) {$\valtwo$};
					\node (b) at (3*\s,0 + \count*\s) {$\valthree$};
					\node (c) at (4*\s,0 + \count*\s) {$\valfour$};
				}
			\end{tikzpicture}
			\\[7pt]
			t = 3 &
			
			\begin{tikzpicture}[scale = 0.9, baseline = -2,every node/.style={scale=1}]
				\def\s{0.5};
				\foreach \x in {-5,...,1,7,8,...,12} {
					\draw[fill=black] (\x*\s,0) circle (0.025);
				}
				\foreach \count/\valone/\valtwo/\valthree/\valfour/\valfive in {0/\bon/\bon/2/1/\bon,1/\bth/\btw/\bth/1/\btw}{
					\node (a) at (2*\s,0 + \count*\s) {$\valone$};
					\node (b) at (3*\s,0 + \count*\s) {$\valtwo$};
					\node (c) at (4*\s,0 + \count*\s) {$\valthree$};
					\node (d) at (5*\s,0 + \count*\s) {$\valfour$};
					\node (d) at (6*\s,0 + \count*\s) {$\valfive$};
				}
			\end{tikzpicture}
			\\[7pt]
			t = 4 &
			
			\begin{tikzpicture}[scale = 0.9, baseline = -2,every node/.style={scale=1}]
				\def\s{0.5};
				\foreach \x in {-5,...,2,5,9,10,...,12} {
					\draw[fill=black] (\x*\s,0) circle (0.025);
				}
				\foreach \count/\valone/\valtwo/\valthree/\valfour/\valfive in {0/\bon/\bon/2/1/\bon,1/\bth/\btw/\bth/1/\btw}{
					\node (a) at (3*\s,0 + \count*\s) {$\valone$};
					\node (b) at (4*\s,0 + \count*\s) {$\valtwo$};
					\node (c) at (6*\s,0 + \count*\s) {$\valthree$};
					\node (d) at (7*\s,0 + \count*\s) {$\valfour$};
					\node (d) at (8*\s,0 + \count*\s) {$\valfive$};
				}
			\end{tikzpicture}
			\\[7pt]
			t = 5 &
			\begin{tikzpicture}[scale = 0.9, baseline = -2,every node/.style={scale=1}]
				\def\s{0.5};
				\foreach \x in {-5,...,3,6,7,11,12} {
					\draw[fill=black] (\x*\s,0) circle (0.025);
				}
				\foreach \count/\valone/\valtwo/\valthree/\valfour/\valfive in {0/\overline{1}/\overline{1}/2/1/\overline{1},1/\overline{3}/\overline{2}/\overline{3}/1/\overline{2}}{
					\node (a) at (4*\s,0 + \count*\s) {$\valone$};
					\node (b) at (5*\s,0 + \count*\s) {$\valtwo$};
					\node (c) at (8*\s,0 + \count*\s) {$\valthree$};
					\node (d) at (9*\s,0 + \count*\s) {$\valfour$};
					\node (d) at (10*\s,0 + \count*\s) {$\valfive$};
				}
			\end{tikzpicture}
			\\[7pt]
		\end{array}
		\]
		Calculating $ N_{\ell} $, we find that $ N_1 = 2,\, N_2 = 2 $ and $ N_{\ell} = 0 $ for $ \ell>2 $.
	\end{ex}
	
	\begin{remark}
		States with an arbitrary number of uncoupled solitons can be reduced to multiple collisions of two solitons.
		Moreover, it is a consequence of the Yang--Baxter equation that the states after all collisions have occurred are independent of the order of collisions.
	\end{remark}
	
	\subsection*{Acknowledgements}
	The authors would like to thank Masato Okado for helping them understand coupled solitons through the use of numbers $ N_{\ell} $ (Section~\ref{sec:coupled}).
	The authors are especially grateful to Travis Scrimshaw for all of his invaluable support and guidance.
	The authors would also like to thank the referee for taking their time to read the paper and providing comments.
	
	\appendix
	
	\section{Proof of Theorem~\ref{thm:speed}}\label{appendix:single_soliton_proof}
	Fix a positive integer $ k\le r $.
	Let $ h=\min\{r-k+1,m-r\} $ and $ j \in B $ with $ 1\leq j \leq n $.
	For convenience, we use the following notation:
	\begin{align*}
		\inlinetab{{\mathbf{X}}} &= 
		\begin{array}{|c|}
			\hline
			\bm\\
			\hline
			\vdots\\
			\hline
			\overline{m-k+2}\\
			\hline
		\end{array} \,, &
		\inlinetab{{\mathbf{Z}}} &=
		\begin{array}{|c|}
			\hline
			\overline{m-r}\\
			\hline
			\vdots\\
			\hline
			\overline{m-r-h+1}\\
			\hline
		\end{array} \,, &
		\inlinetab{\mathbf{j}} &= 
		\left.
		\begin{array}{|c|}
			\hline
			j\\
			\hline
			\vdots\\
			\hline
			j\\
			\hline
		\end{array}
		\right\}r-k+1-h \,, \allowdisplaybreaks \\
		\inlinetab{{\mathbf{Y}}} &= 
		\begin{array}{|c|}
			\hline
			\mathbf{X}\\
			\hline
			\mathbf{Z}\\
			\hline
		\end{array} \,, &
		\inlinetab{\mathbf{S}_j} &= 
		\begin{array}{|c|}
			\hline
			\mathbf{X}\\
			\hline
			\mathbf{Z}\\
			\hline
			\mathbf{j}\\
			\hline
		\end{array} = 
		\begin{array}{|c|}
			\hline
			\mathbf{Y}\\
			\hline
			\mathbf{j}\\
			\hline
		\end{array} \,, &
		\inlinetab{\mathbf{C}_j} &= 
		\begin{array}{|c|}
			\hline
			u_1\\
			\hline
			\mathbf{Z}\\
			\hline
			\mathbf{j}\\
			\hline
		\end{array}
		\,.
	\end{align*}
	Note that if $ k=1 $ then $ \inlinetab{\textbf{X}} $ is an empty tableau.
	Similarly, if $ r-k+1-h=0 $ then $ \inlinetab{\textbf{j}} $ is empty.
	
	Since any state can be reached by successively applying crystal operators to a genuine highest weight state,
	by Lemma~\ref{lem:efcom}
	it is sufficient to prove the theorem for genuine highest weight states
	(with respect to the crystal $ \dcrys $).
	
	The genuine highest weight states are of the form
	\[
	u_1^{\otimes c}\otimes \bigotimes_{j =0}^{s-1}\inlinetab{\mathbf{S}_{s-j}}\otimes u_1^{\otimes\infty}.
	\]
	
	\begin{lemma}\label{lem:loadcarrier}
		If $ s-j < \ell $ then
		\[
		R(
		\begin{array}{|c|c|c|c|}
			\hline
			u_{\ell-(s-j )} & \mathbf{S}_{j+1} &\cdots & \mathbf{S}_s\\
			\hline
		\end{array}
		\otimes \inlinetab{\mathbf{S}_j}
		) = u_1\otimes 
		\begin{array}{|c|c|c|c|}
			\hline
			u_{\ell-(s-j )-1} & \mathbf{S}_j &\cdots & \mathbf{S}_s\\
			\hline
		\end{array}\,.
		\]
	\end{lemma}
	\begin{proof}
		First we compute the bumping algorithm for the left-hand side.
		That is, we want to compute:
		\[
		\col(\inlinetab{\mathbf{S}_j}) \rightarrow  	
		\begin{array}{|c|c|c|c|}
			\hline
			u_{\ell - (s-j)}& \mathbf{S}_{j+1} & \cdots &\mathbf{S}_s \\
			\hline
		\end{array}
		.
		\]
		
		Let $ i $ be a letter of $ \col(\inlinetab{\textbf{X}}) $ 
		(so that $ \overline{m-k+2}\le i\le \bm $).
		To insert $ i $,
		observe that all of the columns of the tableau contain 
		$ \inlinetab{i} $
		except the column after $ \mathbf{S}_s $ (possibly empty). 
		Hence, $ i $ will replace itself in all of these columns, leaving them unchanged.
		Because the cells inserted before $ i $ are smaller than $ i $,
		the column after $ \mathbf{S}_s $ can contain  only cells smaller than $ i $,
		Hence, $ i $ gets inserted at the end of this column.
		Thus,
		\[
		\col(\inlinetab{\mathbf{S}_j})
		\rightarrow  	
		\begin{array}{|c|c|c|c|}
			\hline
			u_{\ell - (s-j)}& \mathbf{S}_{j+1} & \cdots &\mathbf{S}_s \\
			\hline
		\end{array}
		=
		\col\left(
		\begin{array}{|c|}
			\hline
			\mathbf{Z}\\
			\hline
			\mathbf{j}\\
			\hline
		\end{array}
		\right)
		\rightarrow  	
		\renewcommand{\arraystretch}{0}
		\begin{array}{|c|c|c|c|c|}
			\hline
			u_{\ell - (s-j)}& \mathbf{S}_{j+1} & \cdots &\mathbf{S}_s & \mathbf{X} 
			\rule{0pt}{2.5ex}\\
			\cline{5-5}
			&&&\\[1ex]
			\cline{1-4}
		\end{array}
		.
		\]
		The remaining items to be inserted are larger than the last cell of the first column
		($ \overline{m-r+1} $),
		and are in ascending order,
		so they all get appended on to the end of the first column:
		\[
		\col(\inlinetab{\mathbf{S}_j})
		\rightarrow  	
		\begin{array}{|c|c|c|c|}
			\hline
			u_{\ell - (s-j)}& \mathbf{S}_{j+1} & \cdots &\mathbf{S}_s \\
			\hline
		\end{array}
		=
		\renewcommand{\arraystretch}{0}
		\begin{array}{|c|c|c|c|c|c|}
			\hline
			\mathbf{C}_j &u_{\ell - (s-j+1)}& \mathbf{S}_{j+1} & \cdots &\mathbf{S}_{s} & \mathbf{X} 
			\rule{0pt}{2.5ex}\\
			\cline{6-6}
			&&&&\\[1ex]
			\cline{2-5}
			\\[1ex]
			\cline{1-1}
		\end{array}
		.
		\]

		Similarly, we compute the bumping algorithm for the right-hand side:
		\begin{align*}
			\col(
			\begin{array}{|c|c|c|c|}
				\hline
				u_{\ell - (s-j)}  & \mathbf{S}_{j} & \cdots & \mathbf{S}_{s}  \\
				\hline
			\end{array}
			)
			\rightarrow
			\begin{array}{|c|}
				\hline
				u_{1}  \\
				\hline
			\end{array}
			&=
			\col(
			\begin{array}{|c|c|c|c|}
				\hline
				u_{\ell - (s-j)}  & \mathbf{S}_{j} & \cdots & \mathbf{S}_{s-1}  \\
				\hline
			\end{array}
			)
			\rightarrow
			\renewcommand{\arraystretch}{0}
			\begin{array}{|c|c|}
				\hline
				\mathbf{C}_s & \mathbf{X} \rule{0pt}{2.5ex}\\[1ex]
				\cline{2-2}
				\\[1ex]
				\cline{1-1}
			\end{array}\\
			&=
			\col(
			\begin{array}{|c|}
				\hline
				u_{\ell-(s-j)}\\
				\hline
			\end{array}
			)
			\rightarrow
			\renewcommand{\arraystretch}{0}
			\begin{array}{|c|c|c|c|c|c|}
				\hline
				\mathbf{C}_{j} & \mathbf{S}_{j +1} & \mathbf{S}_{j +2} & \cdots & \mathbf{S}_s & \mathbf{X}
				\rule{0pt}{2.5ex}\\
				\cline{6-6}
				&&&&\\[1ex]
				\cline{2-5}
				\\[1ex]
				\cline{1-1}
			\end{array}\\
			&=
			\renewcommand{\arraystretch}{0}
			\begin{array}{|c|c|c|c|c|c|}
				\hline
				\mathbf{C}_j &u_{\ell - (s-j+1)}& \mathbf{S}_{j+1} & \cdots &\mathbf{S}_{s} & \mathbf{X} 
				\rule{0pt}{2.5ex}\\
				\cline{6-6}
				&&&&\\[1ex]
				\cline{2-5}
				\\[1ex]
				\cline{1-1}
			\end{array}
			.
		\end{align*}
		Hence,
		\[
		\col(
		\begin{array}{|c|c|c|c|}
			\hline
			u_{\ell - (s-j + 1)} &  \mathbf{S}_{j} & \cdots & \mathbf{S}_{s}  \\
			\hline
		\end{array}
		)
		\rightarrow
		\begin{array}{|c|}
			\hline
			u_{1}  \\
			\hline
		\end{array}
		=
		\col(
		\begin{array}{|c|}
			\hline
			\mathbf{S}_j \\
			\hline
		\end{array}) \rightarrow  	
		\begin{array}{|c|c|c|c|}
			\hline
			u_{\ell - (s-j)}& \mathbf{S}_{j+1} & \cdots &\mathbf{S}_{s} \\
			\hline
		\end{array}
		\]
		as required to prove the lemma.
	\end{proof}
	
	\begin{lemma}\label{lem:unloadcarrier}
		If $ s-j+1 \leq \ell $ then
		\[
		R(
		\begin{array}{|c|c|c|c|}
			\hline
			u_{\ell-(s-j+1)} & \mathbf{S}_j &\cdots & \mathbf{S}_s\\
			\hline
		\end{array}
		\otimes u_1
		) = \inlinetab{\mathbf{S}_s}\otimes 
		\begin{array}{|c|c|c|c|}
			\hline
			u_{\ell-(s-j )} & \mathbf{S}_j &\cdots & \mathbf{S}_{s-1}\\
			\hline
		\end{array}
		.
		\]
	\end{lemma}
	\begin{proof}
		Observe that
		\[	
		\col(
		\begin{array}{|c|}
			\hline
			u_1  \\
			\hline
		\end{array} 
		)
		\rightarrow  
		\begin{array}{|c|c|c|c|}
			\hline
			u_{\ell - (s-j)-1}& \mathbf{S}_j & \cdots &\mathbf{S}_s \\
			\hline
		\end{array}
		=
		\begin{array}{|c|c|c|c|}
			\hline
			u_{\ell-(s-j )} & \mathbf{S}_j &\cdots & \mathbf{S}_s\\
			\hline
		\end{array}
		\]
		and
		\begin{align*}
			\col(
			\begin{array}{|c|c|c|c|}
				\hline
				u_{\ell - (s-j)}& \mathbf{S}_j & \cdots &\mathbf{S}_{s-1} \\
				\hline
			\end{array}
			)
			\rightarrow  	
			\begin{array}{|c|}
				\hline
				\mathbf{S}_{s}  \\
				\hline
			\end{array}
			&=
			\begin{array}{|c|c|c|c|c|}
				\hline
				u_{\ell - (s-j)} & \mathbf{S}_{j} & \cdots & \mathbf{S}_{s-1} &\mathbf{S}_{s} \\
				\hline
			\end{array}
			&&
		\end{align*}
		are equal, as required.
	\end{proof}
	
	\begin{lemma}\label{lem:smallcarrier}
		$
		R(
		\begin{array}{|c|c|c|c|}
			\hline
			\mathbf{S}_{s-\ell+1} &\cdots & \mathbf{S}_s\\
			\hline
		\end{array}
		\otimes \inlinetab{b}
		) = \inlinetab{\mathbf{S}_s}\otimes 
		\begin{array}{|c|c|c|c|}
			\hline
			b & \mathbf{S}_{s-\ell+1} &\cdots & \mathbf{S}_{s-1}\\
			\hline
		\end{array}
		$
		where $ b=\mathbf{S}_{s-\ell} $ or $ b=u_1 $.
	\end{lemma}
	\begin{proof}
		In the process of performing the bumping algorithm for the right-hand side,
		we find
		\[
		\col(
		\begin{array}{|c|c|c|c|}
			\hline
			b&\mathbf{S_{s-\ell + 1}}    & \cdots & \mathbf{S_{s-1}}  \\
			\hline
		\end{array}
		)
		\rightarrow
		\begin{array}{|c|}
			\hline
			\mathbf{S_{s}}  \\
			\hline
		\end{array}
		= 
		\col(
		\begin{array}{|c|}
			\hline
			b  \\
			\hline
		\end{array}
		)
		\rightarrow
		\begin{array}{|c|c|c|c|}
			\hline
			\mathbf{S_{s-\ell + 1}} & \cdots & \mathbf{S}_{s-1}&\mathbf{S}_s  \\
			\hline
		\end{array}
		\]
		as required.
	\end{proof}

	Using Lemmas~\ref{lem:loadcarrier},~\ref{lem:unloadcarrier}~and~\ref{lem:smallcarrier},
	we can easily show that
	\[
	T_{\ell}\left(u_1^{\otimes c}\otimes 
	\bigotimes_{j =0}^{s-1}\inlinetab{\mathbf{S}_{s-j}}
	\otimes u_1^{\otimes \infty}\right)
	=
	u_1^{\otimes c+\min\{s,\ell\}}\otimes 
	\bigotimes_{j =0}^{s-1}\inlinetab{\mathbf{S}_{s-j}}
	\otimes u_1^{\otimes \infty}.
	\]
	Using Lemma~\ref{lem:efcom},
	this is sufficient to prove Theorem~\ref{thm:speed} for general states.
	
	\section{Proof of Theorem~\ref{thm:scattering}} \label{sec:mainthmproof}
	\subsection{Reductions}\label{sec:simplifications}
	Consider $ \agl(m|n) $ with height $ r $ tableaux, and assume $ m\ge r+2 $.
	Set,
	\begin{align*}
		\inlinetab{\bfX}
		&=
		\begin{array}{|c|}
			\hline
			\raisebox{-1pt}{$ \bm $}\\
			\hline
			\vdots \\
			\hline
			\raisebox{-1pt}{$ \overline{m-r+3} $}\\
			\hline
		\end{array},
		&
		\inlinetab{\bfth} &= \inlinetab{\overline{m-r+2}},
		&
		\inlinetab{\bftw} &= \inlinetab{\overline{m-r+1}},\\
		\inlinetab{\bfon} &= \inlinetab{\overline{m-r}},
		&
		\inlinetab{\bfze} &= \inlinetab{\overline{m-r-1}}.
	\end{align*}
	Note: $ \bfth < \bftw < \bfon < \bfze $.
	
	We employ the following notation:
	\[
	\begin{array}{|c|}
		\hline
		x^a\\
		\hline
	\end{array}
	=
	\underbrace{
		\begin{array}{|c|c|c|c|}
			\hline
			x & x & \cdots  & x\\
			\hline
		\end{array}
	}_{a}
	\qquad \text{and} \qquad
	\begin{array}{|c|}
		\hline
		(x_j)_{j=1}^a\\
		\hline
	\end{array}
	=
	\begin{array}{|c|c|c|c|}
		\hline
		x_1 & x_2 & \cdots  & x_a\\
		\hline
	\end{array}.
	\]
	
	To prove Theorem~\ref{thm:scattering} in the case where $ m\geq r+2 $, it suffices to prove the following proposition:
	\begin{prop}\label{prop:simplescattering}
		Let $ \mathfrak{p} $ be a genuine highest weight state in the $ U_1(\gl(r))\oplus U_q(\gl(m-r|m)) $-crystal $ \dcrys $. That is,
		\begin{equation}\label{eq:hwtstate}
			\mathfrak{p} =
			u_1^{\otimes \mathfrak{c}_1}\otimes 
			\begin{array}{|c|}
				\hline
				\raisebox{-1pt}{$ \bfX $}\\
				\raisebox{-1pt}{$ \bfth $}\\
				\raisebox{-1pt}{$ \bfon $}\\
				\hline
			\end{array}
			^{\otimes \mathfrak{d}_1}
			\otimes u_1^{\otimes \mathfrak{c}_2}\otimes 
			\begin{array}{|c|}
				\hline
				\raisebox{-1pt}{$ \bfX $}\\
				\raisebox{-1pt}{$ \mathfrak{w}_1 $}\\
				\raisebox{-1pt}{$ \mathfrak{y}_1 $}\\
				\hline
			\end{array}
			\otimes \cdots \otimes 
			\begin{array}{|c|}
				\hline
				\raisebox{-1pt}{$ \bfX $}\\
				\raisebox{-1pt}{$ \mathfrak{w}_{\mathfrak{d}_2} $}\\
				\raisebox{-1pt}{$ \mathfrak{y}_{\mathfrak{d}_2} $}\\
				\hline
			\end{array}
			\otimes u_1^{\otimes \infty }
		\end{equation}
		where
		\begin{align*}
			\mathfrak{w}_j = 
			\begin{cases}
				\bftw & \text{if } j\le L\\
				\bfth & \text{if } j>L
			\end{cases},
			&&
			\mathfrak{y}_j = 
			\begin{cases}
				\bfze & \text{if }j\le M\\
				\bfon & \text{if }j>M
			\end{cases},
		\end{align*}
		for some non-negative integers $ L,M \le \mathfrak{d}_2 $.
		If $ \mathfrak{c}_2+L\ge\mathfrak{d}_2-M $ and $ t\in\ZZ $ is sufficiently large,
		then
		\[
		(T_{\mathfrak{d}_2+1})^t(\mathfrak{p}) = 
		u_1^{\otimes \mathfrak{c}_3}\otimes 
		\begin{array}{|c|}
			\hline
			\raisebox{-1pt}{$ \bfX $}\\
			\raisebox{-1pt}{$ \bfth $}\\
			\raisebox{-1pt}{$ \bfon $}\\
			\hline
		\end{array}
		^{\otimes \mathfrak{d}_2}
		\otimes u_1^{\otimes \mathfrak{c}_4}\otimes 
		\begin{array}{|c|}
			\hline
			\raisebox{-1pt}{$ \bfX $}\\
			\raisebox{-1pt}{$ \mathfrak{w}_1 $}\\
			\raisebox{-1pt}{$ \mathfrak{y}_1 $}\\
			\hline
		\end{array}
		\otimes \cdots \otimes 
		\begin{array}{|c|}
			\hline
			\raisebox{-1pt}{$ \bfX $}\\
			\raisebox{-1pt}{$ \mathfrak{w}_{\mathfrak{d}_2} $}\\
			\raisebox{-1pt}{$ \mathfrak{y}_{\mathfrak{d}_2} $}\\
			\hline
		\end{array}
		\otimes
		\begin{array}{|c|}
			\hline
			\raisebox{-1pt}{$ \bfX $}\\
			\raisebox{-1pt}{$ \bfth $}\\
			\raisebox{-1pt}{$ \bfon $}\\
			\hline
		\end{array}
		^{\otimes \mathfrak{d}_1-\mathfrak{d}_2}
		\otimes u_1^{\otimes \infty }
		\]
		with phase shift $ \delta = 2\mathfrak{d}_2-L-M $.
	\end{prop}
	
	Indeed, by Lemma~\ref{lem:efcom}, certain crystal operators
	commute with the time evolution operator $ T_\ell $.
	Since any state can be reached by successively applying crystal operators to a genuine highest weight state~\cite{BKK00},
	it is sufficient to prove Theorem~\ref{thm:scattering} for genuine highest weight states of $ \dcrys $.
	
	All genuine highest weight states are given by~\ref{eq:hwtstate}.
	Indeed,
	the vacuum elements are invariant under the crystal operators of $ \dcrys $,
	so can be safely be ignored.
	Additionally,
	given the assumption $ k=r $ in Theorem~\ref{thm:scattering},
	we can treat the bottom row and the other rows as two seperate crystals,
	(which are a $ U_q(\gl(r)) $-crystal and a $ U_q(\gl(m-r|n)) $-crystal respectively).
	Finally, due to the semistandard assumption of each soliton,
	the bottom row is isomorphic to the $ U_q(\gl(r)) $-crystal
	$ B^{1,\mathfrak{d}_1}\otimes B^{1,\mathfrak{d}_2} $
	and the other rows are isomorphic to the $ U_q(\gl(m-r|n)) $-crystal
	$ B^{(r-1),\mathfrak{d}_1}\otimes B^{(r-1),\mathfrak{d}_2} $ 
	(the isomorphisms are derived from Schensted's bumping algorithm).
	We know the general form for genuine highest weight elements for these crystals~\cite[Proposition 7.1]{KO18}.
	
	If we can prove 
	$ \up{\widetilde{W}}\otimes \up{\widetilde{V}} = R(\up{V}\otimes \up{W}) $ and 
	$ \down{\widetilde{W}}\otimes \down{\widetilde{V}} = R(\down{V}\otimes \down{W}) $
	for the genuine highest weight states (using the notation from Theorem~\ref{thm:scattering}) then 
	it will also hold for a general state,
	since the crystal operators commute with the $ R $-matrix.
	Note that we can explicitly calculate $ R(\up{V}\otimes\up{W}) $ and $ R(\down{V}\otimes\down{W}) $
	for the genuine highest weight states,
	and we do so in Lemma~\ref{lem:hwtR}.
	Then, by inspection, we see that 
	$ \up{\widetilde{W}}\otimes \up{\widetilde{V}} = R(\up{V}\otimes \up{W}) $ and 
	$ \down{\widetilde{W}}\otimes \down{\widetilde{V}} = R(\down{V}\otimes \down{W}) $
	in Proposition~\ref{prop:simplescattering}.
	\begin{lemma}\label{lem:hwtR}
		\begin{align*}
			R\left(
			\begin{array}{|c|}
				\hline
				\raisebox{-1pt}{$ \bfX^{\mathfrak{d}_1} $}\\
				\raisebox{-1pt}{$ \bfth^{\mathfrak{d}_1} $}\\
				\hline
			\end{array}
			\otimes
			\begin{array}{|c|}
				\hline
				\raisebox{-1pt}{$ \bfX^{\mathfrak{d}_2} $}\\
				\raisebox{-1pt}{$ (\mathfrak{w}_{\mathfrak{d}_2-j})_{j=0}^{\mathfrak{d}_2-1} $}\\
				\hline
			\end{array}
			\right)
			&=
			\begin{array}{|c|}
				\hline
				\raisebox{-1pt}{$ \bfX^{\mathfrak{d}_2} $}\\
				\raisebox{-1pt}{$ \bfth^{\mathfrak{d}_2} $}\\
				\hline
			\end{array}
			\otimes
			\begin{array}{|c|}
				\hline
				\raisebox{-1pt}{$ \bfX^{\mathfrak{d}_1} $}\\
				\raisebox{-1pt}{$ \bfth^{\mathfrak{d}_1-\mathfrak{d}_2}(\mathfrak{w}_{\mathfrak{d}_2-j})_{j=0}^{\mathfrak{d}_2-1} $}\\
				\hline
			\end{array}\,,
			\\
			R\left(
			\begin{array}{|c|}
				\hline
				\raisebox{-1pt}{$ \bfon ^{\mathfrak{d}_1} $}\\
				\hline
			\end{array}
			\otimes
			\begin{array}{|c|}
				\hline
				\raisebox{-1pt}{$ (\mathfrak{y}_{\mathfrak{d}_2-j})_{j=0}^{\mathfrak{d}_2-1} $}\\
				\hline
			\end{array}
			\right)
			&=
			\begin{array}{|c|}
				\hline
				\raisebox{-1pt}{$ \bfon ^{\mathfrak{d}_2} $}\\
				\hline
			\end{array}
			\otimes
			\begin{array}{|c|}
				\hline
				\raisebox{-1pt}{$ \bfon^{\mathfrak{d}_1-\mathfrak{d}_2}(\mathfrak{y}_{\mathfrak{d}_2-j})_{j=0}^{\mathfrak{d}_2-1} $}\\
				\hline
			\end{array}\,.
		\end{align*}
	\end{lemma}
	\begin{proof}
		\begin{align*}
			\col\left(
			\begin{array}{|c|}
				\hline
				\raisebox{-1pt}{$ \bfX^{\mathfrak{d}_2} $}\\
				\raisebox{-1pt}{$ (\mathfrak{w}_{\mathfrak{d}_2-j})_{j=0}^{\mathfrak{d}_2-1} $}\\
				\hline
			\end{array}
			\right) \rightarrow
			\begin{array}{|c|}
				\hline
				\raisebox{-1pt}{$ \bfX^{\mathfrak{d}_1} $}\\
				\raisebox{-1pt}{$ \bfth^{\mathfrak{d}_1} $}\\
				\hline
			\end{array}
			&=
			\col\left(
			\begin{array}{|c|}
				\hline
				\raisebox{-1pt}{$ \bfX^{\mathfrak{d}_2} $}\\
				\bfth^{\mathfrak{d}_2-L}\bftw^{L}\\
				\hline
			\end{array}
			\right) \rightarrow
			\begin{array}{|c|}
				\hline
				\raisebox{-1pt}{$ \bfX^{\mathfrak{d}_1} $}\\
				\raisebox{-1pt}{$ \bfth^{\mathfrak{d}_1} $}\\
				\hline
			\end{array} \\
			&=
			\col\left(
			\begin{array}{|c|}
				\hline
				\raisebox{-1pt}{$ \bfX^{\mathfrak{d}_2-L} $}\\
				\bfth^{\mathfrak{d}_2-L}\\
				\hline
			\end{array}
			\right) \rightarrow
			\begin{array}{|c|cc}
				\hline
				\multicolumn{3}{|c|}{\bfX^{\mathfrak{d}_1+L}}\\
				\cline{3-3}
				\multicolumn{2}{|c|}{\bfth^{\mathfrak{d}_1} } &\\
				\cline{2-2}
				\bftw^L & \\
				\cline{1-1}
			\end{array} \\
			&=
			\begin{array}{|c|c|c}
				\hline
				\multicolumn{3}{|c|}{\bfX^{\mathfrak{d}_1+\mathfrak{d}_2}}\\
				\cline{3-3}
				\multicolumn{2}{|c|}{\bfth^{\mathfrak{d}_1+\mathfrak{d}_2-L}} &\\
				\cline{2-2}
				\bftw^L\\
				\cline{1-1}
			\end{array} \allowdisplaybreaks \\
			\col\left(
			\begin{array}{|c|}
				\hline
				\raisebox{-1pt}{$ \bfX^{\mathfrak{d}_1} $}\\
				\raisebox{-1pt}{$ \bfth^{\mathfrak{d}_1-\mathfrak{d}_2}(\mathfrak{w}_{\mathfrak{d}_2-j})_{j=0}^{\mathfrak{d}_2-1} $}\\
				\hline
			\end{array}
			\right) \rightarrow
			\begin{array}{|c|}
				\hline
				\raisebox{-1pt}{$ \bfX^{\mathfrak{d}_2} $}\\
				\raisebox{-1pt}{$ \bfth^{\mathfrak{d}_2} $}\\
				\hline
			\end{array}
			&=
			\col\left(
			\begin{array}{|c|}
				\hline
				\raisebox{-1pt}{$ \bfX^{\mathfrak{d}_1} $}\\
				\bfth^{\mathfrak{d}_1-L}\bftw^{L}\\
				\hline
			\end{array}
			\right) \rightarrow
			\begin{array}{|c|}
				\hline
				\raisebox{-1pt}{$ \bfX^{\mathfrak{d}_2} $}\\
				\raisebox{-1pt}{$ \bfth^{\mathfrak{d}_2} $}\\
				\hline
			\end{array} \\
			&=
			\begin{array}{|c|c|c}
				\hline
				\multicolumn{3}{|c|}{\bfX^{\mathfrak{d}_1+\mathfrak{d}_2}}\\
				\cline{3-3}
				\multicolumn{2}{|c|}{\bfth^{\mathfrak{d}_1+\mathfrak{d}_2-L}} &\\
				\cline{2-2}
				\bftw^L\\
				\cline{1-1}
			\end{array}
		\end{align*}
		The above two insertions are identical, proving the first part of the lemma.
		\begin{align*}
			\col\left(
			\begin{array}{|c|}
				\hline
				\raisebox{-1pt}{$ (\mathfrak{y}_{\mathfrak{d}_2-j})_{j=0}^{\mathfrak{d}_2-1} $}\\
				\hline
			\end{array}
			\right) \rightarrow
			\begin{array}{|c|}
				\hline
				\raisebox{-1pt}{$ \bfon ^{\mathfrak{d}_1} $}\\
				\hline
			\end{array}
			&=
			\col\left(
			\begin{array}{|c|}
				\hline
				\raisebox{-1pt}{$ \bfon^{\mathfrak{d}_2-M}\bfze^M $}\\
				\hline
			\end{array}
			\right) \rightarrow
			\begin{array}{|c|}
				\hline
				\raisebox{-1pt}{$ \bfon ^{\mathfrak{d}_1} $}\\
				\hline
			\end{array} \\
			&=
			\col\left(
			\begin{array}{|c|}
				\hline
				\raisebox{-1pt}{$ \bfon^{\mathfrak{d}_2-M} $}\\
				\hline
			\end{array}
			\right) \rightarrow
			\begin{array}{|c|c}
				\hline
				\multicolumn{2}{|c|}{\bfon ^{\mathfrak{d}_1}}\\
				\cline{2-2}
				\bfze^M & \\
				\cline{1-1}
			\end{array} \\
			&=
			\begin{array}{|c|c}
				\hline
				\multicolumn{2}{|c|}{\bfon ^{\mathfrak{d}_1+\mathfrak{d}_2-M}}\\
				\cline{2-2}
				\bfze^M & \\
				\cline{1-1}
			\end{array} \allowdisplaybreaks \\
			\col\left(
			\begin{array}{|c|}
				\hline
				\raisebox{-1pt}{$ \bfon^{\mathfrak{d}_1-\mathfrak{d}_2}(\mathfrak{y}_{\mathfrak{d}_2-j})_{j=0}^{\mathfrak{d}_2-1} $}\\
				\hline
			\end{array}
			\right) \rightarrow
			\begin{array}{|c|}
				\hline
				\raisebox{-1pt}{$ \bfon ^{\mathfrak{d}_2} $}\\
				\hline
			\end{array}
			&=
			\col\left(
			\begin{array}{|c|}
				\hline
				\raisebox{-1pt}{$ \bfon^{\mathfrak{d}_1-M}\bfze^M $}\\
				\hline
			\end{array}
			\right) \rightarrow
			\begin{array}{|c|}
				\hline
				\raisebox{-1pt}{$ \bfon ^{\mathfrak{d}_2} $}\\
				\hline
			\end{array} \\
			&=
			\begin{array}{|c|c}
				\hline
				\multicolumn{2}{|c|}{\bfon ^{\mathfrak{d}_1+\mathfrak{d}_2-M}}\\
				\cline{2-2}
				\bfze^M & \\
				\cline{1-1}
			\end{array}
		\end{align*}
		The above two insertions are identical, completing the proof.
	\end{proof}
	
	Note that applying $ e_i $ or $ f_i $ 
	(with $ i\in I\setminus\{\bze,\overline{m-r}\} $)
	to a state will not change the positions of the solitons.
	Note also that $ H(e_i(x\otimes y)) = H(x\otimes y) $ 
	and $ H(f_i(x\otimes y)) = H(x\otimes y) $ by definition.
	Therefore, if we prove that the phase shift of any genuine highest weight state
	is $ \delta = 2\mathfrak{d}_2+H(\down{V}\otimes \down{W})+H(\up{V}\otimes \up{W}) $,
	then this formula must also hold in general.
	We can compute the values of $ H(\down{V}\otimes \down{W}) $ and $ H(\up{V}\otimes \up{W}) $
	for a genuine highest weight state explicitly, and do so in Lemma~\ref{lem:hwtphase}.
	It then suffices to show that the phase shift for a highest weight state is $ 2\mathfrak{d}_2 - L - M $.
	
	\begin{lemma}\label{lem:hwtphase}
		\[
		H\left(
		\begin{array}{|c|}
			\hline
			\raisebox{-1pt}{$ \bfX^{\mathfrak{d}_1} $}\\
			\raisebox{-1pt}{$ \bfth^{\mathfrak{d}_1} $}\\
			\hline
		\end{array}
		\otimes
		\begin{array}{|c|}
			\hline
			\raisebox{-1pt}{$ \bfX^{\mathfrak{d}_2} $}\\
			\raisebox{-1pt}{$ (\mathfrak{w}_{\mathfrak{d}_2-j})_{j=0}^{\mathfrak{d}_2-1} $}\\
			\hline
		\end{array}
		\right)
		=-L,
		\qquad
		H\left(
		\begin{array}{|c|}
			\hline
			\raisebox{-1pt}{$ \bfon ^{\mathfrak{d}_1} $}\\
			\hline
		\end{array}
		\otimes
		\begin{array}{|c|}
			\hline
			\raisebox{-1pt}{$ (\mathfrak{y}_{\mathfrak{d}_2-j})_{j=0}^{\mathfrak{d}_2-1} $}\\
			\hline
		\end{array}
		\right)
		=-M.
		\]
	\end{lemma}
	\begin{proof}
		From the proof of Lemma~\ref{lem:hwtR}, we know
		\[
		\col\left(
		\begin{array}{|c|}
			\hline
			\raisebox{-1pt}{$ \bfX^{\mathfrak{d}_2} $}\\
			\raisebox{-1pt}{$ (\mathfrak{w}_{\mathfrak{d}_2-j})_{j=0}^{\mathfrak{d}_2-1} $}\\
			\hline
		\end{array}
		\right) \rightarrow
		\begin{array}{|c|}
			\hline
			\raisebox{-1pt}{$ \bfX^{\mathfrak{d}_1} $}\\
			\raisebox{-1pt}{$ \bfth^{\mathfrak{d}_1} $}\\
			\hline
		\end{array}
		=
		\begin{array}{|c|c|c}
			\hline
			\multicolumn{3}{|c|}{\bfX^{\mathfrak{d}_1+\mathfrak{d}_2}}\\
			\cline{3-3}
			\multicolumn{2}{|c|}{\bfth^{\mathfrak{d}_1+\mathfrak{d}_2-L}} &\\
			\cline{2-2}
			\bftw^L\\
			\cline{1-1}
		\end{array}
		.
		\]
		There are $ \mathfrak{d}_2(r-1)-L $ boxes to the right of the 
		$ \max\{\mathfrak{d}_1,\mathfrak{d}_2\}=\mathfrak{d}_1 $ column.
		By definition, we have
		\[
		H\left(
		\begin{array}{|c|}
			\hline
			\raisebox{-1pt}{$ \bfX^{\mathfrak{d}_1} $}\\
			\raisebox{-1pt}{$ \bfth^{\mathfrak{d}_1} $}\\
			\hline
		\end{array}
		\otimes
		\begin{array}{|c|}
			\hline
			\raisebox{-1pt}{$ \bfX^{\mathfrak{d}_2} $}\\
			\raisebox{-1pt}{$ (\mathfrak{w}_{\mathfrak{d}_2-j})_{j=0}^{\mathfrak{d}_2-1} $}\\
			\hline
		\end{array}
		\right)
		=(\mathfrak{d}_2(r-1)-L)-\mathfrak{d}_2(r-1)=-L
		\]
		as required.
		
		From the proof of Lemma~\ref{lem:hwtR}, we know
		\[
		\col\left(
		\begin{array}{|c|}
			\hline
			\raisebox{-1pt}{$ (\mathfrak{y}_{\mathfrak{d}_2-j})_{j=0}^{\mathfrak{d}_2-1} $}\\
			\hline
		\end{array}
		\right) \rightarrow
		\begin{array}{|c|}
			\hline
			\raisebox{-1pt}{$ \bfon ^{\mathfrak{d}_1} $}\\
			\hline
		\end{array}
		=
		\begin{array}{|c|c}
			\hline
			\multicolumn{2}{|c|}{\bfon ^{\mathfrak{d}_1+\mathfrak{d}_2-M}}\\
			\cline{2-2}
			\bfze^M & \\
			\cline{1-1}
		\end{array}
		.
		\]
		There are $ \mathfrak{d}_2-M $ boxes to the right of the $ \max\{\mathfrak{d}_1,\mathfrak{d}_2\}=\mathfrak{d}_1 $ column.
		By definition, we have
		\[
		H\left(
		\begin{array}{|c|}
			\hline
			\raisebox{-1pt}{$ \bfon ^{\mathfrak{d}_1} $}\\
			\hline
		\end{array}
		\otimes
		\begin{array}{|c|}
			\hline
			\raisebox{-1pt}{$ (\mathfrak{y}_{\mathfrak{d}_2-j})_{j=0}^{\mathfrak{d}_2-1} $}\\
			\hline
		\end{array}
		\right)
		=(\mathfrak{d}_2-M)-\mathfrak{d}_2\cdot1
		=-M
		\]
		as required.
	\end{proof}
	
	Finally, we have that $ (T_\ell)^t = (T_{\mathfrak{d}_2+1})^{-t'}(T_{\ell})^t(T_{\mathfrak{d}_2+1})^{t'} $
	since time evolution operators commute by Proposition~\ref{prop:Tlcomm},
	(note that $ (T_{\mathfrak{d}_2+1})^{-1} $ exists; see Remark~\ref{rmk:timereversible}).
	Therefore, if we prove the theorem for $ T_{\mathfrak{d}_2+1} $ and choose $ t^\prime $ sufficiently large
	(so that the solitons have finished colliding after applying $ (T_{\mathfrak{d}_2+1})^{t'} $)
	and choose $ t>t' $ 
	(so that $ (T_{\mathfrak{d}_2+1})^{-t'}(T_{\ell})^t $ does not undo the collision;
	this follows from Theorem~\ref{thm:speed}),
	then we can prove the theorem in the general case.

	To prove Proposition~\ref{prop:simplescattering}, we could use direct (and tedious) computation.
	However, 
	in the rest of this section we give an alternate proof,
	in which we reduce the behaviour of the general system to the behaviour of a modified height $1$ system.
	
	\subsection{Behaviour of the \texorpdfstring{$ R $}{R}-matrix}\label{sec:Rmatrixbehaviour}
	It is first useful to classify how the $ R $-matrix acts on common tableaux.
	\begin{lemma}\label{lem:classRmat}
		Let $ \ell $ be a positive integer.
		Let $ (\mathbf{h}_j)_{j=1}^\ell=\bfth^{a_3}\bftw^{a_2}\bfon^{a_1} $ and $ (\mathbf{z}_j)_{j=1}^\ell = \bftw^{b_2}\bfon^{b_1}\bfze^{b_0} $ 
		for some positive integers $ a_3,a_2,a_1,b_2,b_1,b_0 $ such that $ a_3+a_2+a_1=\ell $ and $ b_2+b_1+b_0=\ell $.
		
		For
		\[
		b=\begin{array}{|c|}
			\hline
			\bfX^\ell\\
			(\mathbf{h}_j)_{j=1}^\ell\\
			(\mathbf{z}_j)_{j=1}^\ell\\
			\hline
		\end{array}
		\otimes 
		\begin{array}{|c|}
			\hline
			\bfX\\
			\bfon \\
			\mathbf{y}\\
			\hline
		\end{array}
		\]
		with $ \mathbf{y}=\bfze $,
		we have that $ R(b)= $
		\begin{enumerate}[label = \upshape{($ \bfon $\alph*)}]
			\item \label{enum:q}
			$
			\begin{array}{|c|}
				\hline
				\bfX\\
				\bfth\\
				\bftw\\
				\hline
			\end{array}
			\otimes 
			\begin{array}{|c|}
				\hline
				\bfX^\ell\\
				\bfth^{a_3-1}\bftw^{a_2}\bfon^{a_1+1}\\
				\bftw^{b_2-1}\bfon^{b_1} \mathbf{y}\bfze^{b_0}\\
				\hline
			\end{array}
			$
			if $ b_2>0 $;
			\item \label{enum:s}
			$
			\begin{array}{|c|}
				\hline
				\bfX\\
				\mathbf{h}_{a_3+a_2}\\
				\bfon\\
				\hline
			\end{array}
			\otimes 
			\begin{array}{|c|}
				\hline
				\bfX^\ell\\
				(\mathbf{h}_j)_{j=1}^{a_3+a_2-1}\bfon^{a_1+1}\\
				\bfon^{b_1-1} \mathbf{y} \bfze^{b_0}\\
				\hline
			\end{array}
			$
			if $ b_2=0 $ and $ b_2+b_1=a_3+a_2 $ and $ b_1>0 $.
		\end{enumerate}
		
		For
		\[
		b=
		\begin{array}{|c|}
			\hline
			\bfX^\ell\\
			(\mathbf{h}_j)_{j=1}^\ell\\
			(\mathbf{z}_j)_{j=1}^\ell\\
			\hline
		\end{array}
		\otimes 
		\begin{array}{|c|}
			\hline
			\bfX\\
			\bftw \\
			\mathbf{y}\\
			\hline
		\end{array}
		\]
		with $ \mathbf{y}\in\{\bfze,\bfon\} $, we have that $ R(b)= $
		\begin{enumerate}[label=\upshape{($ \bftw $\alph*)}]
			\item \label{enum:k}
			$
			\begin{array}{|c|}
				\hline
				\bfX\\
				\mathbf{h}_{a_3+a_2}\\
				\bfon\\
				\hline
			\end{array}
			\otimes 
			\begin{array}{|c|}
				\hline
				\bfX^\ell\\
				(\mathbf{h}_j)_{j=1}^{a_3+a_2-1}\bfon^{a_1+1}\\
				\bftw^{b_2+1}\bfon^{b_1-2}\mathbf{y} \bfze^{b_0}\\
				\hline
			\end{array}
			$
			if $ b_2<a_3 $ and $ b_2+b_1=a_3+a_2 $ and $ b_1>1 $ and $ \mathbf{y}=\bfze $;
			\item \label{enum:l}
			$
			\begin{array}{|c|}
				\hline
				\bfX\\
				\bfth\\
				\bfon\\
				\hline
			\end{array}
			\otimes 
			\begin{array}{|c|}
				\hline
				\bfX^\ell\\
				\bfth^{a_3-1}\bftw^{a_2+1}\bfon^{a_1}\\
				\bftw^{b_2}\bfon^{b_1-1}\mathbf{y} \bfze^{b_0}\\
				\hline
			\end{array}
			$
			if $ b_2<a_3 $ and $ b_2+b_1=a_3+a_2 $ and ($ b_1=1 $ or $ \mathbf{y}=\bfon $).
			\item \label{enum:e}
			$
			\begin{array}{|c|}
				\hline
				\bfX\\
				\bfth\\
				\bfon\\
				\hline
			\end{array}
			\otimes 
			\begin{array}{|c|}
				\hline
				\bfX^\ell\\
				\bfth^{a_3-1}\bftw^{a_2+2}\bfon^{a_1-1}\\
				\bftw^{b_2-1}\bfon^{b_1} \mathbf{y} \bfze^{b_0}\\
				\hline
			\end{array}
			$
			if $ 0<b_2=a_3 $ and $ a_1>0 $ and ($ b_1=0 $ or $ \mathbf{y}=\bfon $);
			\item \label{enum:f}
			$
			\begin{array}{|c|}
				\hline
				\bfX\\
				\bfth\\
				\bftw\\
				\hline
			\end{array}
			\otimes 
			\begin{array}{|c|}
				\hline
				\bfX^\ell\\
				\bfth^{a_3-1}\bftw^{a_2+1}\\
				\bftw^{b_2-1}\bfon^{b_1} \mathbf{y} \bfze^{b_0} \\
				\hline
			\end{array}
			$
			if $ 0<b_2=a_3 $ and $ a_1=0 $ and ($ b_1=0 $ or $ \mathbf{y}=\bfon $).
		\end{enumerate}
		
		For
		\[
		b=
		\begin{array}{|c|}
			\hline
			\bfX^\ell\\
			(\mathbf{h}_j)_{j=1}^\ell\\
			(\mathbf{z}_j)_{j=1}^\ell\\
			\hline
		\end{array}
		\otimes 
		\begin{array}{|c|}
			\hline
			\bfX\\
			\bfth \\
			\mathbf{y}\\
			\hline
		\end{array}
		\]
		with $ \mathbf{y}\in\{\bfze,\bfon,\bftw\} $ we have that $ R(b)= $
		\begin{enumerate}[label=\upshape{($ \bfth $\alph*)}]
			\item \label{enum:m}
			$
			\begin{array}{|c|}
				\hline
				\bfX\\
				\mathbf{h}_{a_3+a_2}\\
				\bfon\\
				\hline
			\end{array}
			\otimes 
			\begin{array}{|c|}
				\hline
				\bfX^\ell\\
				\bfth(\mathbf{h}_j)_{j=1}^{a_3+a_2-1}\bfon^{a_1}\\
				\bftw^{b_2}\bfon^{b_1-1}\mathbf{y}\bfze^{b_0}\\
				\hline
			\end{array}
			$
			if $ b_1>0 $ and $ \mathbf{y}=\bfze $;
			\item \label{enum:n}
			$
			\begin{array}{|c|}
				\hline
				\bfX\\
				\bfth\\
				\bfon\\
				\hline
			\end{array}
			\otimes 
			\begin{array}{|c|}
				\hline
				\bfX^\ell\\
				\bfth^{a_3}\bftw^{a_2+1}\bfon^{a_1-1}\\
				\bftw^{b_2-1}\bfon^{b_1} \mathbf{y}\bfze^{b_0}\\
				\hline
			\end{array}
			$
			if $ b_2>0 $ and $ a_1>0 $ and $ \mathbf{y}\ge \bfon $ and ($ b_1=0 $ or $ \mathbf{y}=\bfon $);
			\item \label{enum:o}
			$
			\begin{array}{|c|}
				\hline
				\bfX\\
				\bfth\\
				\bftw\\
				\hline
			\end{array}
			\otimes 
			\begin{array}{|c|}
				\hline
				\bfX^\ell\\
				\bfth^{a_3}\bftw^{a_2}\\
				\bftw^{b_2-1}\bfon^{b_1} \mathbf{y} \bfze^{b_0}\\
				\hline
			\end{array}
			$
			if $ b_2>0 $ and $ a_1=0 $ and $ \mathbf{y}\ge \bfon $ and ($ b_1=0 $ or $ \mathbf{y}=\bfon $);
			\item \label{enum:p}
			$
			\begin{array}{|c|}
				\hline
				\bfX\\
				\mathbf{h}_\ell\\
				\mathbf{z}_\ell\\
				\hline
			\end{array}
			\otimes 
			\begin{array}{|c|}
				\hline
				\bfX^\ell\\
				\bfth(\mathbf{h}_j)_{j=1}^{\ell-1}\\
				\mathbf{y}(\mathbf{z}_j)_{j=1}^{\ell-1}\\
				\hline
			\end{array}
			$
			if $ \mathbf{y}=\bftw $ or ($ b_2=0 $ and ($ b_1=0 $ or $ \mathbf{y}=\bfon $)).
		\end{enumerate}
	\end{lemma}
	
	The proof of the lemma follows by a direct application of Schensted's bumping algorithm.
	See Appendix~\ref{appendix:classRmatprf} for details.
	
	\begin{remark}
		The above lemma only lists the cases required to prove the main theorem
		and does not completely classify the action of the $ R $-matrix.
	\end{remark}
	
	\begin{remark}\label{rmk:omitX}
		In Lemma~\ref{lem:classRmat}, $ \bfX $ does not affect the action of $ R $-matrix on genuine highest weight states.
		Therefore, it is sufficient to only consider height $2$ states.
	\end{remark}

	Note that many cases of the lemma can be informally described using the height $ 1 $ $ R $-matrix:
	\begin{ex}
		Assuming $ b_1>1 $,
		$ R(
		\begin{array}{|c|}
			\hline
			\bftw^{b_2} \bfon^{b_1} \bfze^{b_0}\\
			\hline
		\end{array}
		\otimes
		\begin{array}{|c|}
			\hline
			\bfze\\
			\hline
		\end{array}
		)
		=
		\begin{array}{|c|}
			\hline
			\bfon\\
			\hline
		\end{array}
		\otimes
		\begin{array}{|c|}
			\hline
			\bftw^{b_2} \bfon^{b_1-1} \bfze^{b_0+1}\\
			\hline
		\end{array}
		$.
		Compare this with the bottom row of~\ref{enum:k}:
		$
		\begin{array}{|c|}
			\hline
			\bfon\\
			\hline
		\end{array}
		\otimes
		\begin{array}{|c|}
			\hline
			\bftw^{b_2+1} \bfon^{b_1-2}\bfze^{b_0+1}\\
			\hline
		\end{array}
		$.
		Additionally, note that the second last row has an extra $ \bfon $.
		One interpretation of this is that when the carrier picks up the $ \bftw $ 
		in the second last row, it gets swapped with a $ \bfon $ in the bottom row.
	\end{ex}
	
	\begin{ex}
		Assuming $ b_2>0 $, $ \mathbf{y}\in\{\bfon,\bfze\} $ and ($ b_1=0 $ or $ \mathbf{y}=\bfon $),
		$ R(
		\begin{array}{|c|}
			\hline
			\bftw^{b_2} \bfon^{b_1} \bfze^{b_0}\\
			\hline
		\end{array}
		\otimes
		\begin{array}{|c|}
			\hline
			\mathbf{y}\\
			\hline
		\end{array}
		) = 
		\begin{array}{|c|}
			\hline
			\bftw\\
			\hline
		\end{array}
		\otimes
		\begin{array}{|c|}
			\hline
			\bftw^{b_2-1} \bfon^{b_1} \mathbf{y} \bfze^{b_0}\\
			\hline
		\end{array}
		$.
		Compare this with the bottom row of~\ref{enum:n}:
		$ 
		\begin{array}{|c|}
			\hline
			\bfon\\
			\hline
		\end{array}
		\otimes
		\begin{array}{|c|}
			\hline
			\bftw^{b_2-1}\bfon^{b_1} \mathbf{y} \bfze^{b_0}\\
			\hline
		\end{array}
		$.
		Additionally, note that the second last row has an extra $ \bftw $,
		and has lost a $ \bfon $.
		One interpretation of this is that the carrier `should' put down a $ \bftw $,
		but this $ \bftw $ is instead swapped with a $ \bfon $ from the second last row.
	\end{ex}
	
	With these examples in mind, we can see
	when the carrier picks up a $ \bftw $ in the second last row, it generally gets swapped with a $ \bfon $ in the bottom row;
	and when the carrier `should' put down a $ \bftw $ in the bottom row,
	it generally gets swapped with a $ \bfon $ from the second last row.
	Note that sometimes this happens simultaneously.
	Note also that there are exceptions (\ref{enum:f} for example).
	
	This observation tells us that when computing time evolution 
	by moving the carrier through a genuine highest weight state,
	the bottom row of the super BBS is similar to the single row case,
	except for some swapping of $ \bftw $ and $ \bfon $.
	This swapping can occur simultaneously and is hard to keep track of.
	Our strategy is to swap all of the $ \bftw $ from the second last row with $ \bfon $ from the bottom row 
	\emph{before} passing the carrier through,
	so that no swapping occurs when passing the carrier through.
	After employing this strategy to create a `swapped state',
	the bottom row would theoretically behave identically to the single row case.
	This would reduce the problem of computing the time evolution to the single row case.
	This is advantageous because the behaviour of the single row BBS is simpler and also easier to compute
	since we can use an alternate algorithm for the $ R $-matrix detailed in~\cite[Section 2]{HI00}.
	Unfortunately, the first soliton moves slower in the `swapped state' than in the original state.
	We can fix this by `speeding up' the first soliton.
	
	\subsection{\texorpdfstring{$ L $}{L}-sped box-ball system}
	With this goal of `speeding up' the first soliton, we define a new time evolution operator:
	
	\begin{dfn}
		Consider a general state
		\[
		P= P_1\otimes P_2 \otimes P_3 \otimes \cdots 
		=u_1^{\otimes c_1} \otimes \inlinetab{\bfon}^{\otimes d}
		\otimes u_1^{\otimes c_2} \otimes 
		\inlinetab{y_1} \otimes \inlinetab{y_2} \otimes \cdots \otimes \inlinetab{y_{d_2}}
		\otimes u_1^{\otimes \infty}
		\]
		where $
		\begin{array}{|c|c|c|c|}
			\hline
			y_{d_2} & \cdots & y_2 & y_1\\
			\hline
		\end{array}
		$ is semistandard
		and $ \inlinetab{y_j}\neq u_1 $ for any $ j=1,\ldots,n $.
		(Note that $ d $ or $ c_2 $ can be $ 0 $.)
		For some non-negative integer $ L $, set
		\[
		P^L = 
		u_1^{\otimes c_1}\otimes 
		P_{c_1+d+1} \otimes \cdots \otimes P_{c_1+d+L}\otimes 
		\inlinetab{\bfon}^{\otimes \min(d,\ell-L)}
		\otimes P_{c_1+d+L+1}\otimes P_{c_1+d+L+2}\otimes\cdots.
		\]
		Let $ T_\ell(P^L) = \widetilde{P}_1\otimes \widetilde{P}_2 \otimes \widetilde{P}_3\otimes\cdots $
		and assume that $ \widetilde{P}_{\widetilde{c}_1+1} $ is the first non-vacuum element in $ T_\ell(P^L) $.
		Define the \defn{$ L $-sped time evolution operator},
		denoted $ T_\ell^L $, by 
		\[
		T_{\ell}^L(P) = u_1^{\otimes \widetilde{c}_1} \otimes \inlinetab{\bfon}^{\otimes d - \min(d,\ell-L)}
		\otimes \widetilde{P}_{\widetilde{c}_1+1}\otimes \widetilde{P}_{\widetilde{c}_1+2}\otimes\cdots.
		\]
	\end{dfn}
	
	Note that in the case where $ d\le \ell-L $, we have $ T_\ell^L(P) = T_{\ell}(P^L) $.
	
	\begin{ex}
		Let $ L=2 $ and $ \ell=5 $, and
		consider $ T^L_\ell $ acting on the state
		\[
		P=\inlinetab{\bfon}\otimes\inlinetab{\bfon}\otimes\inlinetab{\bfon}\otimes\inlinetab{\bfon}\otimes u_1 
		\otimes \inlinetab{\bfze} \otimes \inlinetab{\bfze} \otimes \inlinetab{\bfon}
		\otimes u_1^{\otimes\infty}.
		\]
		We have that 
		$ 
		P^L=u_1\otimes\inlinetab{\bfze}
		\otimes\inlinetab{\bfon}\otimes\inlinetab{\bfon}\otimes\inlinetab{\bfon}
		\otimes \inlinetab{\bfze} \otimes \inlinetab{\bfon}\otimes u_1^{\otimes\infty}.
		$
		Then, 
		$ T_\ell(P^L) = 
		u_1^{\otimes 5}\otimes \inlinetab{\bfon} \otimes 
		u_1 \otimes \inlinetab{\bfze}\otimes\inlinetab{\bfze}\otimes\inlinetab{\bfon}\otimes\inlinetab{\bfon}\otimes\inlinetab{\bfon} 
		\otimes u_1^{\otimes\infty}$.
		Therefore,
		\[
		T^L_\ell(P) = 
		u_1^{\otimes 5}\otimes \inlinetab{\bfon}\otimes \inlinetab{\bfon} \otimes 
		u_1 \otimes \inlinetab{\bfze}\otimes\inlinetab{\bfze}\otimes\inlinetab{\bfon}\otimes\inlinetab{\bfon}\otimes\inlinetab{\bfon} 
		\otimes u_1^{\otimes\infty}.
		\]
	\end{ex}
	
	Note that we have only defined the $ L $-sped time evolution operator for \emph{some} two-soliton states.
	This is sufficient for our purposes, but the definition can be extended using crystal operators if so desired.
	
	\begin{prop}\label{prop:dspedsolitons}
		Let 
		\[
		P=u_1^{\otimes \mathfrak{c}_1} \otimes 
		\inlinetab{\bfon}^{\otimes \mathfrak{d}_1-L} \otimes u_1^{\otimes \mathfrak{c}_2+L} \otimes 
		\inlinetab{\bfze}^{\otimes M}\otimes \inlinetab{\bfon}^{\otimes \mathfrak{d}_2-M}
		\otimes u_1^{\otimes\infty}
		\]
		be an arbitrary height $ 1 $ genuine highest weight two soliton state
		with $ \mathfrak{c}_2+L\ge \mathfrak{d}_2-M $.
		If $ L\le \mathfrak{d}_2 $ and $ \mathfrak{d}_1>\mathfrak{d}_2 $ then,
		\[
		(T_{\mathfrak{d}_2+1}^L)^t(P) =
		u_1^{\otimes \mathfrak{c}_3} \otimes 
		\inlinetab{\bfon}^{\otimes \mathfrak{d}_2-L} \otimes u_1^{\otimes \mathfrak{c}_4} \otimes 
		\inlinetab{\bfze}^{\otimes M}\otimes \inlinetab{\bfon}^{\otimes \mathfrak{d}_1-M}
		\otimes u_1^{\otimes\infty}
		\]
		for $ t $ sufficiently large.
		Moreover, the phase shift, as measured from the start of the solitons,
		is $ 2\mathfrak{d}_2-M-L $.
	\end{prop}
	\begin{proof}
		We follow a similar strategy as~\cite{HI00} and~\cite{FOY00}.
		Consider applying the time evolution operator to $ P $.
		For $ \mathfrak{c}_2 $ suficiently large,
		the first soliton initially moves with speed $ \mathfrak{d}_2+1 $ and the second soliton with speed $ \mathfrak{d}_2 $.
		Eventually, we arrive at the state:
		\[
		u_1^{\otimes \widetilde{\mathfrak{c}}_1} \otimes 
		\inlinetab{\bfon}^{\otimes \mathfrak{d}_1-L} \otimes u_1^{\otimes \mathfrak{d}_2-M} \otimes 
		\inlinetab{\bfze}^{\otimes M}\otimes \inlinetab{\bfon}^{\otimes \mathfrak{d}_2-M}
		\otimes u_1^{\otimes\infty}.
		\]
		By direct computation 
		(using either the $ R $-matrix algorithm in Section~\ref{sec:comb_R} or in~\cite[Section 2]{HI00}),
		we find that after another $ t' $ time units (with $ t'\le \mathfrak{d}_1-\mathfrak{d}_2 $),
		this state becomes:
		\[
		u_1^{\otimes \widetilde{\mathfrak{c}}_1} \otimes 
		u_1^{\otimes t'(\mathfrak{d}_2+1)} \otimes
		\inlinetab{\bfon}^{\otimes \mathfrak{d}_1- L-t'} \otimes u_1^{\otimes \mathfrak{d}_2-M} \otimes 
		\inlinetab{\bfze}^{\otimes M}\otimes \inlinetab{\bfon}^{\otimes \mathfrak{d}_2-M+t'}
		\otimes u_1^{\otimes\infty}.
		\]
		When $ t'>\mathfrak{d}_1-\mathfrak{d}_2 $ the solitons do not interact,
		and the first soliton moves with speed $ \mathfrak{d}_2 $ and the second soliton moves with speed  $ \mathfrak{d}_2+1 $.
		
		Let $ \delta_3 $ be the phase shift of the speed $ \mathfrak{d}_2+1 $ soliton,
		and $ \delta_2 $ be the phase shift of the speed $ \mathfrak{d}_2 $ soliton.
		After time $ t'=\mathfrak{d}_1-\mathfrak{d}_2 $, the phase shifts will remain the same,
		so computing both phase shifts at time $ t'=\mathfrak{d}_1-\mathfrak{d}_2 $, we find
		\begin{align*}
			\delta_3 &= (\widetilde{\mathfrak{c}}_1+(\mathfrak{d}_1-\mathfrak{d}_2)(\mathfrak{d}_2+1)+(\mathfrak{d}_1-L-\mathfrak{d}_1+\mathfrak{d}_2)+(\mathfrak{d}_2-M))
			-(\widetilde{\mathfrak{c}}_1+(\mathfrak{d}_1-\mathfrak{d}_2)(\mathfrak{d}_2+1))
			\\
			&= 2\mathfrak{d}_2-M-L\\
			\delta_2 &= 
			(\widetilde{\mathfrak{c}}_1 + (\mathfrak{d}_1-\mathfrak{d}_2)(\mathfrak{d}_2+1))
			-(\widetilde{\mathfrak{c}}_1 + (\mathfrak{d}_1 - L) + (\mathfrak{d}_2 - M) + (\mathfrak{d}_1-\mathfrak{d}_2)\mathfrak{d}_2)\\
			&= -(2\mathfrak{d}_2-M-L).
		\end{align*}
		Thus, the total phase shift of the system is 
		$ \delta = \delta_3 = -\delta_2 = 2\mathfrak{d}_2-M-L $.
	\end{proof}
	
	Note that Proposition~\ref{prop:dspedsolitons} essentially shows us that two-soliton states behave solitonically
	with respect to $ L $-sped time evolution.
	Our aim is to reduce the behaviour of the general super BBS to the behaviour of a height $1$ $ L $-sped BBS.
	
	\subsection{Reducing to height 1}\label{sec:ht1isomorphism}
	As alluded to in Section~\ref{sec:Rmatrixbehaviour},
	we can define a map from a genuine highest weight super BBS state to a $ L $-sped state
	by swapping all of the $ \bfon $ in the bottom row with $ \bftw $ in the second last row,
	and then removing all but the bottom row.
	More precisely,
	assuming $ d_1-L>0 $,
	the super BBS state
	\[
	u_1^{\otimes C_1}\otimes 
	\begin{array}{|c|}
		\hline
		\bfth\\
		\bfon\\
		\hline
	\end{array}^{\otimes d_1-Q}
	\otimes u_1^{\otimes C_2} \otimes
	\begin{array}{|c|}
		\hline
		\bfon\\
		y_1\\
		\hline
	\end{array}
	\otimes\cdots\otimes
	\begin{array}{|c|}
		\hline
		\bfon\\
		y_{Q}\\
		\hline
	\end{array}
	\otimes
	\begin{array}{|c|}
		\hline
		\bftw\\
		y_{Q+1}\\
		\hline
	\end{array}
	\otimes\cdots\otimes
	\begin{array}{|c|}
		\hline
		\bftw\\
		y_L\\
		\hline
	\end{array}
	\otimes
	\begin{array}{|c|}
		\hline
		\bfth\\
		y_{L+1}\\
		\hline
	\end{array}
	\otimes\cdots\otimes
	\begin{array}{|c|}
		\hline
		\bfth\\
		y_{d_2}\\
		\hline
	\end{array}
	\otimes u_1^{\otimes\infty}
	\]
	(recall that we are omitting $ \bfX $ as per Remark~\ref{rmk:omitX}),
	where $ y_1y_2\cdots y_{d_2} = \bfze^M\bfon^{d_2-M} $ and ($ Q>0 $ only if $ C_2=0 $),
	is mapped to the $ L $-sped state
	\[
	u_1^{\otimes C_1} \otimes \inlinetab{\bfon}^{\otimes d_1-L}
	\otimes u_1^{\otimes C_2+L-Q} \otimes 
	\inlinetab{y_1} \otimes \inlinetab{y_2} \otimes \cdots \otimes \inlinetab{y_{d_2}}
	\otimes u_1^{\otimes \infty}.
	\]
	Call this map $ \mathcal{F} $.
	
	\begin{ex}
		Let
		\[
		P = 
		\begin{array}{|c|}
			\hline
			\bfth\\
			\bfon\\
			\hline
		\end{array}
		\otimes
		\begin{array}{|c|}
			\hline
			\bfth\\
			\bfon\\
			\hline
		\end{array}
		\otimes
		\begin{array}{|c|}
			\hline
			\bfth\\
			\bfon\\
			\hline
		\end{array}
		\otimes
		\begin{array}{|c|}
			\hline
			\bfth\\
			\bfon\\
			\hline
		\end{array}
		\otimes
		\begin{array}{|c|}
			\hline
			\bfon\\
			\bfze\\
			\hline
		\end{array}
		\otimes
		\begin{array}{|c|}
			\hline
			\bftw\\
			\bfze\\
			\hline
		\end{array}
		\otimes
		\begin{array}{|c|}
			\hline
			\bfth\\
			\bfon\\
			\hline
		\end{array}
		\otimes u_1^{\otimes\infty}.
		\]
		Then, $
		\mathcal{F}(P)=\inlinetab{\bfon}\otimes\inlinetab{\bfon}\otimes\inlinetab{\bfon}\otimes u_1 
		\otimes \inlinetab{\bfze} \otimes \inlinetab{\bfze} \otimes \inlinetab{\bfon}
		\otimes u_1^{\otimes\infty}
		$.
	\end{ex}
	
	The map $ \mathcal{F} $ is invertible.
	Indeed, $ \mathcal{F}^{-1} $ maps the $ L $-sped state
	\[
	u_1^{\otimes c_1} \otimes \inlinetab{\bfon}^{\otimes d}
	\otimes u_1^{\otimes c_2} \otimes 
	\inlinetab{y_1} \otimes \inlinetab{y_2} \otimes \cdots \otimes \inlinetab{y_{d_2}}
	\otimes u_1^{\otimes \infty},
	\]
	where $ d>0 $,
	to the super BBS state
	\[
	u_1^{\otimes c_1}\otimes 
	\begin{array}{|c|}
		\hline
		\bfth\\
		\bfon\\
		\hline
	\end{array}^{\otimes d_1-Q}
	\otimes u_1^{\otimes \max(0,c_2-L)} \otimes
	\begin{array}{|c|}
		\hline
		\bfon\\
		y_1\\
		\hline
	\end{array}
	\otimes\cdots\otimes
	\begin{array}{|c|}
		\hline
		\bfon\\
		y_{Q}\\
		\hline
	\end{array}
	\otimes
	\begin{array}{|c|}
		\hline
		\bftw\\
		y_{Q+1}\\
		\hline
	\end{array}
	\otimes\cdots\otimes
	\begin{array}{|c|}
		\hline
		\bftw\\
		y_L\\
		\hline
	\end{array}
	\otimes
	\begin{array}{|c|}
		\hline
		\bfth\\
		y_{L+1}\\
		\hline
	\end{array}
	\otimes\cdots\otimes
	\begin{array}{|c|}
		\hline
		\bfth\\
		y_{d_2}\\
		\hline
	\end{array}
	\otimes u_1^{\otimes\infty},
	\]
	where $ d_1 = d+L $ and $ Q=\max(0,L-c_2) $.
	
	\begin{remark}\label{rmk:exceptall2}
		In defining the domain of $ \mathcal{F} $,
		we assumed that $ d_1-L>0 $.
		This assumption ensures that $ \mathcal{F}^{-1} $ is well-defined.
		We sometimes encounter states where $ d_1=L $
		and we need to deal with these edge cases separately.
	\end{remark}
	
	Recalling that 
	\[
	u_1 = 
	\begin{array}{|c|}
		\hline
		\bfth\\
		\bftw\\
		\hline
	\end{array},
	\]
	we can alternatively compute $ \mathcal{F}^{-1} $ 
	by the following process:
	\begin{enumerate}
		\item Add a row above the $ L $-sped state, and fill it entirely with $ \bfth $
		(if needed, add rows filled with $ \bfX $ above this added row).
		\item In the second soliton, replace the first $ L $ positions of this added row with $ \bftw $.
		\item Replace the first $ L $ occurrences of $ \bftw $ that occur immediately after the first soliton 
		(regardless of the row they appear in) with $ \bfon $.
	\end{enumerate}
	
	\begin{prop}\label{prop:ht1isomorphism}
		Let $ P $ be the state
		\[
		u_1^{\otimes C_1}\otimes 
		\begin{array}{|c|}
			\hline
			\bfth\\
			\bfon\\
			\hline
		\end{array}^{\otimes d_1-Q}
		\otimes u_1^{\otimes C_2} \otimes
		\begin{array}{|c|}
			\hline
			\bfon\\
			y_1\\
			\hline
		\end{array}
		\otimes\cdots\otimes
		\begin{array}{|c|}
			\hline
			\bfon\\
			y_{Q}\\
			\hline
		\end{array}
		\otimes
		\begin{array}{|c|}
			\hline
			\bftw\\
			y_{Q+1}\\
			\hline
		\end{array}
		\otimes\cdots\otimes
		\begin{array}{|c|}
			\hline
			\bftw\\
			y_L\\
			\hline
		\end{array}
		\otimes
		\begin{array}{|c|}
			\hline
			\bfth\\
			y_{L+1}\\
			\hline
		\end{array}
		\otimes\cdots\otimes
		\begin{array}{|c|}
			\hline
			\bfth\\
			y_{d_2}\\
			\hline
		\end{array}
		\otimes u_1^{\otimes\infty}
		\]
		where $ y_1y_2\cdots y_{d_2} = \bfze^M\bfon^{d_2-M} $ and
		$ C_2+L-Q\ge \min(d_1,\ell)-1-M $
		and ($ Q>0 $ only if $ C_2=0 $).
		Further assume that $ d_1-L>0 $ and that
		if $ C_2+L-Q=\min(d_1,\ell)-1-M $ then $ d_1-L-1>0 $.
		Then, for $ \ell\le d_2+1 $ and $ L,M < \ell $,
		we have $ T_\ell(P) = \mathcal{F}^{-1}T^L_\ell\mathcal{F}(P) $.
	\end{prop}
	
	The only states encountered when evolving a genuine highest weight super BBS are the form $ P $ 
	(or of the form discussed in Remark~\ref{rmk:exceptall2});
	this will be a consequence of the proof of Proposition~\ref{prop:ht1isomorphism}.
	Therefore, Proposition~\ref{prop:ht1isomorphism} allows us to compute the time evolution 
	of a genuine highest weight super BBS state using almost entirely height $1$ $ L $-sped states.
	Since $ L $-sped states behave solitonically by Proposition~\ref{prop:dspedsolitons},
	this will complete the proof of Proposition~\ref{prop:simplescattering} and hence
	Theorem~\ref{thm:scattering} (except in the edge cases of Remark~\ref{rmk:exceptall2}).
	
	To prove Proposition~\ref{prop:ht1isomorphism}, first observe that
	\[
	\mathcal{F}(P) = 
	u_1^{\otimes C_1} \otimes \inlinetab{\bfon}^{\otimes d_1-L}
	\otimes u_1^{\otimes C_2+L-Q} \otimes 
	\inlinetab{\bfze}^{\otimes M} \otimes \inlinetab{\bfon}^{\otimes d_2-M}
	\otimes u_1^{\otimes \infty}.
	\]
	Since $ R(u_\ell\otimes u_1) = u_1\otimes u_\ell $,
	(for $ u_1,\,u_\ell $ of any height)
	and $ \mathcal{F}^{-1} $ does not change $ u_1^{\otimes C_1} $,
	the result holds for the first $ C_1 $ tableaux of the states.
	We can thus assume, without loss of generality, that $ C_1=0 $.
	Additionally,
	if $ C_2\ge\min(d_1,\ell) $ then the solitons do not interact in either system at this time step,
	and the result is obvious.
	So, assume $ C_2<\min(d_1,\ell) $.
	
	Let $ \widetilde{C}=\min(d_1,\ell) $.
	If $ C_2+L-Q>\widetilde{C}-1-M $ we find that
	\begin{equation}\label{eq:ht1sizesame}
		T^L_\ell\mathcal{F}(P) = 
		u_1^{\otimes \widetilde{C}} \otimes \inlinetab{\bfon}^{\otimes d_1-L}
		\otimes u_1^{\otimes C_2+L-Q+(\min(d_2,\ell)-\widetilde{C})} \otimes 
		\inlinetab{\bfze}^{\otimes M} \otimes \inlinetab{\bfon}^{\otimes d_2-M}
		\otimes u_1^{\otimes \infty}
	\end{equation}
	and if $ C_2+L-Q = \widetilde{C}-1-M $ we find that
	\begin{equation}\label{eq:ht1sizechange}
		T^L_\ell\mathcal{F}(P) = 
		u_1^{\otimes \widetilde{C}} \otimes \inlinetab{\bfon}^{\otimes d_1-L-1}
		\otimes u_1^{\otimes \ell-1-M} \otimes 
		\inlinetab{\bfze}^{\otimes M} \otimes \inlinetab{\bfon}^{\otimes d_2-M + 1}
		\otimes u_1^{\otimes \infty}
	\end{equation}
	(noting that $ \min(\ell-1-M,d_2-M)=\ell-1-M $).
	
	Let $ P=P_1\otimes P_2\otimes P_3\otimes\cdots $
	and let $ u(j) $ denote the carrier after applying the $ R $-matrix to $ P $ $ j $-times.
	That is,
	\[
	R_{j-1}\cdots R_1R_0(u_{\ell}\otimes P)
	=
	\widetilde{P}_1\otimes \cdots \otimes \widetilde{P}_j \otimes u(j) \otimes P_{j+1}\otimes P_{j+2}\otimes \cdots
	\]
	for some $ \widetilde{P}_1,\ldots,\widetilde{P}_j $.
	Note that $ u(0)=u_\ell $.
	Let
	\[
	u(j) = 
	\begin{array}{|c|}
		\hline
		\bfth^{A_3(j)}\bftw^{A_2(j)}\bfon^{A_1(j)}\\
		\bftw^{B_2(j)}\bfon^{B_1(j)}\bfze^{B_0(j)}\\
		\hline
	\end{array}
	\]
	where $ A_3,A_2,A_1,B_2,B_1,B_0 $ are functions on non-negative integers.
	We will often write $ A_3 $ for $ A_3(j) $, with $ j $ being clear from context.
	Similarly for the other functions.
	
	From Lemma~\ref{lem:classRmat}~\ref{enum:o},
	with the assumption $ C_1=0 $,
	we find that
	$ \widetilde{P}_1=\cdots=\widetilde{P}_{\min(d_1-Q,\ell)}=u_1 $
	and 
	\[
	u(\min(d_1-Q,\ell)) =
	\begin{array}{|c|}
		\hline
		\bfth^{\ell}\\
		\bftw^{B_2}\bfon^{B_1}\\
		\hline
	\end{array}
	\]
	with $ B_1 = \min(d_1-Q,\ell) $ and $ B_2=\ell-B_1 $.
	If $ Q=0 $, then we clearly have $ \widetilde{P}_1=\cdots=\widetilde{P}_{\min(d_1,\ell)}=u_1 $,
	and if $ Q>0 $ (so that $ C_2=0 $) then, by~\ref{enum:q}, these equalities also hold.
	That is, the first $ \widetilde{C}=\min(d_1,\ell) $ tableaux in $ T_\ell(P) $ are $ u_1 $.
	This matches the first $ \widetilde{C} $ tableaux in $ \mathcal{F}^{-1}T^L_\ell\mathcal{F}(P) $
	(cf.~\eqref{eq:ht1sizesame} and~\eqref{eq:ht1sizechange}).
	Note that
	\[
	u(\widetilde{C}) = 
	\begin{array}{|c|}
		\hline
		\bfth^{\ell-A_1}\bfon^{A_1}\\
		\bftw^{B_2}\bfon^{B_1} \bfze^{A_1}\\
		\hline
	\end{array}
	\]
	where $ B_1=\min(d_1-Q,\ell) $, \, $ A_1=\widetilde{C}-B_1 $ and $ B_2=\ell-\widetilde{C} $.
	
	In the case where $ d_1\le\ell $ and $ C_2>0 $ or in the case where $ d_1-Q>\ell $, we have 
	\begin{equation}\label{eq:firstsoliton}
		\widetilde{P}_{\widetilde{C}+1} = 
		\begin{array}{|c|}
			\hline
			\bfth\\
			\bfon\\
			\hline
		\end{array}
	\end{equation}
	by~\ref{enum:p}.
	If $ d_1-Q\le \ell <d_1 $ (which implies $ Q>0 $),
	then~\eqref{eq:firstsoliton} holds by~\ref{enum:s}.
	If $ d_1\le \ell $ and $ C_2=0 $, then by~\ref{enum:k},\ref{enum:l},\ref{enum:m} or \ref{enum:n},
	we have that~\eqref{eq:firstsoliton} holds.
	Note that~\ref{enum:e} and~\ref{enum:f} do not apply, since $ B_1=d_1-Q\ge d_1-L>0 $
	and $ B_2+B_1 = A_3 $, so we cannot have $ B_2=A_3 $.
	\ref{enum:o} and~\ref{enum:p} do not apply,
	for otherwise $ C_2=0 $
	and $ M=Q $ (since $ B_1=d_1-Q>0 $ which would imply $ \mathbf{y}=\bfon $),
	so $ C_2+L-Q\ge \widetilde{C}-1-M $ if and only if $ L\ge \widetilde{C}-1 = d_1-1 $.
	But since $ d_1-L>0 $, we must have $ L=d_1-1 $.
	Then we have $ C_2+L-Q = \widetilde{C}-1-M $,
	which implies that $ d_1-L-1>0 $ by assumption, a contradiction.
	
	Therefore,
	the $ (\widetilde{C}+1) $-th tableaux in $ T_\ell(P) $ 
	is the same as the $ (\widetilde{C}+1) $-th tableaux in $ \mathcal{F}^{-1}T^M_\ell\mathcal{F}(P) $.
	In particular, the first soliton starts at the same point in both of these states.
	
	By~\ref{enum:p} and~\ref{enum:s}, we have that 
	\[
	\widetilde{P}_j =
	\begin{array}{|c|}
		\hline
		\bfth\\
		\bfon\\
		\hline
	\end{array}
	\qquad \text{and} \qquad
	u(d_1+C_2)
	=
	\begin{array}{|c|}
		\hline
		\bfth^{\ell-{A_1}}\bfon^{A_1}\\
		\bftw^{\ell-B_1-A_1} \bfon^{B_1} \bfze^{A_1}\\
		\hline
	\end{array}
	\]
	for $\widetilde{C}+1 \leq j \leq d_1 + C_2$,\, $ A_1=Q $ and $ B_1=\widetilde{C}-C_2-Q $.
	Observe that $ B_1>0 $.
	Indeed, if $ A_1=Q>0 $ then $ C_2=0 $,
	and the fact that $ Q<\widetilde{C} $ (which follows from $ Q\le L<\ell $ and $ Q\le L<d_1 $)
	implies that $ B_1=\widetilde{C}-Q>0 $.
	Alternatively, if $ A_1=Q=0 $ then the assumption that $ C_2<\min(d_1,\ell)=\widetilde{C} $ implies that
	$ B_1=\widetilde{C}-C_2>0 $.
	Note that a consequence of $ B_1>0 $ is that $ B_2(d_1+C_2)<A_3(d_1+C_2) $.
	
	Now, we investigate what happens when we move the carrier starting from $ u(d_1+C_2) $ 
	and pausing when we reach $ u((d_1-Q)+C_2+L) $.
	We find that Case~\ref{enum:k} applies until $ B_1=1 $ or we reach $ u((d_1-Q)+C_2+M) $
	(note that since $ B_2+B_1 = A_3 $ and $ B_1>0 $ when applying Case~\ref{enum:k}, we always have $ B_2 < A_3 $);
	then Case~\ref{enum:l} applies until $ B_2=A_3 $ 
	(note that, in Case~\ref{enum:l}, if $ B_1(j)=1 $ then $ A_3(j)=B_2(j)+1 $,
	so if $ B_1(j+1)=0 $ then $ A_3(j+1)=B_2(j)=B_2(j+1) $);
	Case~\ref{enum:e} applies until $ A_1=0 $; and from then onwards,~\ref{enum:f} applies.
	
	\begin{remark}
		Note that the only way to leave~\ref{enum:f} is to have $ B_2=A_3=0 $.
		Since $ A_3 $ is always decremented by one in~\ref{enum:k}, \ref{enum:l}, \ref{enum:e}, and \ref{enum:f},
		to have $ A_3=0 $ we would need $ (d_1-Q+C_2+L)-(d_1+C_2) = L-Q \ge A_3(d_1+C_2) = \ell-Q $.
		However, this inequality does not hold, because we assumed $ L<\ell $.
		We have thus shown that~\ref{enum:f} is the last possible case,
		and consequently that these four cases are the only four which we can encounter.
		However, note that some of these cases may be skipped,
		and that we might not reach the last case before we reach $ u((d_1-Q)+C_2+L) $.
	\end{remark}
	
	Next, we investigate what happens when we move the carrier starting from $ u((d_1-Q)+C_2+L) $
	and continuing until we complete the time evolution.
	We find that Case~\ref{enum:m} applies until $ B_1=0 $ or we reach $ u((d_1-Q)+C_2+M) $;
	then Case~\ref{enum:n} applies until $ B_2=0 $ or $ A_1=0 $, or until we reach $ u((d_1-Q)+C_2+d_2) $;
	Case~\ref{enum:o} applies until $ B_2=0 $ or until we reach $ u((d_1-Q) + C_2 + d_2) $;
	and from then onwards,~\ref{enum:p} applies.
	Observe that~\ref{enum:m} can only be preceded by~\ref{enum:k} or~\ref{enum:m};
	and that~\ref{enum:n} can only be preceded by~\ref{enum:k}, \ref{enum:l}, \ref{enum:e}, \ref{enum:m} or \ref{enum:n}.
	
	\begin{remark}
		Under the assumptions of Lemma~\ref{lem:classRmat},
		Cases~\ref{enum:m},~\ref{enum:n},~\ref{enum:o}~and~\ref{enum:p}
		are the only possible cases.
		However, note that some of these cases may be skipped.
	\end{remark}
	
	\begin{obs}\label{obs:1before2}
		The cases which put down a $ \bfon $ in the bottom row occur before
		the cases which put down a $ \bftw $ in the bottom row.
		Moreover, the cases following a $ \bftw $ in the bottom row
		only put down a $ \bftw $ (except in the final case,~\ref{enum:p}).
		Only $ \bfth $'s are put down in the second last row,
		except in~\ref{enum:p}.
	\end{obs}
	
	Let $ N(j) $ be the number of $ \bfon $'s in $ \widetilde{P}_1\otimes\cdots\otimes\widetilde{P}_j $.
	\begin{lemma}\label{lem:B_1takeA_2}
		For every case except~\ref{enum:o} and~\ref{enum:p}, we have $ d_1=N+\max(B_1-A_2,0)+A_1 $.
	\end{lemma}
	\begin{proof}
		Observe that 
		\begin{align*}
			N(d_1+C_2) + B_1(d_1+C_2) - A_2(d_1+C_2) + A_1(d_1+C_2) 
			&=(d_1+C_2-\widetilde{C})+(\widetilde{C}-C_2-Q)-0+Q \\
			&= d_1.
		\end{align*}
		We proceed by induction.
		\begin{itemize}
			\item If Case~\ref{enum:k} applies to $ u(j) $,
			then $ N(j+1) = N(j)+1 $, $ B_1(j+1) = B_1(j)-2 $, $ A_2(j+1)=0 $, $ A_1(j+1)=A_1(j)+1 $,
			so $  N(j+1)+B_1(j+1)-A_2(j+1)+A_1(j+1) = d_1 $ by the inductive hypothesis
			(noting that $ B_1(j+1)-A_2(j+1)\ge 0 $).
			
			Observe that $ B_2(j+1)+B_1(j+1)=A_3(j+1)+A_2(j+1) $,
			so if $ B_2(j+1)=A_3(j+1) $ then $ B_1(j+1)-A_2(j+1)=0 $.
			
			\item If Case~\ref{enum:l} applies to $ u(j) $ and $ j\ge (d_1-Q)+C_2+M $ (i.e.\ $ \mathbf{y}=\bfon $),
			then $ N(j+1) = N(j)+1 $, $ B_1(j+1) = B_1(j) $, $ A_2(j+1)=A_2(j)+1 $, $ A_1(j+1)=A_1(j) $,
			so $  N(j+1)+B_1(j+1)-A_2(j+1)+A_1(j+1) = d_1 $ by the inductive hypothesis.
			Since~\ref{enum:l} decreases $ B_1-A_2 $ by $ 1 $,
			and $ B_1(j)-A_2(j) = A_3(j)-B_2(j) > 0  $,
			we have that $ B_1(j+1)-A_2(j+1) \ge 0 $.
			
			Alternatively, if $ j<(d_1-Q)+C_2+M $ (i.e.\ $ \mathbf{y}\neq \bfon $) then $ B_1(j)=1 $.
			Since $ 1-A_2(j) = B_1(j)-A_2(j) = A_3(j)-B_2(j) > 0 $,
			we have that $ A_2(j)=0 $.
			Thus, $ B_1(j+1)-A_2(j+1) = -1 $.
			Therefore,
			\begin{align*}
				N(j+1) + \max(B_1(j+1)+A_2(j+1), 0) + A_1(j+1)
				&= (N(j) + 1) + 0 + A_1(j)\\
				&= N(j) + B_1(j)-A_2(j) + A_1(j)\\
				&=d_1.
			\end{align*}
			
			Observe that, if $ j\ge (d_1-Q)+C_2+M $
			then $ B_2(j+1)+B_1(j+1)=A_3(j+1)+A_2(j+1) $.
			And, if $ j<(d_1-Q)+C_2+M $ then $ B_1(j+1)=0 $.
			In either case, if $ A_3(j+1)=B_2(j+1) $ then $ B_1(j+1)-A_2(j+1)\le 0 $.
			
			\item Assume Case~\ref{enum:k}~or~\ref{enum:l} 
			applies to $ u(j_0) $ and Case~\ref{enum:e} applies to $ u(j_0+1) $.
			(Note that the first~\ref{enum:e} \emph{must} be preceded by~\ref{enum:k} or~\ref{enum:l} because $ B_2(d_1+C_2)<A_3(d_1+C_2) $.)
			Since we have $ A_3(j_0+1)=B_2(j_0+1) $,
			we have $ B_1(j_0+1)-A_2(j_0+1)\le 0 $,
			as noted in the above cases.
			
			That is, the first application of~\ref{enum:e} has $ B_1-A_2\le 0 $.
			Since further applications of~\ref{enum:e} only decrease
			the value of $ B_1-A_2 $, we have
			$ N + \max(B_1-A_2,0) + A_1 = N+A_1 $.
			Since~\ref{enum:e} increases $ N $ by $ 1 $ and decreases $ A_1 $ by $ 1 $,
			we have $ N+A_1=d_1 $ by the inductive hypothesis.
			
			\item Assume $ u(j_0+1) $ is the first carrier to which Case~\ref{enum:f} applies.
			If Case~\ref{enum:k}~or~\ref{enum:l} applied to $ u(j_0) $,
			then we have $ B_1(j_0+1)-A_2(j_0+1)\le 0 $ since $ A_3(j_0+1)=B_2(j_0+1) $, as noted in the above cases.
			If Case~\ref{enum:e} applies to $ u(j_0) $ then,
			as above, we have that $ B_1(j_0+1)-A_2(j_0+1)\le 0 $.
			(Note that the first~\ref{enum:f} \emph{must} be preceded by~\ref{enum:k},\ref{enum:l} or \ref{enum:e} because $ B_2(d_1+C_2)<A_3(d_1+C_2) $.)
			
			Since~\ref{enum:f} can only decrease the value of $ B_1-A_2 $,
			we have that $ B_1(j)-A_2(j)\le 0 $ for $ j > j_0 $.
			Therefore 
			\[
			N(j+1)+\max(B_1(j+1)-A_2(j+1),0) + A_1(j+1) = N(j)+A_1(j) = d_1
			\]
			by the inductive hypothesis.
			\item Case~\ref{enum:m} can only be preceded by Case~\ref{enum:k} or else is the first case.
			Regardless, $ A_2=0 $ for~\ref{enum:m}.
			Therefore, $ B_1-A_2\ge 0 $.
			Moreover, if~\ref{enum:m} applies to $ u(j) $
			then $ N(j+1)=N(j)+1,\,B_1(j+1)=B_1(j)-1,\,A_2(j+1)=0,\,A_1(j+1)=A_1(j) $.
			Therefore, $ N(j+1)+B_1(j+1)-A_2(j+1)+A_1(j+1)=d_1 $ by the inductive hypothesis.
			\item Assume~\ref{enum:n} applies to $ u(j) $.
			If $ B_1(j)=0 $ then $ B_1(j+1)-A_2(j+1)\le B_1(j)-A_2(j)\le 0 $,
			so $ \max(B_1(j+1)-A_2(j+1),0)=0 $.
			If $ B_1(j)\neq 0 $, we have that $ B_1(j+1)=B_1(j)+1 $ and $ A_2(j+1)=A_2(j)+1 $,
			so $ \max(B_1(j+1)-A_2(j+1),0)=\max(B_1(j)-A_2(j),0) $.
			
			Additionally, $ N(j+1)=N(j)+1 $ and $ A_1(j+1)=A_1(j)-1 $.
			Therefore,
			\[
			N(j+1)+\max(B_1(j+1)-A_2(j+1),0)+A_1(j+1) = d_1
			\]
			by the inductive hypothesis. \qedhere
		\end{itemize}
	\end{proof}

	\begin{lemma}\label{lem:calcB_1takeA_2}
		If $ C_2+L-Q>\widetilde{C}-1-M $ then $ B_1-A_2 \le 0 $ immediately before the first~\ref{enum:o}/\ref{enum:p}.
		If $ C_2+L-Q=\widetilde{C}-1-M $ then $ B_1-A_2 = 1 $ immediately before the first~\ref{enum:o}/\ref{enum:p}.
	\end{lemma}
	\begin{proof}
		We consider three cases.
		\begin{itemize}
			\item Assume $ C_2+L-Q>\widetilde{C}-M-1 $.
			Recall that $ B_1(d_1+C_2)=\widetilde{C}-C_2-Q $ and $ A_2(d_1+C_2)=0 $.
			We can apply~\ref{enum:k} $ \lfloor B_1(d_1+C_2)/2 \rfloor $ times or $ M-Q $
			or $ L-Q $ times, whichever is smaller.
			
			If we apply~\ref{enum:k} $ \lfloor B_1(d_1+C_2)/2 \rfloor $ times,
			then $ B_1=0 $ or $ B_1=1 $ (and $ A_2=0 $).
			If $ B_1=0 $ then we are done,
			since $ B_1-A_2\le0 $ and the cases afterwards can only decrease the value of $ B_1-A_2 $.
			If $ B_1=1 $ and Case~\ref{enum:l}~or~\ref{enum:m} applied,
			then the value of $ B_1-A_2 $ will decrease by $ 1 $ or more,
			so $ B_1-A_2\le 0 $.
			If Case~\ref{enum:e}~or~\ref{enum:f} applied,
			then we would already have $ B_1-A_2\le 0 $ 
			(as in the proof of Lemma~\ref{lem:B_1takeA_2}).
			Also note that if $ B_1=1 $ then~\ref{enum:n}, \ref{enum:o}, \ref{enum:p} do not apply,
			for otherwise we would have 
			$ \lfloor B_1(d_1+C_2)/2\rfloor = (\widetilde{C}-C_2-Q-1)/2=L-Q=M-Q $ 
			and hence $ \widetilde{C}-M-1=C_2+L-Q $,
			contradicting $ \widetilde{C}-M-1<C_2+L-Q $.
			After the $ B_1=1 $ case, the value of $ B_1-A_2 $ can only decrease.
			
			Alternatively, if we apply~\ref{enum:k} $ M-Q $ times, then
			$ B_1=\widetilde{C}-C_2-Q-2(M-Q) \le L-M $ (and $ A_2=0 $).
			If we ever apply~\ref{enum:e} or~\ref{enum:f} then $ B_1-A_2\le 0 $,
			as in the proof of Lemma~\ref{lem:B_1takeA_2}.
			If we never apply~\ref{enum:e} or~\ref{enum:f},
			then we must apply~\ref{enum:l} $ L-M $ times.
			Thus, $ B_1-A_2 = \widetilde{C}-C_2-Q-2(M-Q) - (L-M) \le 0 $.
			The remaining cases can only decrease the value of $ B_1-A_2 $.
			
			If we apply~\ref{enum:k} $ L-Q $ times, then
			$ B_1=\widetilde{C}-C_2-Q-2(L-Q)\le M-L $ (and $ A_2=0 $).
			Then, we apply~\ref{enum:m} $ \widetilde{C}-C_2-Q-2(L-Q) $ times,
			so we have $ B_1=0 $ and $ A_2=0 $.
			Since~\ref{enum:n} can only decrease the value of $ B_1-A_2 $,
			we have that $ B_1-A_2\le 0 $.
			
			\item Assume $ C_2+L-Q=\widetilde{C}-1-M $ and $ L\ge M $.
			Recall that $ B_1(d_1+C_2)=\widetilde{C}-C_2-Q $ and $ A_2(d_1+C_2)=0 $.
			We can apply~\ref{enum:k} $ \lfloor B_1(d_1+C_2)/2 \rfloor $ times or $ M-Q $ times,
			whichever is smaller.
			Observe that $ \widetilde{C}+Q-C_2=L+M+1\ge 2M+1 $, so $ \widetilde{C}-Q-C_2\ge 2(M-Q)+1 $.
			Therefore, $ M-Q\le\lfloor B_1(d_1+C_2)/2 \rfloor $ and we can only apply~\ref{enum:k} $ M-Q $ times.
			We find $ B_1(d_1+C_2 + M-Q)=\widetilde{C}-C_2-Q - 2(M-Q) $ and $ A_2(d_1+C_2+M-Q)=0 $.
			Then, we apply~\ref{enum:l} $ L-M $ times,
			with $ B_1 $ remaining the same and $ A_2(d_1+C_2+L-Q) = L-M $.
			That is,
			\[
			B_1(d_1+C_2+L-Q) - A_2(d_1+C_2+L-Q)
			= \widetilde{C}-C_2+Q -M - L = 1.
			\]
			(Note that we were able to apply~\ref{enum:l} $ L-M $ times because $ B_1-A_2>0 $;
			we would have $ B_1-A_2\le 0 $ when applying~\ref{enum:e} or~\ref{enum:f},
			as in the proof of Lemma~\ref{lem:B_1takeA_2}.)
			Then we apply~\ref{enum:n}, which does not change the value of $ B_1-A_2 $ in this case.
			
			\item Assume $ C_2+L-Q=\widetilde{C}-1-M $ and $ L<M $.
			Recall that $ B_1(d_1+C_2)=\widetilde{C}-C_2-Q $ and $ A_2(d_1+C_2)=0 $.
			We can apply~\ref{enum:k} $ \lfloor B_1(d_1+C_2)/2 \rfloor $ times or $ L-Q $ times,
			whichever is smaller.
			Observe that $ \widetilde{C}-C_2+Q=L+M+1> 2L+1 $, so $ \widetilde{C}-C_2-Q> 2(L-Q)+1 $.
			Therefore, $ L-Q\le\lfloor B_1(d_1+C_2)/2\rfloor $ and we can only apply~\ref{enum:k} $ L-Q $ times.
			We find $ B_1(d_1+C_2 + L-Q)=\widetilde{C}-C_2-Q - 2(L-Q) $ and $ A_2(d_1+C_2+L-Q)=0 $.
			Then, we apply~\ref{enum:m} $ M-L $ times,
			to get $ B_1(d_1+C_2+M-Q)=\widetilde{C}-C_2-Q-2(L-Q)-(M-L) = 1 $ and $ A_2(d_1+C_2+M-Q)=0 $.
			Next, we apply~\ref{enum:n}, which does not change the value of $ B_1-A_2 $ in this case.
			\qedhere
		\end{itemize}
	\end{proof}
	
	From Observation~\ref{obs:1before2},
	we can conclude that the tableaux in $ T_\ell(P) $ before we reach Case~\ref{enum:p}
	are of the form
	\[ 
	\begin{array}{|c|}
		\hline
		\bfth\\
		\bftw\\
		\hline
	\end{array}
	^{\otimes \widetilde{C}} \otimes
	\begin{array}{|c|}
		\hline
		\bfth\\
		\bfon\\
		\hline
	\end{array}
	^{\otimes \widetilde{d}_1}\otimes
	\begin{array}{|c|}
		\hline
		\bfth\\
		\bftw\\
		\hline
	\end{array}
	^{\otimes \widetilde{C}_2}.
	\]
	for some non-negative integers $ \widetilde{d}_1 $ and $ \widetilde{C}_2 $.
	Observe that~\ref{enum:p} puts down the columns of the carrier in reverse order,
	so we simply need to examine the carrier $ u(\widetilde{C}+\widetilde{d}_1+\widetilde{C}_2) $.
	
	\begin{obs}\label{obs:A_2plusA_1}
		The value of $ A_2+A_1 $ is incremented by $ 1 $
		in cases~\ref{enum:q}, \ref{enum:s}, \ref{enum:k}, \ref{enum:l}, \ref{enum:e}, \ref{enum:f}
		and the value of $ A_2+A_1 $ stays the same in cases~\ref{enum:m}, \ref{enum:n}, \ref{enum:o}.
	\end{obs}
	
	Observation~\ref{obs:A_2plusA_1} tells us that 
	$ A_2(\widetilde{C}+\widetilde{d}_1+\widetilde{C}_2)+A_1(\widetilde{C}+\widetilde{d}_1+\widetilde{C}_2) = L $.
	Therefore, the number of $ \bfon $ and $ \bftw $ in the second last row of $ T_\ell(P) $ is $ L $,
	so $ \mathcal{F}T_\ell(P) $ is a $ L $-sped state.
	
	Now, we determine the value of $ \widetilde{d}_1 $ by examining $ u(\widetilde{C}+\widetilde{d}_1) $.
	Observe that~\ref{enum:o} and~\ref{enum:p} do not apply before $ u(\widetilde{C}+\widetilde{d}_1) $,
	so Lemma~\ref{lem:B_1takeA_2} applies.
	That is, at $ u(\widetilde{C}+\widetilde{d}_1) $ we have $ d_1=N+\max(B_1-A_2,0)+A_1 $.
	Rearranging, we get $ \widetilde{d}_1 = N = d_1-A_1-\max(B_1-A_2,0) $.
	By Lemma~\ref{lem:calcB_1takeA_2}, 
	\[
	\max(B_1-A_2,0)
	=
	\begin{cases}
		0 & \text{if}~ C_2+L-Q>\widetilde{C}-1-M \\
		1 & \text{if}~ C_2+L-Q=\widetilde{C}-1-M 
	\end{cases}
	.
	\]
	If $ \widetilde{C}_2>0 $, then~\ref{enum:f} or~\ref{enum:o} applies to $ u(\widetilde{C}+\widetilde{d}_1) $;
	in both of these cases we must have $ A_1=0 $, which implies $ \widetilde{d}_1=d_1-\max(B_1-A_2,0) $.
	We find that the first $ \widetilde{C}+\widetilde{d}_1 $ factors in $ \mathcal{F}T_\ell(P) $ are
	\begin{equation}\label{eq:Fsoliton1}
		\begin{cases}
			u_1^{\otimes \widetilde{C}} \otimes \inlinetab{\bfon}^{\otimes d_1-L}\otimes u_1^{\otimes L} 
			& \text{if}~C_2+L-Q>\widetilde{C}-1-M\\
			u_1^{\otimes \widetilde{C}} \otimes \inlinetab{\bfon}^{\otimes d_1-L-1} \otimes u_1^{\otimes L}
			& \text{if}~C_2+L-Q=\widetilde{C}-1-M
		\end{cases}
	\end{equation}
	which matches~\eqref{eq:ht1sizesame} and~\eqref{eq:ht1sizechange}.
	If $ \widetilde{C}_2=0 $ then~\ref{enum:p} applies to $ u(\widetilde{C}+\widetilde{d}_1) $,
	and the columns are put down in reverse order.
	Thus, in $ T_\ell(P) $, there are $ A_1(\widetilde{C}+\widetilde{d}_1) $ factors with $ \bfon $ in the second last row,
	and $ \widetilde{d}_1 $ factors $ \bfon $ in the bottom row immediately before them.
	Therefore, the first $ \widetilde{C}+\widetilde{d}_1 $ factors in $ \mathcal{F}T_\ell(P) $
	are $ u_1^{\otimes \widetilde{C}}\otimes\inlinetab{\bfon}^{\otimes \widetilde{d}_1+A_1-L}\otimes u_1^{\otimes L-A_1} $
	(note that $ A_1\le L $ by Observation~\ref{obs:A_2plusA_1}).
	Since $ \widetilde{d}_1+A_1 = d_1-\max(B_1-A_2,0) $,
	this is the same as~\eqref{eq:ht1sizesame} and~\eqref{eq:ht1sizechange}.
	
	Now, consider $ u(\widetilde{C}+\widetilde{d}_1+\widetilde{C}_2) $.
	Let $ D=\max(B_1(\widetilde{C}+\widetilde{d}_1)-A_2(\widetilde{C}+\widetilde{d}_1),0) $.
	Observe that 
	we have $ d_1=\widetilde{d}_1+D+A_1(\widetilde{C}+\widetilde{d}_1+\widetilde{C}_2) $.
	But, we have a total of 
	$ \widetilde{d}_1+B_1(\widetilde{C}+\widetilde{d}_1+\widetilde{C}_2)+A_1(\widetilde{C}+\widetilde{d}_1+\widetilde{C}_2) $
	appearances of $ \bfon $
	in the state or carrier at $ u(\widetilde{C}+\widetilde{d}_1+\widetilde{C}_2) $.
	Therefore, 
	\begin{equation}\label{eq:keeptrack}
		(\widetilde{d}_1+B_1(\widetilde{C}+\widetilde{d}_1+\widetilde{C}_2)+A_1(\widetilde{C}+\widetilde{d}_1+\widetilde{C}_2)) 
		-(\widetilde{d}_1+D+A_1(\widetilde{C}+\widetilde{d}_1+\widetilde{C}_2)) 
		=B_1(\widetilde{C}+\widetilde{d}_1+\widetilde{C}_2)-D 
	\end{equation}
	of the $ \bfon $ have come from the second soliton's bottom row in $ P $.
	
	Since~\ref{enum:p} applies to $ u(\widetilde{C}+\widetilde{d}_1+\widetilde{C}_2) $,
	we must have $ \widetilde{C}+\widetilde{d}_1+\widetilde{C}_2 = d_1+C_2+d_2-Q $ or $ B_2=0 $ (with $ B_1=0 $ or $ \mathbf{y}=\bfon $).
	We claim that $ \widetilde{C}+\widetilde{d}_1+\widetilde{C}_2 = d_1+C_2+\min(\ell-D,d_2)-Q $.
	Indeed:
	\begin{itemize}
		\item In the case $ \widetilde{C}+\widetilde{d}_1+\widetilde{C}_2=d_1+C_2+d_2-Q $,
		all of the $ \bfon $ from the second soliton have been picked up by the carrier,
		so $ B_1-D = d_2-M $ (using~\eqref{eq:keeptrack}).
		The carrier also has $ M $ copies of $ \bfze $,
		so we have that $ \ell \geq B_1+M = (d_2-M+D)+M $;
		in particular, $ d_2\le \ell-D $.
		\item In the case $ B_2=0 $ we have $ B_1+B_0 = \ell $.
		If $ \mathbf{y}=\bftw $ then the previous case applies,
		so assume $ \mathbf{y}\neq\bftw $.
		If $ \mathbf{y}=\bfon $, then the carrier has picked up all of the
		$ \bfze $, so $ B_0=M $.
		If $ \mathbf{y}=\bfze $, then we must have $ B_1=0 $ for~\ref{enum:p} to apply, and so $ \ell=B_0\le M $,
		contradicting $ M<\ell $.
		We deduce that, in this case, $ \mathbf{y}=\bfon $ and $ B_0=M $.
		Since $ B_1-D = (\ell-M)-D $ of the $ \bfon $'s have come from the second soliton
		(using~\eqref{eq:keeptrack}), we can deduce that 
		$ \widetilde{C}+\widetilde{d}_1+\widetilde{C}_2 = d_1+C_2+(M-Q)+(\ell-M-D) = d_1+C_2+\ell-D-Q $.
		Moreover, since $ \ell-M-D $ of the $ \bfon $ come from the $ d_2-M $ copies of $ \bfon $ in
		the second soliton, we have $ \ell-M-D\le d_2-M $ which implies $ \ell-D\le d_2 $.
		
	\end{itemize}
	We conclude that $ \widetilde{C}+\widetilde{d}_1+\widetilde{C}_2 = d_1+C_2+\min(\ell-D,d_2)-Q $.
	Comparing with~\eqref{eq:ht1sizesame} and~\eqref{eq:ht1sizechange},
	this is what we desire.
	That is, the second soliton in $ \mathcal{F}T_\ell(P) $ starts at the same place
	as the second soliton in $ T^L_\ell \mathcal{F}(P) $.
	Since the first solitons in both of these are the same length,
	we have that the second solitons are also the same length.
	Since~\ref{enum:p} puts the columns down in reverse order,
	we can thus deduce that
	\[
	\mathcal{F}T_\ell(P) = T^L_\ell \mathcal{F}(P).
	\]
	This proves Proposition~\ref{prop:ht1isomorphism}.
	
	Observe that $ T_\ell(P) $ satisfies the assumptions of Proposition~\ref{prop:ht1isomorphism},
	except we may have $ d_1-L=0 $ or $ d_1-L-1=0 $ with $ C_2+L-Q=\min(d_1,\ell)-1-M $.
	
	\subsection{Edge cases}
	As previously noted,
	Propositions~\ref{prop:ht1isomorphism}~and~\ref{prop:dspedsolitons}
	together prove the main theorem except in edge cases.
	It is not difficult to see that
	these edge cases are exactly the genuine highest weight states in which $ L=\mathfrak{d}_2 $.
	
	Let $ \mathfrak{p} $ be a genuine highest weight state (as in~\eqref{eq:hwtstate}) in which $ L=\mathfrak{d}_2 $,
	and let $ t_0 = ((\mathfrak{c}_2+L)-(\mathfrak{d}_2-M))+(\mathfrak{d}_1-\mathfrak{d}_2) 
	= \mathfrak{c}_2+M+\mathfrak{d}_1-\mathfrak{d}_2 $.
	We can use repeated applications of Proposition~\ref{prop:ht1isomorphism} to compute 
	\[
	(T_{\mathfrak{d}_2+1})^{t_0-1}(\mathfrak{p}) = 
	u_1^{\otimes \mathfrak{c}_1+(\mathfrak{d}_2+1)(t_0-1)}\otimes 
	\begin{array}{|c|}
		\hline
		\bfth\\
		\bfon\\
		\hline
	\end{array}^{\otimes \mathfrak{d}_2-M+1}
	\otimes
	\begin{array}{|c|}
		\hline
		\bfon\\
		\bfze\\
		\hline
	\end{array}
	^{\otimes M}
	\otimes
	\begin{array}{|c|}
		\hline
		\bftw\\
		\bfon\\
		\hline
	\end{array}
	^{\otimes \mathfrak{d}_2-M}
	\otimes
	\begin{array}{|c|}
		\hline
		\bfth\\
		\bfon\\
		\hline
	\end{array}
	^{\otimes \mathfrak{d}_1-\mathfrak{d}_2-1}
	\otimes u_1^{\otimes\infty}.
	\]
	At this point, we can no longer apply Proposition~\ref{prop:ht1isomorphism}.
	So, we manually compute the time steps after this point using Lemma~\ref{lem:classRmat}.
	
	If $ t_0\le t\le t_0+M $ then $ (T_{\mathfrak{d}_2+1})^t(\mathfrak{p}) $ is:
	\[
	u_1^{\otimes \mathfrak{c}_1+(\mathfrak{d}_2+1)t_0 + \mathfrak{d}_2(t-t_0)}\otimes 
	\begin{array}{|c|}
		\hline
		\bfth\\
		\bfon\\
		\hline
	\end{array}^{\otimes \mathfrak{d}_2 - M + (t-t_0)}
	\otimes
	\begin{array}{|c|}
		\hline
		\bfon\\
		\bfze\\
		\hline
	\end{array}
	^{\otimes M-(t-t_0)}
	\otimes
	\begin{array}{|c|}
		\hline
		\bftw\\
		\bfze\\
		\hline
	\end{array}
	^{\otimes t-t_0}
	\otimes
	\begin{array}{|c|}
		\hline
		\bftw\\
		\bfon\\
		\hline
	\end{array}
	^{\otimes \mathfrak{d}_2-M}
	\otimes
	\begin{array}{|c|}
		\hline
		\bfth\\
		\bfon\\
		\hline
	\end{array}
	^{\otimes \mathfrak{d}_1-\mathfrak{d}_2}
	\otimes u_1^{\otimes\infty}.
	\]
	Indeed, if $ t=t_0 $, we simply apply~\ref{enum:p}, \ref{enum:o}, \ref{enum:q}, \ref{enum:l}, \ref{enum:p}
	to $ (T_{\mathfrak{d}_2+1})^{t_0-1}(\mathfrak{p}) $.
	If $ t_0<t\le t_0+M $ then, to $ (T_{\mathfrak{d}_2+1})^{t-1} $, we first apply~\ref{enum:p}, \ref{enum:o}, \ref{enum:q}.
	Then, we apply~\ref{enum:k}, \ref{enum:l}, \ref{enum:e} --- note that we apply~\ref{enum:k} and~\ref{enum:e} the same number of times. Finally, we apply~\ref{enum:n} once and then~\ref{enum:p} for the remainder.
	
	If $ t > t_0+M $ then $ (T_{\mathfrak{d}_2+1})^t(\mathfrak{p}) $ is:
	\[
	u_1^{\otimes \mathfrak{c}_1+(\mathfrak{d}_2+1)t_0 + \mathfrak{d}_2(t-t_0)}\otimes 
	\begin{array}{|c|}
		\hline
		\bfth\\
		\bfon\\
		\hline
	\end{array}^{\otimes \mathfrak{d}_2}
	\otimes u_1^{\otimes t-(t_0+M)} \otimes
	\begin{array}{|c|}
		\hline
		\bftw\\
		\bfze\\
		\hline
	\end{array}
	^{\otimes M}
	\otimes
	\begin{array}{|c|}
		\hline
		\bftw\\
		\bfon\\
		\hline
	\end{array}
	^{\otimes \mathfrak{d}_2-M}
	\otimes
	\begin{array}{|c|}
		\hline
		\bfth\\
		\bfon\\
		\hline
	\end{array}
	^{\otimes \mathfrak{d}_1-\mathfrak{d}_2}
	\otimes u_1^{\otimes\infty}.
	\]
	Indeed, this is obvious for $ t>t_0+M+\mathfrak{d}_2 $.
	If $ t\leq t_0+M+\mathfrak{d}_2 $we apply~\ref{enum:p}, \ref{enum:o}, \ref{enum:p}.
	Then we apply~\ref{enum:k}, \ref{enum:l}, \ref{enum:e}, \ref{enum:f} --- note that, during this computation, $ a_3 - b_2 + a_1 $ is the number of $ \bfon $ in the carrier from the first soliton;
	that $ a_3 - b_2 + a_1 < \mathfrak{d}_2 $ at the start of this computation;
	and that $ a_3-b_2+a_1 $ is decreased by one in~\ref{enum:k}, \ref{enum:l}, \ref{enum:e}.
	Finally,~\ref{enum:o} applies once and then~\ref{enum:p} applies for the remainder.
	
	Therefore, the case where $ L=\mathfrak{d}_2 $ behaves solitonically.
	Moreover, we can easily calculate 
	that the phase shift is $ \delta=\mathfrak{d}_2-M = 2\mathfrak{d}_2-L-M $.
	This takes care of the edge cases and completes the proof of Theorem~\ref{thm:scattering}.
	
	\begin{remark}
		We could possibly reformulate $ L $-sped solitons to 
		allow size $ 0 $ (i.e.\ speed $ L $) solitons.
		This would have the benefit of removing the edge cases,
		but would also make the definitions and proofs more technical.
	\end{remark}
	
	\subsection{The \texorpdfstring{$ m=r+1 $}{m=r+1} case}
	In Section~\ref{sec:simplifications} we assumed that $ m \ge r+2 $.
	All that remains is to prove the $ m = r+1 $ and $ m = r $ cases.
	(Recall that we will not consider $ m<r $, see Remark~\ref{rmk:r>m}.)
	
	For the $ m=r+1 $ case,
	we can define $ \bfX,\, \bfth,\, \bftw, \bfon $ as in Subsection~\ref{sec:simplifications}
	(note that $ \bfth = \bth $, $ \bftw = \btw $ and $ \bfon = \bon $).
	The genuine highest weight states are of the form:
	\[
	\mathfrak{p} =
	u_1^{\otimes \mathfrak{c}_1}\otimes 
	\begin{array}{|c|}
		\hline
		\raisebox{-1pt}{$ \bfX $}\\
		\raisebox{-1pt}{$ \bth $}\\
		\raisebox{-1pt}{$ \bon $}\\
		\hline
	\end{array}
	^{\otimes \mathfrak{d}_1}
	\otimes u_1^{\otimes \mathfrak{c}_2}\otimes 
	\begin{array}{|c|}
		\hline
		\raisebox{-1pt}{$ \bfX $}\\
		\raisebox{-1pt}{$ \mathfrak{w}_1 $}\\
		\raisebox{-1pt}{$ \mathfrak{y}_1 $}\\
		\hline
	\end{array}
	\otimes \cdots \otimes 
	\begin{array}{|c|}
		\hline
		\raisebox{-1pt}{$ \bfX $}\\
		\raisebox{-1pt}{$ \mathfrak{w}_{\mathfrak{d}_2} $}\\
		\raisebox{-1pt}{$ \mathfrak{y}_{\mathfrak{d}_2} $}\\
		\hline
	\end{array}
	\otimes u_1^{\otimes \infty }
	\]
	where
	\begin{align*}
		\mathfrak{w}_j = 
		\begin{cases}
			\btw & \text{if } j\le L\\
			\bth & \text{if } j>L
		\end{cases},
		&&
		\mathfrak{y}_j = 
		\begin{cases}
			M+1-j & \text{if }j\le M\\
			\bon & \text{if }j>M
		\end{cases},
	\end{align*}
	for some non-negative integers $ L,M \le \mathfrak{d}_2 $.
	Note that the value of $ n $ may limit which of these highest weight states are possible.
	Note also that there are similarities with the genuine highest weight states for the $ m\ge r+2 $ case.
	Lemma~\ref{lem:classRmat} can be adapted for $ m=r+1 $ by:
	\begin{itemize}
		\item replacing $ \bfze^{b_0} $ with $ (\mathbf{z}_j)_{j=b_2+b_1+1}^\ell $ 
		where $ \mathbf{z}_{b_2+b_1+1}\ge 1 $ 
		and $ \mathbf{z}_j = \mathbf{z}_{b_2+b_1+1} + j-(b_2+b_1+1) $ for $ j\ge p+q+1 $;
		\item replacing any instance of `$ \mathbf{y}=\bfze $' with `$ 1 \le \mathbf{y} < \mathbf{z}_{b_2+b_1+1} $'; and
		\item replacing `$ \mathbf{y}\in\{\bfze,\bfon\} $' with `$ \bfon\le \mathbf{y} < \mathbf{z}_{b_2+b_1+1} $'
		and replacing `$ \mathbf{y}\in\{\bfze,\bfon,\bftw\} $' with `$ \bftw\le \mathbf{y} < \mathbf{z}_{b_1+b_2+1} $'.
	\end{itemize}
	After these modifications, the proof of Lemma~\ref{lem:classRmat} 
	remains similar to Appendix~\ref{appendix:classRmatprf}.
	Then, with some care, the arguments for the $ m\geq r+2 $ case can be adapted to prove Theorem~\ref{thm:scattering} in the $ m=r+1 $ case.
	
	\subsection{The \texorpdfstring{$ m=r $}{m=r} case}
	The $ m=r $ case is slightly different than the other cases,
	owing to the fact that the genuine highest weight states have a different form:
	\[
	\mathfrak{p} =
	u_1^{\otimes \mathfrak{c}_1}\otimes 
	\begin{array}{|c|}
		\hline
		\bfX\\
		\btw\\
		\mathfrak{y}_1\\
		\hline
	\end{array}
	\otimes\cdots\otimes
	\begin{array}{|c|}
		\hline
		\bfX\\
		\btw\\
		\mathfrak{y}_{\mathfrak{d}_1-M}\\
		\hline
	\end{array}
	\otimes
	\begin{array}{|c|}
		\hline
		\bfX\\
		\btw\\
		\mathfrak{x}_1\\
		\hline
	\end{array}
	\otimes \cdots \otimes 
	\begin{array}{|c|}
		\hline
		\bfX\\
		\btw\\
		\mathfrak{x}_M\\
		\hline
	\end{array}
	\otimes u_1^{\otimes \mathfrak{c}_2}\otimes 
	\begin{array}{|c|}
		\hline
		\bfX\\
		\mathfrak{w}_{\mathfrak{d}_1-M+1}\\
		\mathfrak{y}_{\mathfrak{d}_1-M+1}\\
		\hline
	\end{array}
	\otimes \cdots \otimes 
	\begin{array}{|c|}
		\hline
		\bfX\\
		\mathfrak{w}_{\mathfrak{d}_1+\mathfrak{d}_2-M}\\
		\mathfrak{y}_{\mathfrak{d}_1+\mathfrak{d}_2-M}\\
		\hline
	\end{array}
	\otimes u_1^{\otimes \infty }
	\]
	where
	\begin{align*}
		\mathfrak{y}_j = 
		\mathfrak{d}_1+\mathfrak{d}_2-M+1-j.
		&&
		\mathfrak{x}_j =
		M+1-j,
		&&
		\mathfrak{w}_j = 
		\begin{cases}
			\bon & \text{if } j\le \mathfrak{d}_1-M+L\\
			\btw & \text{if } j> \mathfrak{d}_1-M+L
		\end{cases},
	\end{align*}
	for some non-negative integers $ L,M\le \mathfrak{d}_2 $.
	Note that, for $ m=r $, we have $ \bfth=\btw $ and $ \bftw = \bon $.
	Note also that solitons can not be longer than $ n $ in this case and the value of $ n $ may limit which of these highest weight states are possible.
	When evolving the time of this genuine highest weight state, 
	we encounter cases that would not be covered by Lemma~\ref{lem:classRmat}.
	So, the proof presented in Section~\ref{sec:ht1isomorphism} does not work for $ m=r $ 
	(at least, not without major modifications).
	
	However, Theorem~\ref{thm:scattering} holds for the $ m=r $ case.
	\begin{ex}
		Consider a state composed of elements from the $ U_q(\agl(2|4)) $ crystal $ B^{2,1} $
		(so that $ m=r=2 $).
		\begin{align*}
			t = 0 & \quad
			\begin{tikzpicture}[scale = 0.9, baseline = -2,every node/.style={scale=1}]
				\def\s{0.5};
				\foreach \x in {-2,-1,3,4,7,8,...,21} {
					\draw[fill=black] (\x*\s,0) circle (0.025);
				}
				\foreach \count/\valone/\valtwo/\valthree/\valfour/\valfive in {0/4/3/1/2/1,1/\btw/\btw/\btw/\bon/\btw}{
					\node (a) at (0*\s,0 + \count*\s) {$\valone$};
					\node (b) at (1*\s,0 + \count*\s) {$\valtwo$};
					\node (c) at (2*\s,0 + \count*\s) {$\valthree$};
					\node (d) at (5*\s,0 + \count*\s) {$\valfour$};
					\node (d) at (6*\s,0 + \count*\s) {$\valfive$};
				}
			\end{tikzpicture}
			\\[7pt]
			t = 1 & \quad
			\begin{tikzpicture}[scale = 0.9, baseline = -2,every node/.style={scale=1}]
				\def\s{0.5};
				\foreach \x in {-2,...,2,6,9,10,...,21} {
					\draw[fill=black] (\x*\s,0) circle (0.025);
				}
				\foreach \count/\valone/\valtwo/\valthree/\valfour/\valfive in {0/4/3/1/2/1,1/\btw/\btw/\btw/\bon/\btw}{
					\node (a) at (3*\s,0 + \count*\s) {$\valone$};
					\node (b) at (4*\s,0 + \count*\s) {$\valtwo$};
					\node (c) at (5*\s,0 + \count*\s) {$\valthree$};
					\node (d) at (7*\s,0 + \count*\s) {$\valfour$};
					\node (d) at (8*\s,0 + \count*\s) {$\valfive$};
				}
			\end{tikzpicture}
			\allowdisplaybreaks \\[7pt]
			t = 2 & \quad
			\begin{tikzpicture}[scale = 0.9, baseline = -2,every node/.style={scale=1}]
				\def\s{0.5};
				\foreach \x in {-2,...,5,11,12,...,21} {
					\draw[fill=black] (\x*\s,0) circle (0.025);
				}
				\foreach \count/\valone/\valtwo/\valthree/\valfour/\valfive in {0/4/3/1/2/1,1/\btw/\btw/\btw/\bon/\btw}{
					\node (a) at (6*\s,0 + \count*\s) {$\valone$};
					\node (b) at (7*\s,0 + \count*\s) {$\valtwo$};
					\node (c) at (8*\s,0 + \count*\s) {$\valthree$};
					\node (d) at (9*\s,0 + \count*\s) {$\valfour$};
					\node (d) at (10*\s,0 + \count*\s) {$\valfive$};
				}
			\end{tikzpicture}
			\allowdisplaybreaks \\[7pt]
			t = 3 & \quad
			\begin{tikzpicture}[scale = 0.9, baseline = -2,every node/.style={scale=1}]
				\def\s{0.5};
				\foreach \x in {-2,...,8,14,15,...,21} {
					\draw[fill=black] (\x*\s,0) circle (0.025);
				}
				\foreach \count/\valone/\valtwo/\valthree/\valfour/\valfive in {0/4/1/3/2/1,1/\btw/\btw/\bon/\btw/\btw}{
					\node (a) at (9*\s,0 + \count*\s) {$\valone$};
					\node (b) at (10*\s,0 + \count*\s) {$\valtwo$};
					\node (c) at (11*\s,0 + \count*\s) {$\valthree$};
					\node (d) at (12*\s,0 + \count*\s) {$\valfour$};
					\node (d) at (13*\s,0 + \count*\s) {$\valfive$};
				}
			\end{tikzpicture}
			\allowdisplaybreaks \\[7pt]
			t = 4 & \quad
			\begin{tikzpicture}[scale = 0.9, baseline = -2,every node/.style={scale=1}]
				\def\s{0.5};
				\foreach \x in {-2,...,10,13,17,18,...,21} {
					\draw[fill=black] (\x*\s,0) circle (0.025);
				}
				\foreach \count/\valone/\valtwo/\valthree/\valfour/\valfive in {0/4/1/3/2/1,1/\btw/\btw/\bon/\btw/\btw}{
					\node (a) at (11*\s,0 + \count*\s) {$\valone$};
					\node (b) at (12*\s,0 + \count*\s) {$\valtwo$};
					\node (c) at (14*\s,0 + \count*\s) {$\valthree$};
					\node (d) at (15*\s,0 + \count*\s) {$\valfour$};
					\node (d) at (16*\s,0 + \count*\s) {$\valfive$};
				}
			\end{tikzpicture}
			\\[7pt]
			t = 5 & \quad
			\begin{tikzpicture}[scale = 0.9, baseline = -2,every node/.style={scale=1}]
				\def\s{0.5};
				\foreach \x in {-2,...,12,15,16,20,21} {
					\draw[fill=black] (\x*\s,0) circle (0.025);
				}
				\foreach \count/\valone/\valtwo/\valthree/\valfour/\valfive in {0/4/1/3/2/1,1/\btw/\btw/\bon/\btw/\btw}{
					\node (a) at (13*\s,0 + \count*\s) {$\valone$};
					\node (b) at (14*\s,0 + \count*\s) {$\valtwo$};
					\node (c) at (17*\s,0 + \count*\s) {$\valthree$};
					\node (d) at (18*\s,0 + \count*\s) {$\valfour$};
					\node (d) at (19*\s,0 + \count*\s) {$\valfive$};
				}
			\end{tikzpicture}
		\end{align*}
		Note that this state satisfies the assumptions of Theorem~\ref{thm:scattering}
		with $ m=r $.
		We observe that this state behaves solitonically.
		Moreover $ R(\up{V}\otimes\up{W}) = R(
		\begin{array}{|c|c|c|}
			\hline
			\raisebox{-1pt}{$\btw$}
			& \raisebox{-1pt}{$\btw$}
			& \raisebox{-1pt}{$\btw$}\\
			\hline
		\end{array}
		\otimes
		\begin{array}{|c|c|}
			\hline
			\raisebox{-1pt}{$\btw$}
			& \raisebox{-1pt}{$\bon$}\\
			\hline
		\end{array}
		) = 
		\begin{array}{|c|c|}
			\hline
			\raisebox{-1pt}{$\btw$}
			& \raisebox{-1pt}{$\btw$}\\
			\hline
		\end{array}
		\otimes
		\begin{array}{|c|c|c|}
			\hline
			\raisebox{-1pt}{$\btw$}
			& \raisebox{-1pt}{$\btw$}
			& \raisebox{-1pt}{$\bon$}\\
			\hline
		\end{array}
		$
		and $ R(\down{V}\otimes\down{W}) = R(
		\begin{array}{|c|c|c|}
			\hline
			1 & 3 & 4\\
			\hline
		\end{array}
		\otimes
		\begin{array}{|c|c|}
			\hline
			1 & 2\\
			\hline
		\end{array}
		) = 
		\begin{array}{|c|c|}
			\hline
			1 & 4\\
			\hline
		\end{array}
		\otimes
		\begin{array}{|c|c|c|}
			\hline
			1 & 2 & 3\\
			\hline
		\end{array}
		$.
		The phase shift is $ 2 = 4 - 1 - 1 = 2\mathfrak{d}_2 + H(\up{V}\otimes\up{W}) + H(\down{V}\otimes\down{W})$.
	\end{ex}
	
	We believe that it would be possible to prove this theorem by reducing it to a modified version of the~\cite{HI00} system by using something similar to an $ L $-sped BBS. However, we would likely need to modify the speed of both solitons (as opposed to only of them in the $ L $-sped BBS), which would make the proof much more techinical.
	So, we instead prove the theorem by direct computation.
	
	We first compute general results about the carrier,
	similar to Lemma~\ref{lem:classRmat}.
	\begin{lemma}\label{lem:m=rclassRmat}
		Take some $ \ell>0 $.
		Let $ (\mathbf{x}_j)_{j=1}^{b_1+a_1} $ be a sequence
		such that $ 1 \leq \mathbf{x}_1<\mathbf{x}_2<\cdots < \mathbf{x}_{b_2+a_1} $.
		Let $ (\mathbf{h}_j)_{j=1}^\ell=\btw^{a_3}\bon^{a_2} (\mathbf{x}_j)_{j=b_1+1}^{b_1+a_1} $ for some $ a_3,a_2,a_1 $ such that $ a_3+a_2+a_1=\ell $.
		Let $ (\mathbf{z}_j)_{j=1}^\ell = \bon^{b_2}(\mathbf{x}_j)_{j=1}^{b_1}(\mathbf{z}_j)_{j=b_2+b_1+1}^\ell $ 
		for some $ b_2,b_1 $ such that $ b_2+b_1\le \ell $, 
		with $ \mathbf{x}_{b_1} < \mathbf{z}_{b_2+b_1+1} < \cdots < \mathbf{z}_{\ell} $.
		
		For
		\[
		b=
		\begin{array}{|c|}
			\hline
			\bfX^\ell\\
			(\mathbf{h}_j)_{j=1}^\ell\\
			(\mathbf{z}_j)_{j=1}^\ell\\
			\hline
		\end{array}
		\otimes 
		\begin{array}{|c|}
			\hline
			\bfX\\
			\mathbf{x}_0 \\
			\mathbf{y} \\
			\hline
		\end{array}
		\]
		with $ 1\leq \mathbf{x}_0<\mathbf{x}_1 $ and $ \mathbf{x}_{b_1}\le \mathbf{y}<\mathbf{z}_{b_2+b_1+1} $, we have that $ R(b)= $
		\begin{enumerate}[label = \upshape{($ n $\alph*)}]
			\item \label{enum:fa}
			$
			\begin{array}{|c|}
				\hline
				\bfX\\
				\btw\\
				\bon\\
				\hline
			\end{array}
			\otimes 
			\begin{array}{|c|}
				\hline
				\bfX^\ell\\
				\btw^{a_3-1}\bon^{a_2}(\mathbf{x}_j)_{j=b_1}^{b_1+a_1}\\
				\bon^{b_2-1}(\mathbf{x}_j)_{j=0}^{b_1-1} \mathbf{y}(\mathbf{z}_j)_{j=b_2+b_1+1}^\ell\\
				\hline
			\end{array}
			$
			if $ b_2>0 $;
			\item \label{enum:fb}
			$
			\begin{array}{|c|}
				\hline
				\bfX\\
				\mathbf{h}_{a_3+a_2}\\
				\mathbf{z}_\ell\\
				\hline
			\end{array}
			\otimes 
			\begin{array}{|c|}
				\hline
				\bfX^\ell\\
				(\mathbf{h}_j)_{j=1}^{a_3+a_2-1}(\mathbf{x}_j)_{j=b_1}^{b_1+a_1}\\
				(\mathbf{x}_j)_{j=0}^{b_1-1} \mathbf{y} (\mathbf{z}_j)_{j=b_2+b_1+1}^{\ell-1}\\
				\hline
			\end{array}
			$
			if $ b_2=0 $ and $ b_2+b_1<a_3+a_2 $.
		\end{enumerate}
		
		For
		\[
		b=
		\begin{array}{|c|}
			\hline
			\bfX^\ell\\
			(\mathbf{h}_j)_{j=1}^\ell\\
			(\mathbf{z}_j)_{j=1}^\ell\\
			\hline
		\end{array}
		\otimes 
		\begin{array}{|c|}
			\hline
			\bfX\\
			\bon \\
			\mathbf{y}\\
			\hline
		\end{array}
		\]
		with $ 1\leq\mathbf{y}<\mathbf{z}_{b_2+b_1+1} $, we have that $ R(b)= $
		\begin{enumerate}[label=\upshape{($ \bon $\alph*)}]
			\item \label{enum:bona}
			$
			\begin{array}{|c|}
				\hline
				\bfX\\
				\mathbf{h}_{a_3+a_2}\\
				\mathbf{z}_\ell\\
				\hline
			\end{array}
			\otimes 
			\begin{array}{|c|}
				\hline
				\bfX^\ell\\
				(\mathbf{h}_j)_{j=1}^{a_3+a_2-1}(\mathbf{x}_j)_{j=b_1}^{b_1+a_1}\\
				\bon^{b_2+1}(\mathbf{x}_j)_{j=1}^{b_1-1}\mathbf{y}(\mathbf{z}_j)_{j=b_2+b_1+1}^{\ell-1}\\
				\hline
			\end{array}
			$
			\parbox{18em}{if $ b_2<a_3 $ and $ b_2+b_1<a_3+a_2 $ and $ b_1>0 $ and $ \mathbf{y}\ge \mathbf{x}_{b_1} $;}
			\item \label{enum:bonb}
			$
			\begin{array}{|c|}
				\hline
				\bfX\\
				\btw\\
				\mathbf{z}_\ell\\
				\hline
			\end{array}
			\otimes 
			\begin{array}{|c|}
				\hline
				\bfX^\ell\\
				\btw^{a_3-1}\bon^{a_2+1}(\mathbf{x}_j)_{j=b_1+1}^{b_1+a_1}\\
				\bon^{b_2} \mathbf{y} (\mathbf{x}_j)_{j=1}^{b_1} (\mathbf{z}_j)_{j=b_2+b_1+1}^{\ell-1}\\
				\hline
			\end{array}
			$
			if $ b_2<a_3 $ and $ b_2+b_1<a_3+a_2 $ and $ \mathbf{y}<\mathbf{z}_{b_2+1} $;
			\item \label{enum:bonc}
			$
			\begin{array}{|c|}
				\hline
				\bfX\\
				\mathbf{h}_{a_3+a_2}\\
				\mathbf{x}_{b_1+a_1}\\
				\hline
			\end{array}
			\otimes 
			\begin{array}{|c|}
				\hline
				\bfX^\ell\\
				(\mathbf{h}_j)_{j=1}^{a_3+a_2-1}(\mathbf{x}_j)_{j=b_1-1}^{b_1+a_1-1}\\
				\bon^{b_2+1}(\mathbf{x}_j)_{j=1}^{b_1-2}\mathbf{y}(\mathbf{z}_j)_{j=b_2+b_1+1}^\ell\\
				\hline
			\end{array}
			$
			\parbox{18em}{if $ b_2<a_3 $ and $ b_2+b_1=a_3+a_2 $ and $ b_1>1 $ and $ \mathbf{y}\ge \mathbf{x}_{b_1} $;}
			\item \label{enum:bond}
			$
			\begin{array}{|c|}
				\hline
				\bfX\\
				\btw\\
				\mathbf{x}_{b_1+a_1}\\
				\hline
			\end{array}
			\otimes 
			\begin{array}{|c|}
				\hline
				\bfX^\ell\\
				\btw^{a_3-1}\bon^{a_2+1}(\mathbf{x}_j)_{j=b_1}^{b_1+a_1-1}\\
				\bon^{b_2} \mathbf{y} (\mathbf{x}_j)_{j=1}^{b_1-1}(\mathbf{z}_j)_{j=b_2+b_1+1}^\ell\\
				\hline
			\end{array}
			$
			\parbox{18em}{if $ b_2<a_3 $ and $ b_2+b_1=a_3+a_2 $ and ($ b_1=1 $ or $ \mathbf{y}<\mathbf{x}_1 $).}
			\item \label{enum:bone}
			$
			\begin{array}{|c|}
				\hline
				\bfX\\
				\btw\\
				\mathbf{x}_{b_1+a_1}\\
				\hline
			\end{array}
			\otimes 
			\begin{array}{|c|}
				\hline
				\bfX^\ell\\
				\btw^{a_3-1}\bon^{a_2+2}(\mathbf{x}_j)_{j=b_1+1}^{b_1+a_1-1}\\
				\bon^{b_2-1} \mathbf{y} (\mathbf{x}_j)_{j=1}^{b_1} (\mathbf{z}_j)_{j=b_2+b_1+1}^\ell\\
				\hline
			\end{array}
			$
			if $ 0<b_2=a_3 $ and $ a_1>0 $ and $ \mathbf{y}<\mathbf{z}_{b_2+1} $;
			\item \label{enum:bonf}
			$
			\begin{array}{|c|}
				\hline
				\bfX\\
				\btw\\
				\bon\\
				\hline
			\end{array}
			\otimes 
			\begin{array}{|c|}
				\hline
				\bfX^\ell\\
				\btw^{a_3-1}\bon^{a_2+1}\\
				\bon^{b_2-1}\mathbf{y} (\mathbf{x}_j)_{j=1}^{b_1} (\mathbf{z}_j)_{j=b_2+b_1+1}^\ell\\
				\hline
			\end{array}
			$
			if $ 0<b_2=a_3 $ and $ a_1=0 $ and $ \mathbf{y}<\mathbf{z}_{b_2+1} $;
			\item \label{enum:bong}
			$
			\begin{array}{|c|}
				\hline
				\bfX\\
				\mathbf{h}_\ell\\
				\mathbf{z}_\ell\\
				\hline
			\end{array}
			\otimes 
			\begin{array}{|c|}
				\hline
				\bfX^\ell\\
				\bon(\mathbf{h}_j)_{j=1}^{\ell-1}\\
				\mathbf{y}(\mathbf{z}_j)_{j=1}^{\ell-1}\\
				\hline
			\end{array}
			$
			if $ b_2=a_3=0 $ and $ \mathbf{y}<\mathbf{z}_{b_2+1} $.
		\end{enumerate}
		
		For
		\[
		b=
		\begin{array}{|c|}
			\hline
			\bfX^\ell\\
			(\mathbf{h}_j)_{j=1}^\ell\\
			(\mathbf{z}_j)_{j=1}^\ell\\
			\hline
		\end{array}
		\otimes 
		\begin{array}{|c|}
			\hline
			\bfX\\
			\btw \\
			\mathbf{y}\\
			\hline
		\end{array}
		\]
		with $ \bon\le \mathbf{y}<\mathbf{z}_{b_2+b_1+1} $, we have that $ R(b)= $
		\begin{enumerate}[label=\upshape{($ \btw $\alph*)}]
			\item \label{enum:btwa}
			$
			\begin{array}{|c|}
				\hline
				\bfX\\
				\mathbf{h}_{a_3+a_2}\\
				\mathbf{x}_{b_1+a_1}\\
				\hline
			\end{array}
			\otimes 
			\begin{array}{|c|}
				\hline
				\bfX^\ell\\
				\btw(\mathbf{h}_j)_{j=1}^{a_3+a_2-1}(\mathbf{x}_j)_{j=b_1}^{b_1+a_1-1}\\
				\bon^{b_2} (\mathbf{x}_j)_{j=1}^{b_1-1}\mathbf{y}(\mathbf{z}_j)_{j=b_2+b_1+1}^\ell\\
				\hline
			\end{array}
			$
			if $ b_1>0 $ and $ \mathbf{y}\ge \mathbf{x}_{b_1} $;
			\item \label{enum:btwb}
			$
			\begin{array}{|c|}
				\hline
				\bfX\\
				\btw\\
				\mathbf{x}_{b_1+a_1}\\
				\hline
			\end{array}
			\otimes 
			\begin{array}{|c|}
				\hline
				\bfX^\ell\\
				\btw^{a_3}\bon^{a_2+1}(\mathbf{x}_j)_{j=b_1+1}^{b_1+a_1-1}\\
				\bon^{b_2-1} \mathbf{y} (\mathbf{x}_j)_{j=1}^{b_1} (\mathbf{z}_j)_{j=b_2+b_1+1}^\ell\\
				\hline
			\end{array}
			$
			if $ b_2>0 $ and $ a_1>0 $ and $ 1\le \mathbf{y} < \mathbf{z}_{b_2+1} $.
			\item \label{enum:btwc}
			$
			\begin{array}{|c|}
				\hline
				\bfX\\
				\btw\\
				\bon\\
				\hline
			\end{array}
			\otimes 
			\begin{array}{|c|}
				\hline
				\bfX^\ell\\
				\btw^{a_3}\bon^{a_2}\\
				\bon^{b_2-1} \mathbf{y} (\mathbf{x}_j)_{j=1}^{b_1} (\mathbf{z}_j)_{j=b_2+b_1+1}^\ell\\
				\hline
			\end{array}
			$
			if $ b_2>0 $ and $ a_1=0 $ and $ 1\le \mathbf{y} <\mathbf{z}_{b_2+1} $
			\item \label{enum:btwd}
			$
			\begin{array}{|c|}
				\hline
				\bfX\\
				\mathbf{h}_\ell\\
				\mathbf{z}_\ell\\
				\hline
			\end{array}
			\otimes 
			\begin{array}{|c|}
				\hline
				\bfX^\ell\\
				\btw(\mathbf{h}_j)_{j=1}^{\ell-1}\\
				\mathbf{y}(\mathbf{z}_j)_{j=1}^{\ell-1}\\
				\hline
			\end{array}
			$
			if $ \mathbf{y}=\bon $ or ($ b_2=0 $ and $ 1\le \mathbf{y} < \mathbf{z}_{b_2+1} $).
		\end{enumerate}
	\end{lemma}
	See Appendix~\ref{appendix:m=rclassRmatprf} for the proof.
	
	\begin{lemma}\label{lem:m=roverlap}
		Assume $ Q\le M,L < \ell \le d_2+1 $.
		
		Let $ P=P_1\otimes P_2\otimes P_3 \otimes \cdots $ be the state
		\[
		\begin{array}{|c|}
			\hline
			\btw\\
			x_1\\
			\hline
		\end{array}
		\otimes \cdots \otimes 
		\begin{array}{|c|}
			\hline
			\btw\\
			x_{d_1-Q}\\
			\hline
		\end{array}
		\otimes
		\begin{array}{|c|}
			\hline
			x_{d_1-Q+1}\\
			y_1\\
			\hline
		\end{array}
		\otimes \cdots \otimes 
		\begin{array}{|c|}
			\hline
			x_{d_1}\\
			y_Q\\
			\hline
		\end{array}
		\otimes
		\begin{array}{|c|}
			\hline
			\bon\\
			y_{Q+1}\\
			\hline
		\end{array}
		\otimes \cdots \otimes 
		\begin{array}{|c|}
			\hline
			\bon\\
			y_L\\
			\hline
		\end{array}
		\otimes
		\begin{array}{|c|}
			\hline
			\btw\\
			y_{L+1}\\
			\hline
		\end{array}
		\otimes \cdots \otimes 
		\begin{array}{|c|}
			\hline
			\btw\\
			y_{d_2}\\
			\hline
		\end{array}
		\otimes u_1^{\otimes \infty}
		\]
		where
		\begin{align*}
			x_j =
			\begin{cases}
				d_2+d_1-M+1-j&\text{if}~j\le d_1-M,\\
				d_1+1-j&\text{if}~j> d_1-M,
			\end{cases},
			&&
			y_j = 
			d_2+1-j.
		\end{align*}
		Let $ \widetilde{C}=\min(d_1,\ell) $. Then,
		\begin{enumerate}[label=(\roman*)]
			\item \label{enum:m=rmoreoverlap}
			if $ d_2=\ell-1 $, $ d_1\ge\ell $ and $ L-Q>\widetilde{C}-M-1 $ then $ T_{\ell}(P)= $
			\[
			u_1^{\otimes \widetilde{C}}\otimes
			\begin{array}{|c|}
				\hline
				\btw\\
				x_1\\
				\hline
			\end{array}
			\otimes \cdots \otimes 
			\begin{array}{|c|}
				\hline
				\btw\\
				x_{d_1-Q-1}\\
				\hline
			\end{array}
			\otimes
			\begin{array}{|c|}
				\hline
				x_{d_1-Q}\\
				y_1\\
				\hline
			\end{array}
			\otimes \cdots \otimes 
			\begin{array}{|c|}
				\hline
				x_{d_1}\\
				y_{Q+1}\\
				\hline
			\end{array}
			\otimes
			\begin{array}{|c|}
				\hline
				\bon\\
				y_{Q+2}\\
				\hline
			\end{array}
			\otimes \cdots \otimes 
			\begin{array}{|c|}
				\hline
				\bon\\
				y_L\\
				\hline
			\end{array}
			\otimes
			\begin{array}{|c|}
				\hline
				\btw\\
				y_{L+1}\\
				\hline
			\end{array}
			\otimes \cdots \otimes 
			\begin{array}{|c|}
				\hline
				\btw\\
				y_{d_2}\\
				\hline
			\end{array}
			\otimes u_1^{\otimes \infty};
			\]
			\item \label{enum:m=rsameoverlap}
			if $ d_1\ge\ell $ and $ L-Q=\widetilde{C}-M-1 $ then $ T_{\ell}(P)= $
			\begin{align*}
				&
				u_1^{\otimes \widetilde{C}}\otimes
				\begin{array}{|c|}
					\hline
					\btw\\
					x_1\\
					\hline
				\end{array}
				\otimes \cdots \otimes 
				\begin{array}{|c|}
					\hline
					\btw\\
					x_{d_1-M-1}\\
					\hline
				\end{array}
				\otimes
				\begin{array}{|c|}
					\hline
					\btw\\
					x_{d_1-M+1}\\
					\hline
				\end{array}
				\otimes\cdots\otimes
				\begin{array}{|c|}
					\hline
					\btw\\
					x_{d_1-Q}\\
					\hline
				\end{array}
				\\&\quad
				\otimes
				\begin{array}{|c|}
					\hline
					x_{d_1-Q+1}\\
					x_{d_1-M}\\
					\hline
				\end{array}
				\otimes
				\begin{array}{|c|}
					\hline
					x_{d_1-Q+2}\\
					y_1\\
					\hline
				\end{array}
				\otimes \cdots \otimes 
				\begin{array}{|c|}
					\hline
					x_{d_1}\\
					y_{Q-1}\\
					\hline
				\end{array}
				\otimes
				\begin{array}{|c|}
					\hline
					\bon\\
					y_Q\\
					\hline
				\end{array}
				\otimes \cdots \otimes 
				\begin{array}{|c|}
					\hline
					\bon\\
					y_{L-1}\\
					\hline
				\end{array}
				\otimes
				\begin{array}{|c|}
					\hline
					\btw\\
					y_L\\
					\hline
				\end{array}
				\otimes \cdots \otimes 
				\begin{array}{|c|}
					\hline
					\btw\\
					y_{d_2}\\
					\hline
				\end{array}
				\otimes u_1^{\otimes \infty};
			\end{align*}
			\item \label{enum:m=rlessoverlap}
			if $ Q>0 $, $ d_1=\ell-1 $, $ d_2\ge\ell $ and $ L-Q>\widetilde{C}-M-1 $ then
			$ T_{\ell}(P)= $
			\[
			u_1^{\otimes \widetilde{C}}\otimes
			\begin{array}{|c|}
				\hline
				\btw\\
				x_1\\
				\hline
			\end{array}
			\otimes \cdots \otimes 
			\begin{array}{|c|}
				\hline
				\btw\\
				x_{d_1-Q}\\
				\hline
			\end{array}
			\otimes
			\begin{array}{|c|}
				\hline
				x_{d_1-Q+1}\\
				y_1\\
				\hline
			\end{array}
			\otimes \cdots \otimes 
			\begin{array}{|c|}
				\hline
				x_{d_1}\\
				y_{Q-1}\\
				\hline
			\end{array}
			\otimes
			\begin{array}{|c|}
				\hline
				\bon\\
				y_{Q}\\
				\hline
			\end{array}
			\otimes \cdots \otimes 
			\begin{array}{|c|}
				\hline
				\bon\\
				y_L\\
				\hline
			\end{array}
			\otimes
			\begin{array}{|c|}
				\hline
				\btw\\
				y_{L+1}\\
				\hline
			\end{array}
			\otimes \cdots \otimes 
			\begin{array}{|c|}
				\hline
				\btw\\
				y_{d_2}\\
				\hline
			\end{array}
			\otimes u_1^{\otimes \infty}.
			\]		
		\end{enumerate}
	\end{lemma}
	\begin{proof}
		Let $ T_{\ell}(P) = \widetilde{P}_1\otimes \widetilde{P}_2 \otimes \widetilde{P}_3 \otimes \cdots $
		and let $ u(j) $ be the carrier at step $ j $,
		i.e.\ $ R_{j-1}\cdots R_1 R_0(P) = \widetilde{P}_1\otimes \widetilde{P}_2\otimes\cdots\otimes\widetilde{P}_j \otimes u(j)\otimes P_{j+1} \otimes \widetilde{P}_{j+2}\otimes\cdots $.
		Let
		\[
		u(j) = 
		\begin{array}{|c|}
			\hline
			\btw^{A_3(j)}\bon^{A_2(j)}(\mathbf{x}^j_k)_{k=B_1(j)+1}^{B_1(j)+A_1(j)}\\
			\bon^{B_2(j)}(\mathbf{x}^j_k)_{k=1}^{B_1(j)}(\mathbf{z}^j_k)_{k=B_2(j)+B_1(j)+1}^{\ell}\\
			\hline
		\end{array}
		\]
		for some sequences $ (\mathbf{x}^j_k)_{k=1}^{B_1(j)+A_1(j)} $ and $ (\mathbf{z}^j_k)_{k=B_2(j)+B_1(j)+1}^{\ell} $ such that $ 1\leq \mathbf{x}^j_1 < \cdots < \mathbf{x}^j_{B_1(j)+A_1(j)} $  and $ 1\leq \mathbf{z}^j_{B_2(j)+B_1(j)+1} < \cdots < \mathbf{z}^j_{\ell} $
		where $ A_3,A_2,A_1,B_2,B_1 $ are functions on non-negative integers.
		We will often write $ A_3 $ for $ A_3(j) $, with $ j $ being clear from context.
		Similarly for the other functions.
		Note that 
		\[
		u(0) = u_{\ell} = 
		\begin{array}{|c|}
			\hline
			\btw^{\ell}\\
			\bon^{\ell}\\
			\hline
		\end{array}.
		\]
		Starting with $ u(0) $, we can apply~\ref{enum:btwc} $ \min(d_1-Q,\ell) $ times:
		\begin{align*}
			\widetilde{P}_1\otimes\cdots\otimes\widetilde{P}_{\min(d_1-Q,\ell)} &= u_1^{\otimes \min(d_1-Q,\ell)}\\
			u(\min(d_1-Q,\ell)) &= 
			\begin{array}{|c|}
				\hline
				\btw^{\ell}\\
				\bon^{\ell-\min(d_1-Q,\ell)}(x_{\min(d_1-Q,\ell)-j})_{j=0}^{\min(d_1-Q,\ell)-1}\\
				\hline
			\end{array}
			\,.
		\end{align*}
		If $ \ell<d_1-Q $ we then apply~\ref{enum:btwd} $ d_1-Q-\ell $ times.
		Therefore,
		\begin{align*}
			\widetilde{P}_1\otimes\cdots\otimes\widetilde{P}_{d_1-Q} 
			&= u_1^{\otimes \min(d_1-Q,\ell)}
			\otimes
			\begin{array}{|c|}
				\hline
				\btw\\
				x_1\\
				\hline
			\end{array}
			\otimes\cdots\otimes
			\begin{array}{|c|}
				\hline
				\btw\\
				x_{d_1-Q-\min(d_1-Q,\ell)}\\
				\hline
			\end{array}
			\\
			u(d_1-Q) &= 
			\begin{array}{|c|}
				\hline
				\btw^{\ell}\\
				\bon^{\ell-\min(d_1-Q,\ell)}(x_{d_1-Q-j})_{j=0}^{\min(d_1-Q,\ell)-1}\\
				\hline
			\end{array}
		\end{align*}
		(noting that $ \min(d_1-Q,\ell) = \ell $ when $ \ell < d_1-Q $).
		Additionally, the above equalities clearly hold if $ \ell \geq d_1-Q $.
		We can then apply~\ref{enum:fa} (but only if $ d_1-Q<\ell $)
		and then~\ref{enum:fb} to get:
		\begin{align*}
			\widetilde{P}_1\otimes\cdots\otimes\widetilde{P}_{d_1} 
			&= u_1^{\otimes \min(d_1,\ell)}
			\otimes
			\begin{array}{|c|}
				\hline
				\btw\\
				x_1\\
				\hline
			\end{array}
			\otimes\cdots\otimes
			\begin{array}{|c|}
				\hline
				\btw\\
				x_{d_1-\min(d_1,\ell)}\\
				\hline
			\end{array}
			\\
			u(d_1) &= 
			\begin{array}{|c|}
				\hline
				\btw^{\ell-Q}(x_{d_1-M+Q-j})_{j=0}^{Q-1}\\
				\bon^{\ell-\min(d_1,\ell)}(x_{d_1-j})_{j=0}^{M-Q-1}(y_{Q-j})_{j=0}^{Q-1}(x_{d_1-M-j})_{j=0}^{\min(d_1,\ell)-M-1}\\
				\hline
			\end{array}
			\,.
		\end{align*}
		(Note that, during~\ref{enum:fb}, we always have $ B_2+B_1<A_3+A_2 $.
		Indeed, $ B_1+A_1\le M < \ell $ and $ B_2=0 $,
		so $ B_2+B_1 = B_1 < \ell-A_1 = A_3+A_2 $.)
		Observe that $ M-\min(d_1,\ell)\le 0 $ so
		$ \ell-\min(d_1,\ell) + M - Q \le \ell-Q $.
		
		Let $ \widetilde{C}=\min(d_1,\ell) $. Now we investigate moving the carrier
		starting from $ u(d_1) $ and pausing when we reach $ u(d_1+L-Q) $.
		Case~\ref{enum:bona} applies 
		until $ A_3=B_2+B_1 $ or $ B_1=0 $ (or until $ u(d_1+L-Q) $):
		\begin{itemize}
			\item Assume $ A_3=B_2+B_1 = \ell-\widetilde{C}+M-Q $ 
			(note that~\ref{enum:bona} does not change the value of $ B_2+B_1 $).
			Then, we have applied~\ref{enum:bona} $ \widetilde{C}-M $ times.
			Therefore, 
			\begin{equation}\label{eq:m=r cd at d1-M}
				\begin{aligned}
					\widetilde{P}_1\otimes\cdots\otimes\widetilde{P}_{d_1+\widetilde{C}-M} 
					&= u_1^{\otimes \widetilde{C}}
					\otimes
					\begin{array}{|c|}
						\hline
						\btw\\
						x_1\\
						\hline
					\end{array}
					\otimes\cdots\otimes
					\begin{array}{|c|}
						\hline
						\btw\\
						x_{d_1-M}\\
						\hline
					\end{array}\,,
					\\
					u(d_1+\widetilde{C}-M) &= 
					\begin{array}{|c|}
						\hline
						\btw^{\ell-\widetilde{C}+M-Q}(x_{d_1+\widetilde{C}-2M+Q-j})_{j=0}^{\widetilde{C}-M+Q-1}\\
						\bon^{\ell-M}(x_{d_1-j})_{j=0}^{2M-Q-\widetilde{C}-1}
						(y_{\widetilde{C}-M+Q-j})_{j=0}^{\widetilde{C}-M+Q-1}\\
						\hline
					\end{array}
					\,.
				\end{aligned}	
			\end{equation}
			Then we apply~\ref{enum:bonc} until $ B_2=A_3 $ or $ B_1=1 $.
			If $ B_2=A_3 $ we apply~\ref{enum:bone}.
			If $ B_1=1 $ we apply~\ref{enum:bond} once,
			and then we have $ B_2=A_3 $ so we apply~\ref{enum:bone}.
			\item Assume $ B_1=0 $.
			We must have applied~\ref{enum:bona} $ M-Q $ times.
			Since~\ref{enum:bona} does not change the value of $ B_2+B_1 $,
			we have that $ B_2=\ell-\widetilde{C}+M-Q $.
			We then apply~\ref{enum:bonb} until $ A_3=B_2=\ell-\widetilde{C}+M-Q $.
			Since both~\ref{enum:bona} and~\ref{enum:bonb} decrement $ A_3 $ by $ 1 $,
			we have applied~\ref{enum:bona} and~\ref{enum:bonb} a total of $ \widetilde{C}-M $ times.
			Therefore,
			\begin{equation}\label{eq:m=r b at d1-M}
				\begin{aligned}
					\widetilde{P}_1\otimes\cdots\otimes\widetilde{P}_{d_1+\widetilde{C}-M} 
					&= u_1^{\otimes \widetilde{C}}
					\otimes
					\begin{array}{|c|}
						\hline
						\btw\\
						x_1\\
						\hline
					\end{array}
					\otimes\cdots\otimes
					\begin{array}{|c|}
						\hline
						\btw\\
						x_{d_1-M}\\
						\hline
					\end{array}\,,
					\\
					u(d_1+\widetilde{C}-M) &= 
					\begin{array}{|c|}
						\hline
						\btw^{\ell-\widetilde{C}+M-Q}\bon^{\widetilde{C}-2M+Q}(x_{d_1-j})_{j=0}^{M-1}\\
						\bon^{\ell-\widetilde{C}+M-Q}
						(y_{Q-j})_{j=0}^{\widetilde{C}-M+Q-1}\\
						\hline
					\end{array}
					\,.
				\end{aligned}
			\end{equation}
			Then we apply~\ref{enum:bone}.
		\end{itemize}
		We apply~\ref{enum:bone} until $ A_1=0 $ or $ B_2=A_3=0 $.
		Next, we apply~\ref{enum:bonf} until $ B_2=A_3=0 $.
		From then onwards,~\ref{enum:bong} applies.
		Note that we might not reach all of the above cases before reaching $ u(d_1+L-Q) $ and
		we might skip over cases.
		
		Now, we investigate moving the carrier starting from $ u(d_1+L-Q) $ 
		and pausing when we reach $ u(d_1-Q+d_2) $.
		Case~\ref{enum:btwa} applies until $ B_1=0 $,
		case~\ref{enum:btwb} applies until $ A_1=0 $,
		case~\ref{enum:btwc} applies until $ B_2=0 $,
		and case~\ref{enum:btwd} applies from then onwards.
		Once again, we might not reach all of the cases, and we might skip over cases.
		Note that~\ref{enum:btwa} can only be preceded by~\ref{enum:bona}, \ref{enum:bonc} and \ref{enum:bond}.
		\ref{enum:btwb} can \emph{not} be preceded by~\ref{enum:bonf} or~\ref{enum:bong}.
		\ref{enum:btwc} can \emph{not} be preceded by~\ref{enum:bong}.
		
		After $ u(d_1-Q+d_2) $,~\ref{enum:btwd} puts down the columns of the carrier in reverse order.
		
		Note that~\ref{enum:bona}, \ref{enum:bonb}, \ref{enum:bonc}, \ref{enum:bond}, \ref{enum:bone}, \ref{enum:btwa}, \ref{enum:btwb} do not change the value of $ A_1-B_2 $.
		Therefore, $ A_1-B_2 $ remains constant at $ A_1-B_2=\widetilde{C}-\ell+Q $ in these cases.
		Note also that~\ref{enum:bona}, \ref{enum:bonb}, \ref{enum:bonc}, \ref{enum:bond}, \ref{enum:bone} decrease the value of $ A_3 $,
		while~\ref{enum:btwa}, \ref{enum:btwb}, \ref{enum:btwc} do not change $ A_3 $.
		So, if~\ref{enum:bonf} and~\ref{enum:bong} do not apply, then $ A_3 = \ell-L $ in the carriers between $ u(d_1-Q+L) $ and $ u(d_1-Q+d_2) $.
		
		Assume $ d_2=\ell-1 $, $ d_1\ge\ell $ and $ L-Q>\widetilde{C}-M-1 $.
		We have that $ d_1+\widetilde{C}-M \le d_1+L-Q $,
		So either~\eqref{eq:m=r cd at d1-M} or~\eqref{eq:m=r b at d1-M} applies.
		After $ u(d_1+\widetilde{C}-M) $, all of the applicable cases of Lemma~\ref{lem:m=rclassRmat}
		put down members from $ (x_{d_1-j})_{j=0}^{M-1} $ in order (with $ \btw $ in the row above).
		Observe that
		\[
		(d_1-Q+d_2)-(d_1+\widetilde{C}-M) = d_2-Q-\widetilde{C}+M = (\ell-1) - Q - \ell +M = M-Q-1.
		\]
		Therefore, $ (x_{d_1-Q-j})_{j=1}^{M-Q-1} $ are put down,
		while $ (x_{d_1-j})_{j=0}^Q $ remain in $ u(d_1-Q+d_2) $.
		That is,
		\[
		u(d_1-Q+d_2) =
		\begin{array}{|c|}
			\hline
			\btw^{\ell-L}\bon^{L-Q-1}(x_{d_1-j})_{j=0}^Q\\
			\bon(y_{d_2-j})_{j=0}^{d_2-1}\\
			\hline
		\end{array}\,.
		\]
		(Note that, in~\eqref{eq:m=r cd at d1-M}, $ 2M-Q-\widetilde{C} < M - Q $.
		Consequently, there are no $ x_j $ in the bottom row of $ u(d_1-Q+d_2) $.
		We can deduce the remaining structure of the carrier from the facts that $A_1-B_2 = \widetilde{C}-\ell+Q = Q $ and $ A_3 = \ell-L $.)
		Applying~\ref{enum:btwd} proves~\ref{enum:m=rmoreoverlap}.
		
		Assume $ d_1\ge\ell $ and $ L-Q=\widetilde{C}-M-1 $.
		Then $ d_1+L-Q = d_1+\widetilde{C}-M-1 $ 
		so~\eqref{eq:m=r cd at d1-M} and~\eqref{eq:m=r b at d1-M} do not apply
		(but the state immediately before them does).
		Thus, $ x_{d_1-M} $ is still in $ u(d_1+L-Q) $
		(and, incidentally, will remain in the carrier until we apply~\ref{enum:btwd}).
		After $ u(d_1+L-Q) $,~\ref{enum:btwa} and~\ref{enum:btwb} put down 
		members from $ (x_{d_1-j})_{j=0}^{M-1} $ in order (with $ \btw $ in the row above) until $ B_2=0 $.
		Observe that
		\[
		(d_1-Q+\ell-1)-(d_1+\widetilde{C}-M-1)
		=\ell-Q-\widetilde{C}+M
		\ge \ell - Q -\ell + M
		= M-Q.
		\]
		Therefore $ A_1(d_1-Q+\ell-1)=Q $.
		The fact that $ A_1-B_2=\widetilde{C}+Q-\ell = Q $
		implies $ B_2(d_1-Q+\ell-1)=0 $.
		Thus,
		\[
		u(d_1-Q+\ell-1) = 
		\begin{array}{|c|}
			\hline
			\btw^{\ell-L}\bon^{L-Q}(x_{d_1-j})_{j=0}^{Q-1}\\
			(y_{\ell-j})_{j=1}^{\ell-1}x_{d_1-M}\\
			\hline
		\end{array}\,
		\]
		(using also the fact that $ A_3 = \ell-L $).
		Applying~\ref{enum:btwd} proves~\ref{enum:m=rsameoverlap}.
		
		Assume $ Q>0 $, $ d_1=\ell-1 $, $ d_2\ge\ell $ and $ L-Q>\widetilde{C}-M-1 $.
		We have that $ d_1+\widetilde{C}-M \le d_1+L-Q $,
		So either~\eqref{eq:m=r cd at d1-M} or~\eqref{eq:m=r b at d1-M} applies.
		After $ u(d_1+\widetilde{C}-M) $, all of the applicable cases of Lemma~\ref{lem:m=rclassRmat}
		put down members from $ (x_{d_1-j})_{j=0}^{M-1} $ in order (with $ \btw $ in the row above).
		Observe that
		\[
		(d_1-Q+\ell) - (d_1+\widetilde{C}-M)
		= \ell-Q-\widetilde{C}+M
		\ge \ell-Q-\ell+1+M
		=M-Q+1.
		\]
		Therefore $ A_1(d_1-Q+\ell) = Q-1 $.
		The fact that $ A_1-B_2 = \widetilde{C}+Q-\ell = Q-1 $ 
		implies $ B_2(d_1-Q+\ell)=0 $.
		Thus,
		\[
		u(d_1-Q+\ell)=
		\begin{array}{|c|}
			\hline
			\btw^{\ell-L}\bon^{L-Q+1}(x_{d_1-j})_{j=0}^{Q-2}\\
			(y_{\ell-j})_{j=0}^{\ell-1}\\
			\hline
		\end{array}
		\]
		(using also the fact that $ A_3 = \ell-L$).
		Applying~\ref{enum:btwd} proves~\ref{enum:m=rlessoverlap}.
	\end{proof}
	
	\begin{lemma}\label{lem:m=rnooverlap}
		Assume $ Q\le M,L < \ell \le d_2+1 $ and $ C_2<\min(d_1,\ell) $.
		
		Let $ P=P_1\otimes P_2\otimes P_3 \otimes \cdots $ be the state
		\[
		P =
		\begin{array}{|c|}
			\hline
			\btw\\
			x_1\\
			\hline
		\end{array}
		\otimes \cdots \otimes 
		\begin{array}{|c|}
			\hline
			\btw\\
			x_{d_1}\\
			\hline
		\end{array}
		\otimes u_1^{\otimes C_2}\otimes 
		\begin{array}{|c|}
			\hline
			\bon\\
			y_1\\
			\hline
		\end{array}
		\otimes \cdots \otimes 
		\begin{array}{|c|}
			\hline
			\bon\\
			y_L\\
			\hline
		\end{array}
		\otimes
		\begin{array}{|c|}
			\hline
			\btw\\
			y_{L+1}\\
			\hline
		\end{array}
		\otimes \cdots \otimes
		\begin{array}{|c|}
			\hline
			\btw\\
			y_{d_2}\\
			\hline
		\end{array}
		\otimes u_1^{\otimes \infty },
		\]
		where
		\begin{align*}
			x_j =
			\begin{cases}
				d_2+d_1-M+1-j&\text{if}~j\le d_1-M,\\
				d_1+1-j&\text{if}~j> d_1-M,
			\end{cases}
			&&
			y_j = 
			d_2+1-j.
		\end{align*}
		Let $ \widetilde{C}=\min(d_1,\ell) $.
		Then,
		\begin{enumerate}[label=(\roman*)]
			\item\label{enum:m=rsamelength} if $ C_2+L>\widetilde{C}-M-1 $ then $ T_{\ell}(P)= $
			\[
			u_1^{\otimes \widetilde{C}} \otimes
			\begin{array}{|c|}
				\hline
				\btw\\
				x_1\\
				\hline
			\end{array}
			\otimes \cdots \otimes 
			\begin{array}{|c|}
				\hline
				\btw\\
				x_{d_1}\\
				\hline
			\end{array}
			\otimes u_1^{\otimes C_2 + \min(d_2,\ell) - \widetilde{C}}\otimes 
			\begin{array}{|c|}
				\hline
				\bon\\
				y_1\\
				\hline
			\end{array}
			\otimes \cdots \otimes 
			\begin{array}{|c|}
				\hline
				\bon\\
				y_L\\
				\hline
			\end{array}
			\otimes
			\begin{array}{|c|}
				\hline
				\btw\\
				y_{L+1}\\
				\hline
			\end{array}
			\otimes \cdots \otimes
			\begin{array}{|c|}
				\hline
				\btw\\
				y_{d_2}\\
				\hline
			\end{array}
			\otimes u_1^{\otimes \infty } ;
			\]
			\item\label{enum:m=rdifflength} if  $ C_2+L=\widetilde{C}-M-1 $ then $ T_{\ell}(P)= $
			\begin{align*}
				&
				u_1^{\otimes \widetilde{C}} \otimes
				\begin{array}{|c|}
					\hline
					\btw\\
					x_1\\
					\hline
				\end{array}
				\otimes \cdots \otimes 
				\begin{array}{|c|}
					\hline
					\btw\\
					x_{d_1-M-1}\\
					\hline
				\end{array}
				\otimes
				\begin{array}{|c|}
					\hline
					\btw\\
					x_{d_1-M+1}\\
					\hline
				\end{array}
				\otimes \cdots \otimes 
				\begin{array}{|c|}
					\hline
					\btw\\
					x_{d_1}\\
					\hline
				\end{array}
				\\&\quad
				\otimes u_1^{\otimes C_2+\ell-\widetilde{C}}\otimes 
				\begin{array}{|c|}
					\hline
					\bon\\
					x_{d_1-M}\\
					\hline
				\end{array}
				\otimes
				\begin{array}{|c|}
					\hline
					\bon\\
					y_1\\
					\hline
				\end{array}
				\otimes \cdots \otimes 
				\begin{array}{|c|}
					\hline
					\bon\\
					y_{L-1}\\
					\hline
				\end{array}
				\otimes
				\begin{array}{|c|}
					\hline
					\btw\\
					y_L\\
					\hline
				\end{array}
				\otimes \cdots \otimes
				\begin{array}{|c|}
					\hline
					\btw\\
					y_{d_2}\\
					\hline
				\end{array}
				\otimes u_1^{\otimes \infty }.
			\end{align*}
		\end{enumerate}
		
	\end{lemma}
	\begin{proof}
		We use the same notation as in the proof of Lemma~\ref{lem:m=roverlap}.
		Similar to the proof of Lemma~\ref{lem:m=roverlap},
		we have that
		\begin{align*}
			\widetilde{P}_1\otimes\cdots\otimes\widetilde{P}_{d_1+C_2} 
			&= u_1^{\otimes \widetilde{C}}
			\otimes
			\begin{array}{|c|}
				\hline
				\btw\\
				x_1\\
				\hline
			\end{array}
			\otimes\cdots\otimes
			\begin{array}{|c|}
				\hline
				\btw\\
				x_{d_1+C_2-\widetilde{C}}\\
				\hline
			\end{array}
			\\
			u(d_1+C_2) &= 
			\begin{array}{|c|}
				\hline
				\btw^{\ell}\\
				\bon^{\ell-\widetilde{C}+C_2}(x_{d_1-j})_{j=0}^{\widetilde{C}-C_2-1}\\
				\hline
			\end{array}
			\,.
		\end{align*}

		\begin{itemize}
			\item Assume $ C_2+L>\widetilde{C}-M-1 $ and
			$ M<\widetilde{C}-C_2 $
			(so that $ B_2+B_1<A_3+A_2 $).
			We can apply modified versions of ~\eqref{eq:m=r cd at d1-M} and~\eqref{eq:m=r b at d1-M}
			(because $ d_1+L+C_2\ge d_1+\widetilde{C}-M $).
			After $ u(d_1+\widetilde{C}-M) $,
			all of the applicable cases of Lemma~\ref{lem:m=rclassRmat}
			put down members from $ (x_{d_1-j})_{j=0}^{M-1} $ in order
			(with $ \btw $ in the row above).
			Observe that
			\[
			(d_1+C_2+d_2)-(d_1+\widetilde{C}-M)
			= d_2+C_2-\widetilde{C}+M \geq (\ell-1)+C_2-\ell+M \geq M
			\]
			so \emph{all} of the members from $ (x_{d_1-j})_{j=0}^{M-1} $ will be put down.
			In particular, $ u(d_1+\widetilde{C}) $ will contain no members from $ (x_{d_1-j})_{j=0}^{M-1} $.
			While applying cases~\ref{enum:bona},\ref{enum:bonb}, \ref{enum:bonc}, \ref{enum:bond}, \ref{enum:bone}, \ref{enum:btwa}, \ref{enum:btwb}, to $ u(j) $,
			we find that $ B_2(j)=B_2(d_1+C_2)+A_1(j) $.
			Thus, we have that $ B_2(d_1+\widetilde{C}) = B_2(d_1+C_2)= \ell-\widetilde{C}+C_2 $.
			Thus, after $ u(d_1+\widetilde{C}) $,
			case~\ref{enum:bonf} or~\ref{enum:btwc} apply $ \ell-\widetilde{C}+C_2 $ times (i.e.\ until $ B_2=0 $) unless $ u(d_1+C_2+d_2) $ is reached beforehand.
			Hence, the number of vacuum states after the first soliton is
			\[
			\min((d_1+C_2+d_2)-(d_1+\widetilde{C}),\, \ell-\widetilde{C}+C_2)
			=C_2+\min(d_2,\ell)-\widetilde{C}
			\]
			as required.
			Applying~\ref{enum:bong}/\ref{enum:btwd}
			proves~\ref{enum:m=rsamelength} in this case.
			\item Assume $ C_2+L>\widetilde{C}-M-1 $ and $ M \geq \widetilde{C}-C_2 $
			(so that $ B_1+B_2 = A_3+A_2 $),
			then neither ~\eqref{eq:m=r cd at d1-M} nor~\eqref{eq:m=r b at d1-M} can apply.
			As above,
			all of the applicable cases of Lemma~\ref{lem:m=rclassRmat}
			put down members from $ (x_{d_1-j})_{j=0}^{\widetilde{C}-C_2-1} $ in order and 
			\[
			(d_1+C_2+d_2)-(d_1+C_2)=d_2\ge M \ge \widetilde{C}-C_2
			\]
			so \emph{all} of the members from $ (x_{d_1-j})_{j=0}^{\widetilde{C}-C_2-1} $ will be put down.
			After $ u(d_1+\widetilde{C}) $, case~\ref{enum:bonf} or~\ref{enum:btwc} applies
			until $ B_2=0 $ (or until $ u(d_1+C_2+d_2) $).
			As above,
			we conclude that the number of vacuum states after the first soliton is
			$ C_2+\min(d_2,\ell)-\widetilde{C} $,
			and we apply~\ref{enum:bong}/\ref{enum:btwd} to prove~\ref{enum:m=rsamelength}.
			\item Assume $ C_2+L=\widetilde{C}-M-1 $.
			Observe that $ M=\widetilde{C}-C_2-L-1 < \widetilde{C}-C_2 $.
			We know $ d_1+C_2+L = d_1 + \widetilde{C}-M-1 $,
			so modified versions of~\eqref{eq:m=r cd at d1-M} and~\eqref{eq:m=r b at d1-M}
			do not apply (but the state immediately before them does).
			Thus, $ x_{d_1-M} $ is still in $ u(d_1+\widetilde{C}-M-1) = u(d_1+C_2+L) $.
			After $ u(d_1+C_2+L) $, 
			all applicable cases of Lemma~\ref{lem:m=rclassRmat}
			put down members from $ (x_{d_1-j})_{j=0}^{M-1} $ in order and 
			\[
			(d_1+C_2+d_2)-(d_1+\widetilde{C}-M-1)
			=d_2+C_2-\widetilde{C}+M+1
			\geq \ell + C_2 - \ell + M \geq M
			\]
			so \emph{all} of the members from $ (x_{d_1-j})_{j=0}^{M-1} $ will be put down.
			After $ u(d_1+\widetilde{C}-1) $, case~\ref{enum:bonf} or~\ref{enum:btwc}
			applies until $ B_2=0 $ (or until $ u(d_1+C_2+d_2) $).
			As above,
			we have that $ B_2(d_1+\widetilde{C}-1)=B(d_1+C_2)=\ell-\widetilde{C}+C_2 $.
			Thus, the number of vacuum states after the first soliton is
			\begin{align*}
				\min((d_1+C_2+d_2)-(d_1+\widetilde{C}-1),\ell-\widetilde{C}+C_2)
				&=C_2+\min(d_2+1,\ell)-\widetilde{C}\\
				&=C_2+\ell-\widetilde{C}
			\end{align*}
			as required.
			Applying~\ref{enum:bong}/\ref{enum:btwd} proves~\ref{enum:m=rdifflength}. \qedhere
		\end{itemize}
	\end{proof}
	
	Let
	\[
	\mathfrak{p} = u_1^{\otimes \mathfrak{c}_1}\otimes
	\bigotimes_{j=1}^{\mathfrak{d}_1-M}
	\begin{array}{|c|}
		\hline
		\btw\\
		\mathfrak{y}_j\\
		\hline
	\end{array}
	\otimes
	\bigotimes_{j=1}^{M}
	\begin{array}{|c|}
		\hline
		\btw\\
		\mathfrak{x}_j\\
		\hline
	\end{array}
	\otimes
	u_1^{\otimes\mathfrak{c}_2}
	\otimes
	\bigotimes_{j=\mathfrak{d}_1-M+1}^{\mathfrak{d}_1+\mathfrak{d}_2-M}
	\begin{array}{|c|}
		\hline
		\mathfrak{w}_j\\
		\mathfrak{y}_j\\
		\hline
	\end{array}
	\otimes
	u_1^{\otimes \infty}
	\]
	with $ \mathfrak{c}_2\ge\mathfrak{d}_1>\mathfrak{d}_2 $.
	Choose $ \ell=\mathfrak{d}_2+1 $.
	
	Assume $ M+L\le\mathfrak{d}_2 $.
	For $ 0\le t \le t_1 := \mathfrak{c}_2-\mathfrak{d}_2+L+M $,
	we manually compute and then apply Lemma~\ref{lem:m=rnooverlap}~\ref{enum:m=rsamelength} 
	to find that $ (T_{\ell})^t(\mathfrak{p})= $
	\[
	u_1^{\otimes \mathfrak{c}_1+(\mathfrak{d}_2+1)t}\otimes
	\bigotimes_{j=1}^{\mathfrak{d}_1-M}
	\begin{array}{|c|}
		\hline
		\btw\\
		\mathfrak{y}_j\\
		\hline
	\end{array}
	\otimes
	\bigotimes_{j=1}^{M}
	\begin{array}{|c|}
		\hline
		\btw\\
		\mathfrak{x}_j\\
		\hline
	\end{array}
	\otimes
	u_1^{\otimes\mathfrak{c}_2-t}
	\otimes
	\bigotimes_{j=\mathfrak{d}_1-M+1}^{\mathfrak{d}_1+\mathfrak{d}_2-M}
	\begin{array}{|c|}
		\hline
		\mathfrak{w}_j\\
		\mathfrak{y}_j\\
		\hline
	\end{array}
	\otimes
	u_1^{\otimes \infty}.
	\]
	For $ t_1 < t \le t_0:= \mathfrak{c}_2-2\mathfrak{d}_2+L+M+\mathfrak{d}_1$,
	we apply Lemma~\ref{lem:m=rnooverlap}~\ref{enum:m=rdifflength} to find that 
	$ (T_{\ell})^t(\mathfrak{p})= $
	\begin{align*}
		&
		u_1^{\otimes \mathfrak{c}_1+(\mathfrak{d}_2+1)t}\otimes
		\bigotimes_{j=1}^{\mathfrak{d}_1-M-t+t_1}
		\begin{array}{|c|}
			\hline
			\btw\\
			\mathfrak{y}_j\\
			\hline
		\end{array}
		\otimes
		\bigotimes_{j=1}^{M}
		\begin{array}{|c|}
			\hline
			\btw\\
			\mathfrak{x}_j\\
			\hline
		\end{array}
		\\&\qquad
		\otimes
		u_1^{\mathfrak{d}_2-L-M}
		\otimes
		\bigotimes_{j=\mathfrak{d}_1-M-t+t_1+1}^{\mathfrak{d}_1-M+\mathfrak{d}_2-t+t_1}
		\begin{array}{|c|}
			\hline
			\mathfrak{w}_{j+t-t_1}\\
			\mathfrak{y}_j\\
			\hline
		\end{array}
		\otimes
		\bigotimes_{j=\mathfrak{d}_1+\mathfrak{d}_2-M-t+t_1+1}^{\mathfrak{d}_1+\mathfrak{d}_2-M}
		\begin{array}{|c|}
			\hline
			\btw\\
			\mathfrak{y}_j\\
			\hline
		\end{array}
		\otimes
		u_1^{\otimes \infty}.
	\end{align*}
	For $ t > t_0 $, we apply Lemma~\ref{lem:m=rnooverlap}~\ref{enum:m=rsamelength} to find that
	$ (T_{\ell})^t(\mathfrak{p})= $
	\begin{align*}
		&
		u_1^{\otimes \mathfrak{c}_1+(\mathfrak{d}_2+1)t_0+\mathfrak{d}_2(t-t_0)}\otimes
		\bigotimes_{j=1}^{\mathfrak{d}_2-M}
		\begin{array}{|c|}
			\hline
			\btw\\
			\mathfrak{y}_j\\
			\hline
		\end{array}
		\otimes
		\bigotimes_{j=1}^{M}
		\begin{array}{|c|}
			\hline
			\btw\\
			\mathfrak{x}_j\\
			\hline
		\end{array}
		\\&\qquad
		\otimes
		u_1^{\mathfrak{d}_2-L-M+t-t_0}
		\otimes
		\bigotimes_{j=\mathfrak{d}_2-M+1}^{2\mathfrak{d}_2-M}
		\begin{array}{|c|}
			\hline
			\mathfrak{w}_{j+\mathfrak{d}_1-\mathfrak{d}_2}\\
			\mathfrak{y}_j\\
			\hline
		\end{array}
		\otimes
		\bigotimes_{j=2\mathfrak{d}_2-M+1}^{\mathfrak{d}_1+\mathfrak{d}_2-M}
		\begin{array}{|c|}
			\hline
			\btw\\
			\mathfrak{y}_j\\
			\hline
		\end{array}
		\otimes
		u_1^{\otimes \infty}.
	\end{align*}
	Observe that the phase shift is $ 2\mathfrak{d}_2-L-M $, as required.
	
	Assume $ M+L>\mathfrak{d}_2 $.
	For $ 0\le t \le\mathfrak{c}_2 $,
	we manually compute and apply Lemma~\ref{lem:m=rnooverlap}~\ref{enum:m=rsamelength} 
	to see 
	\[
	(T_{\ell})^t(\mathfrak{p})=  u_1^{\otimes \mathfrak{c}_1+(\mathfrak{d}_2+1)t}\otimes
	\bigotimes_{j=1}^{\mathfrak{d}_1-M}
	\begin{array}{|c|}
		\hline
		\btw\\
		\mathfrak{y}_j\\
		\hline
	\end{array}
	\otimes
	\bigotimes_{j=1}^{M}
	\begin{array}{|c|}
		\hline
		\btw\\
		\mathfrak{x}_j\\
		\hline
	\end{array}
	\otimes
	u_1^{\otimes\mathfrak{c}_2-t}
	\otimes
	\bigotimes_{j=\mathfrak{d}_1-M+1}^{\mathfrak{d}_1+\mathfrak{d}_2-M}
	\begin{array}{|c|}
		\hline
		\mathfrak{w}_j\\
		\mathfrak{y}_j\\
		\hline
	\end{array}
	\otimes
	u_1^{\otimes \infty}.
	\]
	For $ \mathfrak{c}_2<t\le t_1:= \mathfrak{c}_2+L+M-\mathfrak{d}_2 $,
	we apply Lemma~\ref{lem:m=roverlap}~\ref{enum:m=rmoreoverlap} 
	to find that $ (T_{\ell})^t(\mathfrak{p})= $
	\[
	u_1^{\otimes \mathfrak{c}_1+(\mathfrak{d}_2+1)t}\otimes
	\bigotimes_{j=1}^{\mathfrak{d}_1-M}
	\begin{array}{|c|}
		\hline
		\btw\\
		\mathfrak{y}_j\\
		\hline
	\end{array}
	\otimes
	\bigotimes_{j=1}^{M-t+\mathfrak{c}_2}
	\begin{array}{|c|}
		\hline
		\btw\\
		\mathfrak{x}_j\\
		\hline
	\end{array}
	\otimes
	\bigotimes_{j=\mathfrak{d}_1-M+1}^{\mathfrak{d}_1-M+t-\mathfrak{c}_2}
	\begin{array}{|c|}
		\hline
		\mathfrak{x}_{j-\mathfrak{d}_1+2M-t+\mathfrak{c}_2}\\
		\mathfrak{y}_j\\
		\hline
	\end{array}
	\otimes
	\bigotimes_{j=\mathfrak{d}_1-M+t-\mathfrak{c}_2+1}^{\mathfrak{d}_1+\mathfrak{d}_2-M}
	\begin{array}{|c|}
		\hline
		\mathfrak{w}_j\\
		\mathfrak{y}_j\\
		\hline
	\end{array}
	\otimes
	u_1^{\otimes \infty}.
	\]
	For $ t_1<t\le t_0 := \mathfrak{d}_1+\mathfrak{c}_2+L+M-2\mathfrak{d}_2 $,
	we apply Lemma~\ref{lem:m=roverlap}~\ref{enum:m=rsameoverlap} 
	to find that $ (T_{\ell})^t(\mathfrak{p})= $
	\begin{align*}
		&
		u_1^{\otimes \mathfrak{c}_1+(\mathfrak{d}_2+1)t}\otimes
		\bigotimes_{j=1}^{\mathfrak{d}_1-M-t+t_1}
		\begin{array}{|c|}
			\hline
			\btw\\
			\mathfrak{y}_j\\
			\hline
		\end{array}
		\otimes
		\bigotimes_{j=1}^{\mathfrak{d}_2-L}
		\begin{array}{|c|}
			\hline
			\btw\\
			\mathfrak{x}_j\\
			\hline
		\end{array}
		\otimes
		\bigotimes_{j=\mathfrak{d}_1-M+1}^{\mathfrak{d}_1-\mathfrak{d}_2+L}
		\begin{array}{|c|}
			\hline
			\mathfrak{x}_{j-\mathfrak{d}_1+\mathfrak{d}_2-L+M}\\
			\mathfrak{y}_{j-t+t_1}\\
			\hline
		\end{array}
		\\&\qquad
		\otimes
		\bigotimes_{j=\mathfrak{d}_1-\mathfrak{d}_2+L+1}^{\mathfrak{d}_1+\mathfrak{d}_2-M}
		\begin{array}{|c|}
			\hline
			\mathfrak{w}_j\\
			\mathfrak{y}_{j-t+t_1}\\
			\hline
		\end{array}
		\otimes
		\bigotimes_{j=\mathfrak{d}_1+\mathfrak{d}_2-M+1}^{\mathfrak{d}_1+\mathfrak{d}_2-M+t-t_1}
		\begin{array}{|c|}
			\hline
			\btw\\
			\mathfrak{y}_{j-t+t_1}\\
			\hline
		\end{array}
		\otimes
		u_1^{\otimes \infty}.
	\end{align*}
	For $ t_0<t\le t_0+M+L-\mathfrak{d}_2 $,
	we apply Lemma~\ref{lem:m=roverlap}~\ref{enum:m=rlessoverlap} 
	to find that $ (T_{\ell})^t(\mathfrak{p})= $
	\begin{align*}
		&
		u_1^{\otimes \mathfrak{c}_1+(\mathfrak{d}_2+1)t_0+\mathfrak{d}_2(t-t_0)}\otimes
		\bigotimes_{j=1}^{\mathfrak{d}_2-M}
		\begin{array}{|c|}
			\hline
			\btw\\
			\mathfrak{y}_j\\
			\hline
		\end{array}
		\otimes
		\bigotimes_{j=1}^{\mathfrak{d}_2-L+t-t_0}
		\begin{array}{|c|}
			\hline
			\btw\\
			\mathfrak{x}_j\\
			\hline
		\end{array}
		\otimes
		\bigotimes_{j=\mathfrak{d}_2-M+1}^{L-t+t_0}
		\begin{array}{|c|}
			\hline
			\mathfrak{x}_{j-L+M+t-t_0}\\
			\mathfrak{y}_{j}\\
			\hline
		\end{array}
		\\&\qquad
		\otimes
		\bigotimes_{j=L-t+t_0+1}^{2\mathfrak{d}_2-M}
		\begin{array}{|c|}
			\hline
			\mathfrak{w}_{j+\mathfrak{d}_1-\mathfrak{d}_2}\\
			\mathfrak{y}_j\\
			\hline
		\end{array}
		\otimes
		\bigotimes_{j=2\mathfrak{d}_2-M+1}^{\mathfrak{d}_1+\mathfrak{d}_2-M}
		\begin{array}{|c|}
			\hline
			\btw\\
			\mathfrak{y}_j\\
			\hline
		\end{array}
		\otimes
		u_1^{\otimes \infty}.
	\end{align*}
	For $ t > t_0+M+L-\mathfrak{d}_2 $,
	we apply Lemma~\ref{lem:m=rnooverlap}~\ref{enum:m=rsamelength} 
	to find that $ (T_{\ell})^t(\mathfrak{p})= $
	\begin{align*}
		&
		u_1^{\otimes \mathfrak{c}_1+(\mathfrak{d}_2+1)t_0+\mathfrak{d}_2(t-t_0)}\otimes
		\bigotimes_{j=1}^{\mathfrak{d}_2-M}
		\begin{array}{|c|}
			\hline
			\btw\\
			\mathfrak{y}_j\\
			\hline
		\end{array}
		\otimes
		\bigotimes_{j=1}^M
		\begin{array}{|c|}
			\hline
			\btw\\
			\mathfrak{x}_j\\
			\hline
		\end{array}
		\otimes
		u_1^{\otimes t-M-L+\mathfrak{d}_2-t_0}
		\\&\qquad
		\otimes
		\bigotimes_{j=\mathfrak{d}_2-M+1}^{2\mathfrak{d}_2-M}
		\begin{array}{|c|}
			\hline
			\mathfrak{w}_{j+\mathfrak{d}_1-\mathfrak{d}_2}\\
			\mathfrak{y}_j\\
			\hline
		\end{array}
		\otimes
		\bigotimes_{j=2\mathfrak{d}_2-M+1}^{\mathfrak{d}_1+\mathfrak{d}_2-M}
		\begin{array}{|c|}
			\hline
			\btw\\
			\mathfrak{y}_j\\
			\hline
		\end{array}
		\otimes
		u_1^{\otimes \infty}.
	\end{align*}
	Observe that the phase shift is $ 2\mathfrak{d}_2-L-M $, as required.
	
	\section{Proof of Lemma~\ref{lem:classRmat}}\label{appendix:classRmatprf}
	In this subsection,
	the shape of the tableaux are indicative of the true shapes,
	but are often not accurately aligned.
	When performing these calculations,
	we are assuming $ m\ge r+2 $, so  $ \bfth,\bftw,\bfon,\bfze \leq \bon $.
	
	If $ b_2=0 $ and $ b_2+b_1=a_3+a_2 $ then
	\begin{equation}\label{eq:insbonIIandIII}
		\bfX\bfon \mathbf{y}
		\rightarrow
		\begin{array}{|c|}
			\hline
			\bfX^{\ell}\\
			\bfth^{a_3}\bftw^{a_2}\bfon^{a_1}\\
			\bftw^{b_2}\bfon^{b_1}\bfze^{b_0}\\
			\hline
		\end{array}
		=
		\bfon \mathbf{y} \rightarrow
		\begin{array}{|c|c}
			\hline
			\multicolumn{2}{|c|}{\bfX^{\ell+1}}\\
			\cline{2-2}
			\bfth^{a_3}\bftw^{a_2}\bfon^{a_1} & \\
			\bfon^{b_1}\bfze^{b_0} & \\
			\cline{1-1}
		\end{array}
		=
		\mathbf{y} \rightarrow
		\begin{array}{|c|c}
			\hline
			\multicolumn{2}{|c|}{\bfX^{\ell+1}}\\
			\multicolumn{2}{|c|}{\bfth^{a_3}\bftw^{a_2}\bfon^{a_1+1}}\\
			\cline{2-2}
			\bfon^{b_1}\bfze^{b_0} & \\
			\cline{1-1}
		\end{array}
		\,.
	\end{equation}
	If $ b_2<a_3 $ and $ b_2+b_1=a_3+a_2 $ then
	\begin{equation}\label{eq:insIII}
		\bfX\bftw \mathbf{y} 
		\rightarrow
		\begin{array}{|c|}
			\hline
			\bfX^{\ell}\\
			\bfth^{a_3}\bftw^{a_2}\bfon^{a_1}\\
			\bftw^{b_2}\bfon^{b_1}\bfze^{b_0}\\
			\hline
		\end{array}
		=
		\bftw \mathbf{y} \rightarrow
		\begin{array}{|c|c}
			\hline
			\multicolumn{2}{|c|}{\bfX^{\ell+1}}\\
			\cline{2-2}
			\bfth^{a_3}\bftw^{a_2}\bfon^{a_1} & \\
			\bftw^{b_2}\bfon^{b_1}\bfze^{b_0} & \\
			\cline{1-1}
		\end{array}
		=
		\mathbf{y} \rightarrow
		\begin{array}{|c|c}
			\hline
			\multicolumn{2}{|c|}{\bfX^{\ell+1}}\\
			\multicolumn{2}{|c|}{\bfth^{a_3}\bftw^{a_2}\bfon^{a_1+1}}\\
			\cline{2-2}
			\bftw^{b_2+1}\bfon^{b_1-1}\bfze^{b_0} & \\
			\cline{1-1}
		\end{array}
		\,.
	\end{equation}
	If $ b_2=a_3 $ then
	\begin{equation}\label{eq:insII}
		\bfX\bftw \mathbf{y} 
		\rightarrow
		\begin{array}{|c|}
			\hline
			\bfX^{\ell}\\
			\bfth^{a_3}\bftw^{a_2}\bfon^{a_1}\\
			\bftw^{b_2}\bfon^{b_1}\bfze^{b_0}\\
			\hline
		\end{array}
		=
		\bftw \mathbf{y} \rightarrow
		\begin{array}{|c|c}
			\hline
			\multicolumn{2}{|c|}{\bfX^{\ell+1}}\\
			\cline{2-2}
			\bfth^{a_3}\bftw^{a_2}\bfon^{a_1} & \\
			\bftw^{b_2}\bfon^{b_1}\bfze^{b_0} & \\
			\cline{1-1}
		\end{array}
		=
		\mathbf{y} \rightarrow
		\begin{array}{|c|c}
			\hline
			\multicolumn{2}{|c|}{\bfX^{\ell+1}}\\
			\multicolumn{2}{|c|}{\bfth^{a_3}\bftw^{a_2+1}\bfon^{a_1}}\\
			\cline{2-2}
			\bftw^{b_2}\bfon^{b_1}\bfze^{b_0} & \\
			\cline{1-1}
		\end{array}
		\,.
	\end{equation}
	Without any assumptions, we have that
	\begin{equation}\label{eq:insbth}
		\bfX\bfth \mathbf{y} 
		\rightarrow
		\begin{array}{|c|}
			\hline
			\bfX^{\ell}\\
			\bfth^{a_3}\bftw^{a_2}\bfon^{a_1}\\
			\bftw^{b_2}\bfon^{b_1}\bfze^{b_0}\\
			\hline
		\end{array}
		=
		\bfth \mathbf{y} \rightarrow
		\begin{array}{|c|c}
			\hline
			\multicolumn{2}{|c|}{\bfX^{\ell+1}}\\
			\cline{2-2}
			\bfth^{a_3}\bftw^{a_2}\bfon^{a_1} & \\
			\bftw^{b_2}\bfon^{b_1}\bfze^{b_0} & \\
			\cline{1-1}
		\end{array}
		=
		\mathbf{y} \rightarrow
		\begin{array}{|c|c}
			\hline
			\multicolumn{2}{|c|}{\bfX^{\ell+1}}\\
			\multicolumn{2}{|c|}{\bfth^{{a_3}+1}\bftw^{a_2}\bfon^{a_1}}\\
			\cline{2-2}
			\bftw^{b_2}\bfon^{b_1}\bfze^{b_0} & \\
			\cline{1-1}
		\end{array}
		\,.
	\end{equation}
	Observe that~\eqref{eq:insbonIIandIII}, \eqref{eq:insIII}, \eqref{eq:insII} and~\eqref{eq:insbth} 
	are all of the form
	\begin{equation}\label{eq:formIIandIII}
		\mathbf{y} \rightarrow
		\begin{array}{|c|c}
			\hline
			\multicolumn{2}{|c|}{\bfX^{\ell+1}}\\
			\multicolumn{2}{|c|}{\bfth^{\widetilde{a}_3}\bftw^{\widetilde{a}_2}\bfon^{\widetilde{a}_1}}\\
			\cline{2-2}
			\bftw^{\widetilde{b}_2}\bfon^{\widetilde{b}_1}\bfze^{b_0} & \\
			\cline{1-1}
		\end{array}
		\,.
	\end{equation}
	For some integers $ \widetilde{a}_3,\widetilde{a}_2,\widetilde{a}_1,\widetilde{b}_2,\widetilde{b}_1 $ with
	$ \widetilde{a}_3+\widetilde{a}_2+\widetilde{a}_1=\ell+1 $ and $ \widetilde{b}_2+\widetilde{b}_1=b_2+b_1 $.
	Let $ (\widetilde{\mathbf{h}}_j)_{j=1}^{\ell+1}=\bfth^{\widetilde{a}_3}\bftw^{\widetilde{a}_2}\bfon^{\widetilde{a}_1} $
	and $ (\widetilde{\mathbf{z}}_j)_{j=1}^{\ell}=\bftw^{\widetilde{b}_2}\bfon^{\widetilde{b}_1}\bfze^{b_0} $.
	
	If $ \widetilde{b}_2>0 $ or ($ \widetilde{b}_1>0 $ and $ \mathbf{y}=\bfze $) then
	\begin{equation}\label{eq:insi}
		\mathbf{y} \rightarrow
		\begin{array}{|c|c}
			\hline
			\multicolumn{2}{|c|}{\bfX^{\ell+1}}\\
			\multicolumn{2}{|c|}{\bfth^{\widetilde{a}_3}\bftw^{\widetilde{a}_2}\bfon^{\widetilde{a}_1}}\\
			\cline{2-2}
			\bftw^{\widetilde{b}_2}\bfon^{\widetilde{b}_1}\bfze^{b_0} & \\
			\cline{1-1}
		\end{array}
		=
		\begin{array}{|c|c|c}
			\hline
			\multicolumn{3}{|c|}{\bfX^{\ell+1}}\\
			\multicolumn{3}{|c|}{\bfth^{\widetilde{a}_3}\bftw^{\widetilde{a}_2}\bfon^{\widetilde{a}_1}}\\
			\cline{3-3}
			\multicolumn{2}{|c|}{\bftw^{\widetilde{b}_2}\bfon^{\widetilde{b}_1}\bfze^{b_0}} & \\
			\cline{2-2}
			\mathbf{y}\\
			\cline{1-1}
		\end{array}
		\,.
	\end{equation}
	
	\begin{proof}[Proof of~ \ref{enum:s},~\ref{enum:k} and~\ref{enum:m}]
		Suppose $ \widetilde{b}_1>0 $ and $ \mathbf{y}=\bfze $ (but not necessarily $ \widetilde{b}_2>0 $).
		(Note that for~\eqref{eq:insbonIIandIII},
		we must already have $ \mathbf{y}=\bfze $.)
		\begin{align*}
			\col&\left(
			\begin{array}{|c|}
				\hline
				\bfX^{\ell}\\
				(\widetilde{\mathbf{h}}_j)_{j=1}^{\widetilde{a}_3+\widetilde{a}_2-1}\bfon^{\widetilde{a}_1}\\
				\bftw^{\widetilde{b}_2}\bfon^{\widetilde{b}_1-1}\mathbf{y}\bfze^{b_0}\\
				\hline
			\end{array}
			\right)
			\rightarrow
			\begin{array}{|c|}
				\hline
				\bfX\\
				\widetilde{\mathbf{h}}_{\widetilde{a}_3+\widetilde{a}_2}\\
				\bfon\\
				\hline
			\end{array} \allowdisplaybreaks \\
			&=
			\begin{cases}
				\col\left(
				\begin{array}{|c|}
					\hline
					\bfX^{\widetilde{a}_3+\widetilde{a}_2-1}\\
					(\widetilde{\mathbf{h}}_j)_{j=1}^{\widetilde{a}_3+\widetilde{a}_2-1}\\
					\bftw^{\widetilde{b}_2}\bfon^{\widetilde{b}_1-1}\\
					\hline
				\end{array}
				\right)
				\rightarrow
				\begin{array}{|c|c|c}
					\hline
					\multicolumn{3}{|c|}{\bfX^{\widetilde{a}_1+1}}\\
					\multicolumn{3}{|c|}{\widetilde{\mathbf{h}}_{\widetilde{a}_3+\widetilde{a}_2}\bfon^{\widetilde{a}_1}}\\
					\cline{3-3}
					\multicolumn{2}{|c|}{\bfon\bfze^{b_0}} & \\
					\cline{2-2}
					\mathbf{y}\\
					\cline{1-1}
				\end{array}
				& \text{if}~ \widetilde{b}_2+\widetilde{b}_1=\widetilde{a}_3+\widetilde{a}_2\\[30pt]
				\col\left(
				\begin{array}{|c|}
					\hline
					\bfX^{\ell-\widetilde{a}_1}\\
					(\widetilde{\mathbf{h}}_j)_{j=1}^{\widetilde{a}_3+\widetilde{a}_2-1}\\
					\bftw^{\widetilde{b}_2}\bfon^{\widetilde{b}_1-1}\mathbf{y}\bfze^{{b_0}-\widetilde{a}_1}\\
					\hline
				\end{array}
				\right)
				\rightarrow
				\begin{array}{|c|c|c}
					\hline
					\multicolumn{3}{|c|}{\bfX^{\widetilde{a}_1+1}}\\
					\multicolumn{3}{|c|}{\widetilde{\mathbf{h}}_{\widetilde{a}_3+\widetilde{a}_2}\bfon^{\widetilde{a}_1}}\\
					\cline{3-3}
					\multicolumn{2}{|c|}{\bfon\bfze^{\widetilde{a}_1-1}} & \\
					\cline{2-2}
					\bfze\\
					\cline{1-1}
				\end{array}
				&  \text{if}~\widetilde{b}_2+\widetilde{b}_1 \ne \widetilde{a}_3+\widetilde{a}_2~\text{and}~\widetilde{a}_1>0\\[30pt]
				\col\left(
				\begin{array}{|c|}
					\hline
					\bfX^{\ell}\\
					(\widetilde{\mathbf{h}}_j)_{j=1}^{\widetilde{a}_3+\widetilde{a}_2-1}\\
					\bftw^{\widetilde{b}_2}\bfon^{\widetilde{b}_1-1}\mathbf{y}\bfze^{b_0}\\
					\hline
				\end{array}
				\right)
				\rightarrow
				\begin{array}{|c|}
					\hline
					\bfX\\
					\widetilde{\mathbf{h}}_{\widetilde{a}_3+\widetilde{a}_2}\\
					\bfon\\
					\hline
				\end{array}
				&  \text{if}~\widetilde{b}_2+\widetilde{b}_1 \ne \widetilde{a}_3+\widetilde{a}_2~\text{and}~\widetilde{a}_1=0\\
			\end{cases} \allowdisplaybreaks \\
			&=
			\col\left(
			\begin{array}{|c|}
				\hline
				\bfX^{\widetilde{b}_2+\widetilde{b}_1-1}\\
				(\widetilde{\mathbf{h}}_j)_{j=1}^{\widetilde{b}_2+\widetilde{b}_1-1}\\
				\bftw^{\widetilde{b}_2}\bfon^{\widetilde{b}_1-1}\\
				\hline
			\end{array}
			\right)
			\rightarrow
			\begin{array}{|c|c|c}
				\hline
				\multicolumn{3}{|c|}{\bfX^{\ell-\widetilde{b}_2-\widetilde{b}_1+2}}\\
				\multicolumn{3}{|c|}{(\widetilde{\mathbf{h}}_j)_{j=\widetilde{b}_2+\widetilde{b}_1}^{\widetilde{a}_3+\widetilde{a}_2}\bfon^{\widetilde{a}_1}}\\
				\cline{3-3}
				\multicolumn{2}{|c|}{\bfon\bfze^{b_0}} & \\
				\cline{2-2}
				\mathbf{y}\\
				\cline{1-1}
			\end{array}\\
			&=
			\begin{array}{|c|c|c}
				\hline
				\multicolumn{3}{|c|}{\bfX^{\ell+1}}\\
				\multicolumn{3}{|c|}{(\widetilde{\mathbf{h}}_j)_{j=1}^{\widetilde{a}_3+\widetilde{a}_2}\bfon^{\widetilde{a}_1}}\\
				\cline{3-3}
				\multicolumn{2}{|c|}{\bftw^{\widetilde{b}_2}\bfon^{\widetilde{b}_1}\bfze^{b_0}} & \\
				\cline{2-2}
				\mathbf{y}\\
				\cline{1-1}
			\end{array}
			=
			\begin{array}{|c|c|c}
				\hline
				\multicolumn{3}{|c|}{\bfX^{\ell+1}}\\
				\multicolumn{3}{|c|}{\bfth^{\widetilde{a}_3}\bftw^{\widetilde{a}_2}\bfon^{\widetilde{a}_1}}\\
				\cline{3-3}
				\multicolumn{2}{|c|}{\bftw^{\widetilde{b}_2}\bfon^{\widetilde{b}_1}\bfze^{b_0}} & \\
				\cline{2-2}
				\mathbf{y}\\
				\cline{1-1}
			\end{array}
			\,.
		\end{align*}
		This insertion is the same as in~\eqref{eq:insi}.
		Therefore, for these cases,
		\[
		R\left(
		\begin{array}{|c|}
			\hline
			\bfX^{\ell}\\
			(\mathbf{h}_j)_{j=1}^{\ell}\\
			(\mathbf{z}_j)_{j=1}^{\ell}\\
			\hline
		\end{array}
		\otimes 
		\begin{array}{|c|}
			\hline
			\bfX\\
			\mathbf{g} \\
			\mathbf{y}\\
			\hline
		\end{array}
		\right)
		=
		\begin{array}{|c|}
			\hline
			\bfX\\
			\widetilde{\mathbf{h}}_{\widetilde{a}_3+\widetilde{a}_2}\\
			\bfon\\
			\hline
		\end{array}
		\otimes
		\begin{array}{|c|}
			\hline
			\bfX^{\ell}\\
			(\widetilde{\mathbf{h}}_j)_{j=1}^{\widetilde{a}_3+\widetilde{a}_2-1}\bfon^{\widetilde{a}_1}\\
			\bftw^{\widetilde{b}_2}\bfon^{\widetilde{b}_1-1}\mathbf{y}\bfze^{b_0}\\
			\hline
		\end{array}
		\,.
		\]
		
		The substitutions for~\eqref{eq:insbonIIandIII},
		$ \widetilde{a}_3=a_3,\, \widetilde{a}_2=a_2,\,\widetilde{a}_1=a_1+1,\,\widetilde{b}_2=b_2=0,\,\widetilde{b}_1=b_1 $,
		proves~\ref{enum:s}.
		
		The substitutions for~\eqref{eq:insIII},
		$ \widetilde{a}_3=a_3,\,\widetilde{a}_2=a_2,\,\widetilde{a}_1=a_1+1,\,\widetilde{b}_2=b_2+1,\,\widetilde{b}_1=b_1-1 $,
		proves~\ref{enum:k}.
		
		The substitutions for~\eqref{eq:insbth},
		$ \widetilde{a}_3=a_3+1,\, \widetilde{a}_2=a_2,\,\widetilde{a}_1=a_1,\,\widetilde{b}_2=b_2,\,\widetilde{b}_1=b_1 $,
		proves~\ref{enum:m}.
	\end{proof}
	
	\begin{proof}[Proof of~\ref{enum:l},~\ref{enum:e} and~\ref{enum:n}]
		Suppose $ \widetilde{a}_1>0 $ and $ \mathbf{y} \geq \bfon $ and $ \widetilde{b}_2>0 $ 
		but ($ \widetilde{b}_1=0 $ or $ \mathbf{y}=\bfon $).
		(Note that for~\eqref{eq:insIII},
		$ \widetilde{a}_1=a_1+1 $ and $ \widetilde{b}_2=b_2+1 $ 
		so we automatically have that $ \widetilde{a}_1>0 $ and $ \widetilde{b}_2>0 $.
		Additionally, for~\eqref{eq:insII}, we must already have that $ \mathbf{y}\ge\bfon $.)
		\begin{align*}
			\col &\left(
			\begin{array}{|c|}
				\hline
				\bfX^{\ell}\\
				(\widetilde{\mathbf{h}}_j)_{j=2}^{\widetilde{a}_3+\widetilde{a}_2}\bftw\bfon^{\widetilde{a}_1-1}\\
				\bftw^{\widetilde{b}_2-1}\bfon^{\widetilde{b}_1}\mathbf{y}\bfze^{b_0}\\
				\hline
			\end{array}
			\right)
			\rightarrow
			\begin{array}{|c|}
				\hline
				\bfX\\
				\bfth\\
				\bfon\\
				\hline
			\end{array}
			=
			\begin{cases}
				\col\left(
				\begin{array}{|c|}
					\hline
					\bfX^{\ell}\\
					(\widetilde{\mathbf{h}}_j)_{j=2}^{\widetilde{a}_3+\widetilde{a}_2}\bftw\\
					\bftw^{\widetilde{b}_2-1}\bfon^{\widetilde{b}_1}\mathbf{y}\bfze^{b_0}\\
					\hline
				\end{array}
				\right)
				\rightarrow
				\begin{array}{|c|}
					\hline
					\bfX\\
					\bfth\\
					\bfon\\
					\hline
				\end{array}
				& \text{ if $ \widetilde{a}_1=1 $}\\
				\col\left(
				\begin{array}{|c|}
					\hline
					\bfX^{\ell-\widetilde{a}_1+1}\\
					(\widetilde{\mathbf{h}}_j)_{j=2}^{\widetilde{a}_3+\widetilde{a}_2}\bftw\\
					\bftw^{\widetilde{b}_2-1}\bfon^{\widetilde{b}_1}\mathbf{y}\bfze^{{b_0}-\widetilde{a}_1+1}\\
					\hline
				\end{array}
				\right)
				\rightarrow
				\begin{array}{|c|c|c}
					\hline
					\multicolumn{3}{|c|}{\bfX^{\widetilde{a}_1}}\\
					\multicolumn{3}{|c|}{\bfth\bfon^{\widetilde{a}_1-1}}\\
					\cline{3-3}
					\multicolumn{2}{|c|}{\bfon\bfze^{\widetilde{a}_1-2}} & \\
					\cline{2-2}
					\bfze\\
					\cline{1-1}
				\end{array}
				& \text{if $ \widetilde{a}_1>1 $}
			\end{cases} \allowdisplaybreaks \\
			&=
			\begin{cases}
				\col\left(
				\begin{array}{|c|}
					\hline
					\bfX^{\ell-\widetilde{a}_1}\\
					(\widetilde{\mathbf{h}}_j)_{j=2}^{\widetilde{a}_3+\widetilde{a}_2}\\
					\bftw^{\widetilde{b}_2-1}\bfon^{\widetilde{b}_1}\\
					\hline
				\end{array}
				\right)
				\rightarrow
				\begin{array}{|c|c|c}
					\hline
					\multicolumn{3}{|c|}{\bfX^{\widetilde{a}_1+1}}\\
					\multicolumn{3}{|c|}{\bfth\bfon^{\widetilde{a}_1}}\\
					\cline{3-3}
					\multicolumn{2}{|c|}{\bftw\bfze^{b_0}} & \\
					\cline{2-2}
					\mathbf{y}\\
					\cline{1-1}
				\end{array}
				& \text{if}~{b_0}=\widetilde{a}_1-1\\
				\col\left(
				\begin{array}{|c|}
					\hline
					\bfX^{\ell-\widetilde{a}_1}\\
					(\widetilde{\mathbf{h}}_j)_{j=2}^{\widetilde{a}_3+\widetilde{a}_2}\\
					\bftw^{\widetilde{b}_2-1}\bfon^{\widetilde{b}_1}\mathbf{y}\bfze^{{b_0}-\widetilde{a}_1}\\
					\hline
				\end{array}
				\right)
				\rightarrow
				\begin{array}{|c|c|c}
					\hline
					\multicolumn{3}{|c|}{\bfX^{\widetilde{a}_1+1}}\\
					\multicolumn{3}{|c|}{\bfth\bfon^{\widetilde{a}_1}}\\
					\cline{3-3}
					\multicolumn{2}{|c|}{\bftw\bfze^{\widetilde{a}_1-1}} & \\
					\cline{2-2}
					\bfze\\
					\cline{1-1}
				\end{array}
				& \text{if}~{b_0}>\widetilde{a}_1-1
			\end{cases} \allowdisplaybreaks \\
			&=
			\col\left(
			\begin{array}{|c|}
				\hline
				\bfX^{\widetilde{b}_2+\widetilde{b}_1-1}\\
				(\widetilde{\mathbf{h}}_j)_{j=2}^{\widetilde{b}_2+\widetilde{b}_1}\\
				\bftw^{\widetilde{b}_2-1}\bfon^{\widetilde{b}_1}\\
				\hline
			\end{array}
			\right)
			\rightarrow
			\begin{array}{|c|c|c}
				\hline
				\multicolumn{3}{|c|}{\bfX^{\ell-\widetilde{b}_2-\widetilde{b}_1+2}}\\
				\multicolumn{3}{|c|}{\bfth(\widetilde{\mathbf{h}}_j)_{j=\widetilde{b}_2+\widetilde{b}_1+1}^{\widetilde{a}_3+\widetilde{a}_2}\bfon^{\widetilde{a}_1}}\\
				\cline{3-3}
				\multicolumn{2}{|c|}{\bftw\bfze^{b_0}} & \\
				\cline{2-2}
				\mathbf{y}\\
				\cline{1-1}
			\end{array} \allowdisplaybreaks \\
			&=
			\col\left(
			\begin{array}{|c|}
				\hline
				\bfX^{\widetilde{b}_2-1}\\
				(\widetilde{\mathbf{h}}_j)_{j=2}^{\widetilde{b}_2}\\
				\bftw^{\widetilde{b}_2-1}\\
				\hline
			\end{array}
			\right)
			\rightarrow
			\begin{array}{|c|c|c}
				\hline
				\multicolumn{3}{|c|}{\bfX^{\ell-\widetilde{b}_2+2}}\\
				\multicolumn{3}{|c|}{\bfth(\widetilde{\mathbf{h}}_j)_{j=\widetilde{b}_2+1}^{\widetilde{a}_3+\widetilde{a}_2}\bfon^{\widetilde{a}_1}}\\
				\cline{3-3}
				\multicolumn{2}{|c|}{\bftw\bfon^{\widetilde{b}_1}\bfze^{b_0}} & \\
				\cline{2-2}
				\mathbf{y}\\
				\cline{1-1}
			\end{array}
			\qquad \text{because $ \widetilde{b}_1=0 $ or $ \mathbf{y}=\bfon $}\allowdisplaybreaks \\
			&=
			\begin{array}{|c|c|c}
				\hline
				\multicolumn{3}{|c|}{\bfX^{\ell+1}}\\
				\multicolumn{3}{|c|}{\bfth(\widetilde{\mathbf{h}}_j)_{j=2}^{\widetilde{a}_3+\widetilde{a}_2}\bfon^{\widetilde{a}_1}}\\
				\cline{3-3}
				\multicolumn{2}{|c|}{\bftw^{\widetilde{b}_2}\bfon^{\widetilde{b}_1}\bfze^{b_0}} & \\
				\cline{2-2}
				\mathbf{y}\\
				\cline{1-1}
			\end{array}
			=
			\begin{array}{|c|c|c}
				\hline
				\multicolumn{3}{|c|}{\bfX^{\ell+1}}\\
				\multicolumn{3}{|c|}{\bfth^{\widetilde{a}_3}\bftw^{\widetilde{a}_2}\bfon^{\widetilde{a}_1}}\\
				\cline{3-3}
				\multicolumn{2}{|c|}{\bftw^{\widetilde{b}_2}\bfon^{\widetilde{b}_1}\bfze^{b_0}} & \\
				\cline{2-2}
				\mathbf{y}\\
				\cline{1-1}
			\end{array}
			\,.
		\end{align*}
		(Note that $ \widetilde{a}_3\ge \widetilde{b}_2>0 $.)
		This insertion is the same as in~\eqref{eq:insi}.
		Therefore, for these cases,
		\[
		R\left(
		\begin{array}{|c|}
			\hline
			\bfX^{\ell}\\
			(\mathbf{h}_j)_{j=1}^{\ell}\\
			(\mathbf{z}_j)_{j=1}^{\ell}\\
			\hline
		\end{array}
		\otimes 
		\begin{array}{|c|}
			\hline
			\bfX\\
			\mathbf{g} \\
			\mathbf{y}\\
			\hline
		\end{array}
		\right)
		=
		\begin{array}{|c|}
			\hline
			\bfX\\
			\bfth\\
			\bfon\\
			\hline
		\end{array}
		\otimes
		\begin{array}{|c|}
			\hline
			\bfX^{\ell}\\
			(\widetilde{\mathbf{h}}_j)_{j=2}^{\widetilde{a}_3+\widetilde{a}_2}\bftw\bfon^{\widetilde{a}_1-1}\\
			\bftw^{\widetilde{b}_2-1}\bfon^{\widetilde{b}_1}\mathbf{y}\bfze^{b_0}\\
			\hline
		\end{array}
		\\
		\,.
		\]
		
		Making the substitutions for~\eqref{eq:insIII},
		$ \widetilde{a}_3=a_3,\,\widetilde{a}_2=a_2,\,\widetilde{a}_1=a_1+1,\,\widetilde{b}_2=b_2+1,\,\widetilde{b}_1=b_1-1 $,
		proves~\ref{enum:l}.
		
		Making the substitutions for~\eqref{eq:insII},
		$ \widetilde{a}_3=a_3,\,\widetilde{a}_2=a_2+1,\,\widetilde{a}_1=a_1,\,\widetilde{b}_2=b_2,\,\widetilde{b}_1=b_1 $,
		proves~\ref{enum:e}.
		
		Making the substitutions for~\eqref{eq:insbth},
		$ \widetilde{a}_3=a_3+1,\,\widetilde{a}_2=a_2,\,\widetilde{a}_1=a_1,\,\widetilde{b}_2=b_2,\,\widetilde{b}_1=b_1 $,
		proves~\ref{enum:n}.
	\end{proof}
	
	\begin{proof}[Proof of~\ref{enum:f} and~\ref{enum:o}]
		Suppose and $ \widetilde{a}_1=0 $ and $ \mathbf{y}\ge\bfon $
		and $ \widetilde{b}_2>0 $ 
		but ($ \widetilde{b}_1=0 $ or $ \mathbf{y}=\bfon $).
		(Note that for~\eqref{eq:insII}, we must already have that $ \mathbf{y}\ge\bfon $.)
		\begin{align*}
			&
			\col\left(
			\begin{array}{|c|}
				\hline
				\bfX^{\ell}\\
				(\widetilde{\mathbf{h}}_j)_{j=2}^{\widetilde{a}_3+\widetilde{a}_2}\\
				\bftw^{\widetilde{b}_2-1}\bfon^{\widetilde{b}_1} \mathbf{y}\bfze^{b_0}\\
				\hline
			\end{array}
			\right)
			\rightarrow
			\begin{array}{|c|}
				\hline
				\bfX\\
				\bfth\\
				\bftw\\
				\hline
			\end{array}\\
			&=
			\col\left(
			\begin{array}{|c|}
				\hline
				\bfX^{\widetilde{b}_2+\widetilde{b}_1-1}\\
				(\widetilde{\mathbf{h}}_j)_{j=2}^{\widetilde{b}_2+\widetilde{b}_1}\\
				\bftw^{\widetilde{b}_2-1}\bfon^{\widetilde{b}_1}\\
				\hline
			\end{array}
			\right)
			\rightarrow
			\begin{array}{|c|c|c}
				\hline
				\multicolumn{3}{|c|}{\bfX^{\ell-\widetilde{b}_2-\widetilde{b}_1+2}}\\
				\multicolumn{3}{|c|}{\bfth(\widetilde{\mathbf{h}}_j)_{j=\widetilde{b}_2+\widetilde{b}_1+1}^{\widetilde{a}_3+\widetilde{a}_2}}\\
				\cline{3-3}
				\multicolumn{2}{|c|}{\bftw\bfze^{b_0}} & \\
				\cline{2-2}
				\mathbf{y}\\
				\cline{1-1}
			\end{array}\\
			&=
			\col\left(
			\begin{array}{|c|}
				\hline
				\bfX^{\widetilde{b}_2-1}\\
				(\widetilde{\mathbf{h}}_j)_{j=2}^{\widetilde{b}_2}\\
				\bftw^{\widetilde{b}_2-1}\\
				\hline
			\end{array}
			\right)
			\rightarrow
			\begin{array}{|c|c|c}
				\hline
				\multicolumn{3}{|c|}{\bfX^{\ell-\widetilde{b}_2+2}}\\
				\multicolumn{3}{|c|}{\bfth(\widetilde{\mathbf{h}}_j)_{j=\widetilde{b}_2+1}^{\widetilde{a}_3+\widetilde{a}_2}}\\
				\cline{3-3}
				\multicolumn{2}{|c|}{\bftw\bfon^{\widetilde{b}_1}\bfze^{b_0}} & \\
				\cline{2-2}
				\mathbf{y}\\
				\cline{1-1}
			\end{array}
			\qquad \text{because $ \widetilde{b}_1=0 $ or $ \mathbf{y}=\bfon $}\\
			&=
			\begin{array}{|c|c|c}
				\hline
				\multicolumn{3}{|c|}{\bfX^{\ell+1}}\\
				\multicolumn{3}{|c|}{\bfth(\widetilde{\mathbf{h}}_j)_{j=2}^{\widetilde{a}_3+\widetilde{a}_2}}\\
				\cline{3-3}
				\multicolumn{2}{|c|}{\bftw^{\widetilde{b}_2}\bfon^{\widetilde{b}_1}\bfze^{b_0}} & \\
				\cline{2-2}
				\mathbf{y}\\
				\cline{1-1}
			\end{array}
			=
			\begin{array}{|c|c|c}
				\hline
				\multicolumn{3}{|c|}{\bfX^{\ell+1}}\\
				\multicolumn{3}{|c|}{\bfth^{\widetilde{a}_3}\bftw^{\widetilde{a}_2}\bfon^{\widetilde{a}_1}}\\
				\cline{3-3}
				\multicolumn{2}{|c|}{\bftw^{\widetilde{b}_2}\bfon^{\widetilde{b}_1}\bfze^{b_0}} & \\
				\cline{2-2}
				\mathbf{y}\\
				\cline{1-1}
			\end{array}
			\,.
		\end{align*}
		(Note that $ \widetilde{a}_3\ge \widetilde{b}_2>0 $.)
		This insertion is the same as in~\eqref{eq:insi}.
		Therefore, for these cases,
		\[
		R\left(
		\begin{array}{|c|}
			\hline
			\bfX^{\ell}\\
			(\mathbf{h}_j)_{j=1}^{\ell}\\
			(\mathbf{z}_j)_{j=1}^{\ell}\\
			\hline
		\end{array}
		\otimes 
		\begin{array}{|c|}
			\hline
			\bfX\\
			\mathbf{g} \\
			\mathbf{y}\\
			\hline
		\end{array}
		\right)
		=
		\begin{array}{|c|}
			\hline
			\bfX\\
			\bfth\\
			\bftw\\
			\hline
		\end{array}
		\otimes
		\begin{array}{|c|}
			\hline
			\bfX^{\ell}\\
			(\widetilde{\mathbf{h}}_j)_{j=2}^{\widetilde{a}_3+\widetilde{a}_2}\\
			\bftw^{\widetilde{b}_2-1}\bfon^{\widetilde{b}_1}\mathbf{y}\bfze^{b_0}\\
			\hline
		\end{array}
		\,.
		\]
		
		Making the substitutions for~\eqref{eq:insII},
		$ \widetilde{a}_3=a_3,\,\widetilde{a}_2=a_2+1,\,\widetilde{a}_1=a_1,\,\widetilde{b}_2=b_2,\,\widetilde{b}_1=b_1 $,
		proves~\ref{enum:f}.
		
		Making the substitutions for~\eqref{eq:insbth},
		$ \widetilde{a}_3={a_3}+1,\,\widetilde{a}_2=a_2,\,\widetilde{a}_1=a_1,\,\widetilde{b}_2=b_2,\,\widetilde{b}_1=b_1 $,
		proves~\ref{enum:o}.
	\end{proof}
	
	From~\eqref{eq:formIIandIII}, if $ \widetilde{b}_2+\widetilde{b}_1<\widetilde{a}_3+\widetilde{a}_2 $ and 
	($ \widetilde{b}_2<\widetilde{a}_3 $ and $ \mathbf{y}=\bftw $) or ($ \widetilde{b}_2=0 $ and ($ \widetilde{b}_1=0 $ or $ \mathbf{y} = \bfon $)) then
	\begin{equation}\label{eq:insii}
		\mathbf{y} \rightarrow
		\begin{array}{|c|c}
			\hline
			\multicolumn{2}{|c|}{\bfX^{\ell+1}}\\
			\multicolumn{2}{|c|}{\bfth^{\widetilde{a}_3}\bftw^{\widetilde{a}_2}\bfon^{\widetilde{a}_1}}\\
			\cline{2-2}
			\bftw^{\widetilde{b}_2}\bfon^{\widetilde{b}_1}\bfze^{b_0} & \\
			\cline{1-1}
		\end{array}
		=
		\begin{array}{|c|}
			\hline
			\bfX^{\ell+1}\\
			\bfth^{\widetilde{a}_3}\bftw^{\widetilde{a}_2}\bfon^{\widetilde{a}_1}\\
			\mathbf{y}\bftw^{\widetilde{b}_2}\bfon^{\widetilde{b}_1}\bfze^{b_0} \\
			\hline
		\end{array}
		\,.
	\end{equation}
	
	\begin{proof}[Proof of 
		~\ref{enum:p} 
		]
		Suppose $ \widetilde{b}_2+\widetilde{b}_1<\widetilde{a}_3+\widetilde{a}_2 $ and ($ \widetilde{b}_2<\widetilde{a}_3 $ and $ \mathbf{y}=\bftw $) 
		or ($ \widetilde{b}_2=0 $ and ($ \widetilde{b}_1=0 $ or $ \mathbf{y} = \bfon $)).
		(Note that for~\eqref{eq:insbth}, we know $ b_2\le a_3 $,
		so we automatically have that $ \widetilde{b}_2={b_2}<{a_3}+1=\widetilde{a}_3 $.
		Similarly, we automatically  have $ \widetilde{b}_2+\widetilde{b}_1<\widetilde{a}_3+\widetilde{a}_2 $.)
		\begin{align*}
			&
			\col\left(
			\begin{array}{|c|}
				\hline
				\bfX^{\ell}\\
				(\widetilde{\mathbf{h}}_j)_{j=1}^{\ell}\\
				\mathbf{y}(\widetilde{\mathbf{z}}_j)_{j=1}^{\ell-1}\\
				\hline
			\end{array}
			\right)
			\rightarrow
			\begin{array}{|c|}
				\hline
				\bfX\\
				\widetilde{\mathbf{h}}_{\ell+1}\\
				\widetilde{\mathbf{z}}_{\ell}\\
				\hline
			\end{array} 
			= 
			\begin{array}{|c|}
				\hline
				\bfX^{\ell+1}\\
				(\widetilde{\mathbf{h}}_j)_{j=1}^{\ell+1}\\
				\mathbf{y}(\widetilde{\mathbf{z}}_j)_{j=1}^{\ell}\\
				\hline
			\end{array}
			=
			\begin{array}{|c|}
				\hline
				\bfX^{\ell+1}\\
				\bfth^{\widetilde{a}_3}\bftw^{\widetilde{a}_2}\bfon^{\widetilde{a}_1}\\
				\mathbf{y}\bfon^{\widetilde{b}_1}\bfze^{b_0}\\
				\hline
			\end{array}
			\,.
		\end{align*}
		This insertion is the same as in~\eqref{eq:insii}.
		Therefore, for this case,
		\[
		R\left(
		\begin{array}{|c|}
			\hline
			\bfX^{\ell}\\
			(\mathbf{h}_j)_{j=1}^{\ell}\\
			(\mathbf{z}_j)_{j=1}^{\ell}\\
			\hline
		\end{array}
		\otimes 
		\begin{array}{|c|}
			\hline
			\bfX\\
			\bfth \\
			\mathbf{y}\\
			\hline
		\end{array}
		\right)
		=
		\begin{array}{|c|}
			\hline
			\bfX\\
			\widetilde{\mathbf{h}}_{\ell+1}\\
			\widetilde{\mathbf{z}}_{\ell}\\
			\hline
		\end{array}
		\otimes
		\begin{array}{|c|}
			\hline
			\bfX^{\ell}\\
			(\widetilde{\mathbf{h}}_j)_{j=1}^{\ell}\\
			\mathbf{y}(\widetilde{\mathbf{z}}_j)_{j=1}^{\ell-1}\\
			\hline
		\end{array}
		\,.
		\]
		
		Making the substitutions for~\eqref{eq:insbth},
		$ \widetilde{a}_3=a_3+1,\,\widetilde{a}_2=a_2,\,\widetilde{a}_1=a_1,\,\widetilde{b}_2=b_2,\,\widetilde{b}_1=b_1 $,
		proves~\ref{enum:p}.
	\end{proof}
	
	\begin{proof}[Proof of~\ref{enum:q}]
		Suppose that $ b_2>0 $. Then,
		\begin{equation}\label{eq:insbon}
			\bfX\bfon \mathbf{y} 
			\rightarrow
			\begin{array}{|c|}
				\hline
				\bfX^{\ell}\\
				\bfth^{a_3}\bftw^{a_2}\bfon^{a_1}\\
				\bftw^{b_2}(\mathbf{z}_j)_{j=b_2+1}^{\ell}\\
				\hline
			\end{array}
			=
			\bfon \mathbf{y} \rightarrow
			\begin{array}{|c|c}
				\hline
				\multicolumn{2}{|c|}{\bfX^{\ell+1}}\\
				\cline{2-2}
				\bfth^{a_3}\bftw^{a_2}\bfon^{a_1} & \\
				\bftw^{b_2}(\mathbf{z}_j)_{j=b_2+1}^{\ell} & \\
				\cline{1-1}
			\end{array}
			=
			\mathbf{y} \rightarrow
			\begin{array}{|c|c|c}
				\hline
				\multicolumn{3}{|c|}{\bfX^{\ell+1}}\\
				\cline{3-3}
				\multicolumn{2}{|c|}{\bfth^{a_3}\bftw^{a_2}\bfon^{a_1}} & \\
				\multicolumn{2}{|c|}{\bftw^{b_2}(\mathbf{z}_j)_{j={b_2}+1}^{\ell}} & \\
				\cline{2-2}
				\bfon\\
				\cline{1-1}
			\end{array}
			=
			\begin{array}{|c|c|c}
				\hline
				\multicolumn{3}{|c|}{\bfX^{\ell+1}}\\
				\cline{3-3}
				\multicolumn{2}{|c|}{\bfth^{a_3}\bftw^{a_2}\bfon^{a_1}} & \\
				\multicolumn{2}{|c|}{\bftw^{b_2}(\mathbf{z}_j)_{j=b_2+1}^{\ell}} & \\
				\cline{2-2}
				\bfon\\
				\mathbf{y}\\
				\cline{1-1}
			\end{array}
			.
		\end{equation}
		Additionally,
		\begin{align*}
			\col\left(
			\begin{array}{|c|}
				\hline
				\bfX^{\ell}\\
				\bfth^{a_3-1}\bftw^{a_2}\bfon^{a_1+1}\\
				\bftw^{b_2-1}\bfon^{b_1} \mathbf{y}\bfze^{b_0}\\
				\hline
			\end{array}
			\right) \rightarrow
			\begin{array}{|c|}
				\hline
				\bfX\\
				\bfth\\
				\bftw\\
				\hline
			\end{array}
			&=
			\begin{cases}
				\col\left(
				\begin{array}{|c|}
					\hline
					\bfX^{\ell-1}\\
					(\mathbf{h}_j)_{j=2}^{\ell}\\
					(\mathbf{z}_j)_{j=2}^{\ell}\\
					\hline
				\end{array}
				\right) \rightarrow
				\begin{array}{|c|c}
					\hline
					\multicolumn{2}{|c|}{\bfX^2}\\
					\cline{2-2}
					\bfth & \\
					\bftw & \\
					\bfon & \\
					\mathbf{y} & \\
					\cline{1-1}
				\end{array} & \text{if $ b_2+b_1=\ell $}\\
				\col\left(
				\begin{array}{|c|}
					\hline
					\bfX^{\ell-1}\\
					(\mathbf{h}_j)_{j=2}^{\ell}\\
					(\mathbf{z}_j)_{j=2}^{{b_2}+{b_1}} \mathbf{y}\bfze^{{b_0}-1}\\
					\hline
				\end{array}
				\right) \rightarrow
				\begin{array}{|c|c}
					\hline
					\multicolumn{2}{|c|}{\bfX^2}\\
					\cline{2-2}
					\bfth & \\
					\bftw & \\
					\bfon & \\
					\bfze & \\
					\cline{1-1}
				\end{array} & \text{if $ b_2+b_1\ne \ell $}
			\end{cases} \allowdisplaybreaks \\
			&=
			\col\left(
			\begin{array}{|c|}
				\hline
				\bfX^{b_2+b_1}\\
				(\mathbf{h}_j)_{j=2}^{b_2+b_1}\\
				(\mathbf{z}_j)_{j=2}^{b_2+b_1}\\
				\hline
			\end{array}
			\right) \rightarrow
			\begin{array}{|c|c|c}
				\hline
				\multicolumn{3}{|c|}{\bfX^{\ell-b_2-b_1+1}}\\
				\cline{3-3}
				\multicolumn{2}{|c|}{\bfth(\mathbf{h}_j)_{j=b_2+b_1+1}^{\ell}} & \\
				\multicolumn{2}{|c|}{\bftw\bfze^{b_0}} & \\
				\cline{2-2}
				\bfon \\
				\mathbf{y}\\
				\cline{1-1}
			\end{array} \allowdisplaybreaks \\
			&=
			\begin{array}{|c|c|c}
				\hline
				\multicolumn{3}{|c|}{\bfX^{\ell+1}}\\
				\cline{3-3}
				\multicolumn{2}{|c|}{\bfth(\mathbf{h}_j)_{j=2}^{\ell}} & \\
				\multicolumn{2}{|c|}{\bftw(\mathbf{z}_j)_{j=2}^{\ell}} & \\
				\cline{2-2}
				\bfon \\
				\mathbf{y}\\
				\cline{1-1}
			\end{array} 
			=
			\begin{array}{|c|c|c}
				\hline
				\multicolumn{3}{|c|}{\bfX^{\ell+1}}\\
				\cline{3-3}
				\multicolumn{2}{|c|}{\bfth^{a_3}\bftw^{a_2}\bfon^{a_1}} & \\
				\multicolumn{2}{|c|}{\bftw^{b_2}(\mathbf{z}_j)_{j=b_2+1}^{\ell}} & \\
				\cline{2-2}
				\bfon\\
				\mathbf{y}\\
				\cline{1-1}
			\end{array}\,.
			\\
		\end{align*}
		(Note that $ a_3\ge b_2>0 $.) This insertion is the same as in~\eqref{eq:insbon},
		proving~\ref{enum:q}.
	\end{proof}
	
	\section{Proof of Lemma~\ref{lem:m=rclassRmat}}\label{appendix:m=rclassRmatprf}
	The proofs of~\ref{enum:fa},  \ref{enum:bonc},  \ref{enum:bond},  \ref{enum:bone},  \ref{enum:bonf},  \ref{enum:btwa},  \ref{enum:btwb},  \ref{enum:btwc},  \ref{enum:btwd} are similar to~\ref{enum:q},  \ref{enum:k},  \ref{enum:l},  \ref{enum:e},  \ref{enum:f},  \ref{enum:m},  \ref{enum:n},  \ref{enum:o},  \ref{enum:p}, respectively.
	
	The proof of~\ref{enum:bong} is similar to the proof~\ref{enum:p}, instead making substitutions $ \widetilde{a}_3=a_3,\,\widetilde{a}_2 = a_2+1,\, \widetilde{a}_1=a_1,\, \widetilde{b}_2=b_2,\,\widetilde{b}_1=b_1 $ (noting that we assume $ a_3=b_2=\widetilde{b}_2=0 $).
	
	If $ b_2<a_3 $ and $ b_2+b_1 < a_3+a_2 $, then
	\begin{equation}\label{eq:m=rinsbon}
		\begin{aligned}
			\mathbf{X}\bon \mathbf{y} \rightarrow 
			\begin{array}{|c|}
				\hline
				\mathbf{X}^{\ell}\\
				\btw^{a_3}\bon^{a_2}(\mathbf{x}_j)_{j=b_1+1}^{b_1+a_1}\\
				(\mathbf{z}_j)_{j=1}^{\ell}\\
				\hline
			\end{array}
			&=
			\bon \mathbf{y} \rightarrow
			\begin{array}{|c|c}
				\hline
				\multicolumn{2}{|c|}{\mathbf{X}^{\ell+1}}\\
				\cline{2-2}
				\btw^{a_3}\bon^{a_2} (\mathbf{x}_j)_{j=b_1+1}^{b_1+a_1}& \\
				(\mathbf{z}_j)_{j=1}^{\ell} & \\
				\cline{1-1}
			\end{array}\\
			&=
			\mathbf{y} \rightarrow
			\begin{array}{|c|c}
				\hline
				\multicolumn{2}{|c|}{\mathbf{X}^{\ell+1}}\\
				\multicolumn{2}{|c|}{\btw^{a_3}\bon^{a_2} (\mathbf{x}_j)_{j=b_1+1}^{b_1+a_1}\mathbf{z}_{\ell}}\\
				\cline{2-2}
				\bon(\mathbf{z}_j)_{j=1}^{\ell-1} & \\
				\cline{1-1}
			\end{array}\\
			&=
			\begin{array}{|c|c|c}
				\hline
				\multicolumn{3}{|c|}{\mathbf{X}^{\ell+1}}\\
				\multicolumn{3}{|c|}{\btw^{a_3}\bon^{a_2} (\mathbf{x}_j)_{j=b_1+1}^{b_1+a_1}\mathbf{z}_{\ell}}\\
				\cline{3-3}
				\multicolumn{2}{|c|}{\bon(\mathbf{z}_j)_{j=1}^{\ell-1}} & \\
				\cline{2-2}
				\mathbf{y}\\
				\cline{1-1}
			\end{array}
			.
		\end{aligned}
	\end{equation}
	If $ b_2=0 $ and $ b_2+b_1<a_3+a_2 $, then
	\begin{equation}\label{eq:m=rinsx0}
		\begin{aligned}
			\mathbf{X}\mathbf{x}_0 \mathbf{y} \rightarrow 
			\begin{array}{|c|}
				\hline
				\mathbf{X}^{\ell}\\
				\btw^{a_3}\bon^{a_2}(\mathbf{x}_j)_{j=b_1+1}^{b_1+a_1}\\
				(\mathbf{z}_j)_{j=1}^{\ell}\\
				\hline
			\end{array}
			&=
			\mathbf{x}_0 \mathbf{y} \rightarrow
			\begin{array}{|c|c}
				\hline
				\multicolumn{2}{|c|}{\mathbf{X}^{\ell+1}}\\
				\cline{2-2}
				\btw^{a_3}\bon^{a_2} (\mathbf{x}_j)_{j=b_1+1}^{b_1+a_1}& \\
				(\mathbf{z}_j)_{j=1}^{\ell} & \\
				\cline{1-1}
			\end{array}\\
			&=
			\mathbf{y} \rightarrow
			\begin{array}{|c|c}
				\hline
				\multicolumn{2}{|c|}{\mathbf{X}^{\ell+1}}\\
				\multicolumn{2}{|c|}{\btw^{a_3}\bon^{a_2} (\mathbf{x}_j)_{j=b_1+1}^{b_1+a_1}\mathbf{z}_{\ell}}\\
				\cline{2-2}
				\mathbf{x}_0 (\mathbf{z}_j)_{j=1}^{\ell-1} & \\
				\cline{1-1}
			\end{array}\\
			&=
			\begin{array}{|c|c|c}
				\hline
				\multicolumn{3}{|c|}{\mathbf{X}^{\ell+1}}\\
				\multicolumn{3}{|c|}{\btw^{a_3}\bon^b (\mathbf{x}_j)_{j=b_1+1}^{b_1+a_1}\mathbf{z}_{\ell}}\\
				\cline{3-3}
				\multicolumn{2}{|c|}{\mathbf{x}_0 (\mathbf{z}_j)_{j=1}^{\ell-1}} & \\
				\cline{2-2}
				\mathbf{y}\\
				\cline{1-1}
			\end{array}
			.
		\end{aligned}
	\end{equation}
	
	\begin{proof}[Proof of~\ref{enum:bonb}]
		Suppose $ b_2<a_3 $ and $ b_2+b_1<a_3+a_2 $ and $ \mathbf{y}<\mathbf{z}_{b_2+1} $.
		If $ a_1>0 $ then
		\begin{align*}
			&\col\left(
			\begin{array}{|c|}
				\hline
				\mathbf{X}^{\ell}\\
				\btw^{a_3-1}\bon^{a_2+1}(\mathbf{x}_j)_{j=b_1+1}^{b_1+a_1}\\
				\bon^{b_2} \mathbf{y} (\mathbf{x}_j)_{j=1}^{b_1} (\mathbf{z}_j)_{j=b_2+b_1+1}^{\ell-1}\\
				\hline
			\end{array}
			\right)
			\rightarrow
			\begin{array}{|c|}
				\hline
				\mathbf{X}\\
				\btw\\
				\mathbf{z}_{\ell}\\
				\hline
			\end{array}\allowdisplaybreaks\\
			&=
			\col\left(
			\begin{array}{|c|}
				\hline
				\mathbf{X}^{\ell-1}\\
				(\mathbf{h}_j)_{j=2}^{a_3+a_2}\bon (\mathbf{x}_j)_{j=b_1+1}^{b_1+a_1-1}\\
				\bon^{b_2} \mathbf{y} (\mathbf{x}_j)_{j=1}^{b_1} (\mathbf{z}_j)_{j=b_2+b_1+1}^{\ell-2}\\
				\hline
			\end{array}
			\right)
			\rightarrow
			\begin{array}{|c|c}
				\hline
				\multicolumn{2}{|c|}{\mathbf{X}^2}\\
				\multicolumn{2}{|c|}{\btw \mathbf{z}_{\ell}}\\
				\cline{2-2}
				\mathbf{x}_{b_1+a_1} &\\
				\mathbf{z}_{\ell-1}\\
				\cline{1-1}
			\end{array}\allowdisplaybreaks\\
			&=
			\col\left(
			\begin{array}{|c|}
				\hline
				\mathbf{X}^{\ell-a_1-1}\\
				(\mathbf{h}_j)_{j=2}^{a_3+a_2}\bon\\
				\bon^{b_2} \mathbf{y} (\mathbf{x}_j)_{j=1}^{b_1} (\mathbf{z}_j)_{j=b_2+b_1+1}^{\ell-a_1-1}\\
				\hline
			\end{array}
			\right)
			\rightarrow
			\begin{array}{|c|c|c}
				\hline
				\multicolumn{3}{|c|}{\mathbf{X}^{a_1+1}}\\
				\multicolumn{3}{|c|}{\btw (\mathbf{x}_j)_{j=b_1+2}^{b_1+a_1} \mathbf{z}_{\ell}}\\
				\cline{3-3}
				\multicolumn{2}{|c|}{\mathbf{x}_{b_1+1} (\mathbf{z}_j)_{j=\ell-a_1+1}^{\ell-1}}&\\
				\cline{2-2}
				\mathbf{z}_{\ell-a_1}\\
				\cline{1-1}
			\end{array}
			.
		\end{align*}
		Now, regardless of whether $ a_1>0 $ or $ a_1=0 $,
		\begin{align*}
			&
			\col\left(
			\begin{array}{|c|}
				\hline
				\mathbf{X}^{\ell}\\
				\btw^{a_3-1}\bon^{a_2+1}(\mathbf{x}_j)_{j=b_1+1}^{b_1+a_1}\\
				\bon^{b_2} \mathbf{y} (\mathbf{x}_j)_{j=1}^{b_1} (\mathbf{z}_j)_{j=b_2+b_1+1}^{\ell-1}\\
				\hline
			\end{array}
			\right)
			\rightarrow
			\begin{array}{|c|}
				\hline
				\mathbf{X}\\
				\btw\\
				\mathbf{z}_{\ell}\\
				\hline
			\end{array}\allowdisplaybreaks\\
			&=
			\begin{cases}
				\col\left(
				\begin{array}{|c|}
					\hline
					\mathbf{X}^{b_2+b_1}\\
					(\mathbf{h}_j)_{j=2}^{b_2+b_1+1}\\
					\bon^{b_2} \mathbf{y}(\mathbf{x}_j)_{j=1}^{b_1-1} \\
					\hline
				\end{array}
				\right)
				\rightarrow
				\begin{array}{|c|c|c}
					\hline
					\multicolumn{3}{|c|}{\mathbf{X}^{\ell-b_2-b_1+1}}\\
					\multicolumn{3}{|c|}{\btw (\mathbf{h}_j)_{j=b_2+b_1+2}^{a_3+a_2} (\mathbf{x}_j)_{j=b_1+1}^{b_1+a_1} \mathbf{z}_{\ell}}\\
					\cline{3-3}
					\multicolumn{2}{|c|}{\bon (\mathbf{z}_j)_{j=b_2+b_1+1}^{\ell-1}}&\\
					\cline{2-2}
					\mathbf{x}_{b_1}\\
					\cline{1-1}
				\end{array} & \text{if}~b_1>0 \\
				\col\left(
				\begin{array}{|c|}
					\hline
					\mathbf{X}^{b_2+b_1}\\
					(\mathbf{h}_j)_{j=2}^{b_2+b_1+1}\\
					\bon^{b_2}\\
					\hline
				\end{array}
				\right)
				\rightarrow
				\begin{array}{|c|c|c}
					\hline
					\multicolumn{3}{|c|}{\mathbf{X}^{\ell-b_2-b_1+1}}\\
					\multicolumn{3}{|c|}{\btw (\mathbf{h}_j)_{j=b_2+b_1+2}^{a_3+a_2} (\mathbf{x}_j)_{j=b_1+1}^{b_1+a_1} \mathbf{z}_{\ell}}\\
					\cline{3-3}
					\multicolumn{2}{|c|}{\bon (\mathbf{z}_j)_{j=b_2+b_1+1}^{\ell-1}}&\\
					\cline{2-2}
					\mathbf{y}\\
					\cline{1-1}
				\end{array} & \text{if}~b_1=0
			\end{cases}\allowdisplaybreaks\\
			&=
			\col\left(
			\begin{array}{|c|}
				\hline
				\mathbf{X}^{b_2}\\
				(\mathbf{h}_j)_{j=2}^{b_2+1}\\
				\bon^{b_2}\\
				\hline
			\end{array}
			\right)
			\rightarrow
			\begin{array}{|c|c|c}
				\hline
				\multicolumn{3}{|c|}{\mathbf{X}^{\ell-b_2+1}}\\
				\multicolumn{3}{|c|}{\btw (\mathbf{h}_j)_{j=b_2+2}^{a_3+a_2} (\mathbf{x}_j)_{j=b_1+1}^{b_1+a_1} \mathbf{z}_{\ell}}\\
				\cline{3-3}
				\multicolumn{2}{|c|}{\bon (\mathbf{x}_j)_{j=1}^{b_1} (\mathbf{z}_j)_{j=b_2+b_1+1}^{\ell-1}}&\\
				\cline{2-2}
				\mathbf{y}\\
				\cline{1-1}
			\end{array} \\
			&=
			\begin{array}{|c|c|c}
				\hline
				\multicolumn{3}{|c|}{\mathbf{X}^{\ell+1}}\\
				\multicolumn{3}{|c|}{\btw (\mathbf{h}_j)_{j=2}^{a_3+a_2} (\mathbf{x}_j)_{j=b_1+1}^{b_1+a_1} \mathbf{z}_{\ell}}\\
				\cline{3-3}
				\multicolumn{2}{|c|}{\bon^{b_2+1}  (\mathbf{z}_j)_{j=b_2+b_1+1}^{\ell-1}}&\\
				\cline{2-2}
				\mathbf{y}\\
				\cline{1-1}
			\end{array}
			=
			\begin{array}{|c|c|c}
				\hline
				\multicolumn{3}{|c|}{\mathbf{X}^{\ell+1}}\\
				\multicolumn{3}{|c|}{\btw^{a_3}\bon^{a_2} (\mathbf{x}_j)_{j=b_1+1}^{b_1+a_1} \mathbf{z}_{\ell}}\\
				\cline{3-3}
				\multicolumn{2}{|c|}{\bon (\mathbf{z}_j)_{j=1}^{\ell-1}}&\\
				\cline{2-2}
				\mathbf{y}\\
				\cline{1-1}
			\end{array}
			.
		\end{align*}
		This above insertion is the same as~\eqref{eq:m=rinsbon}, proving~\ref{enum:bonb}.
	\end{proof}
	
	\begin{proof}[Proof of~\ref{enum:bona} and~\ref{enum:fb}]
		Let $ \mathbf{v} = \bon $ for~\ref{enum:bona} and $ \mathbf{v} = \mathbf{x}_0 $ for~\ref{enum:fb}.
		Under the assumptions for~\eqref{eq:m=rinsbon} or~\eqref{eq:m=rinsx0},
		additionally assume that $ b_1>0 $ and $ \mathbf{y}\ge \mathbf{x}_{b_1} $ if $ \mathbf{v}=\bon $.
		
		Set,
		\[
		(\widetilde{\mathbf{z}}_j)_{j=1}^{b_2+b_1}=
		\begin{cases}
			\bon^{b_2+1}(\mathbf{x}_j)_{j=1}^{b_1-1} & \text{for~\ref{enum:bona}}\\
			\mathbf{x}_0(\mathbf{x}_j)_{j=1}^{b_1-1} & \text{for~\ref{enum:fb}}
		\end{cases}
		\]
		Now,
		\begin{align*}
			&
			\col\left(
			\begin{array}{|c|}
				\hline
				\mathbf{X}^{\ell}\\
				(\mathbf{h}_j)_{j=1}^{a_3+a_2-1}(\mathbf{x}_j)_{j=b_1}^{b_1+a_1}\\
				(\widetilde{\mathbf{z}}_j)_{j=1}^{b_2+b_1}\mathbf{y}(\mathbf{z}_j)_{j=b_2+b_1+1}^{\ell-1}\\
				\hline
			\end{array}
			\right)
			\rightarrow
			\begin{array}{|c|}
				\hline
				\mathbf{X}\\
				\mathbf{h}_{a_3+a_2}\\
				\mathbf{z}_{\ell}\\
				\hline
			\end{array}\allowdisplaybreaks\\
			&=
			\col\left(
			\begin{array}{|c|}
				\hline
				\mathbf{X}^{\ell-1}\\
				(\mathbf{h}_j)_{j=1}^{a_3+a_2-1}(\mathbf{x}_j)_{j=b_1}^{b_1+a_1-1}\\
				(\widetilde{\mathbf{z}}_j)_{j=1}^{b_2+b_1}\mathbf{y}(\mathbf{z}_j)_{j=b_2+b_1+1}^{\ell-2}\\
				\hline
			\end{array}
			\right)
			\rightarrow
			\begin{array}{|c|c}
				\hline
				\multicolumn{2}{|c|}{\mathbf{X}^2}\\
				\multicolumn{2}{|c|}{\mathbf{h}_{a_3+a_2}\mathbf{z}_{\ell}}\\
				\cline{2-2}
				\mathbf{x}_{b_1+a_1} & \\
				\mathbf{z}_{\ell-1} & \\
				\cline{1-1}
			\end{array}\allowdisplaybreaks\\
			&=
			\begin{cases}
				\col\left(
				\begin{array}{|c|}
					\hline
					\mathbf{X}^{\ell-a_1-1}\\
					(\mathbf{h}_j)_{j=1}^{a_3+a_2-1}\\
					(\widetilde{\mathbf{z}}_j)_{j=1}^{b_2+b_1}\\
					\hline
				\end{array}
				\right)
				\rightarrow
				\begin{array}{|c|c|c}
					\hline
					\multicolumn{3}{|c|}{\mathbf{X}^{a_1+2}}\\
					\multicolumn{3}{|c|}{\mathbf{h}_{a_3+a_2}(\mathbf{x}_j)_{j=b_1+1}^{b_1+a_1}\mathbf{z}_{\ell}}\\
					\cline{3-3}
					\multicolumn{2}{|c|}{\mathbf{x}_{b_1}(\mathbf{z}_j)_{j=b_2+b_1+1}^{\ell-1}} & \\
					\cline{2-2}
					\mathbf{y}\\
					\cline{1-1}
				\end{array}
				& \text{if $ b_2+b_1=a_3+a_2-1 $}\\
				\col\left(
				\begin{array}{|c|}
					\hline
					\mathbf{X}^{\ell-a_1-1}\\
					(\mathbf{h}_j)_{j=1}^{a_3+a_2-1}\\
					(\widetilde{\mathbf{z}}_j)_{j=1}^{b_2+b_1}\mathbf{y}(\mathbf{z}_j)_{j=b_2+b_1+1}^{\ell-a_1-2}\\
					\hline
				\end{array}
				\right)
				\rightarrow
				\begin{array}{|c|c|c}
					\hline
					\multicolumn{3}{|c|}{\mathbf{X}^{a_1+2}}\\
					\multicolumn{3}{|c|}{\mathbf{h}_{a_3+a_2}(\mathbf{x}_j)_{j=b_1+1}^{b_1+a_1}\mathbf{z}_{\ell}}\\
					\cline{3-3}
					\multicolumn{2}{|c|}{\mathbf{x}_{b_1}(\mathbf{z}_j)_{j=\ell-a_1}^{\ell-1}} & \\
					\cline{2-2}
					\mathbf{z}_{\ell-a_1-1}\\
					\cline{1-1}
				\end{array}
				& \text{if $ b_2+b_1\ne a_3+a_2-1 $}
			\end{cases}\allowdisplaybreaks\\
			&=
			\col\left(
			\begin{array}{|c|}
				\hline
				\mathbf{X}^{b_2+b_1}\\
				(\mathbf{h}_j)_{j=1}^{b_2+b_1}\\
				(\widetilde{\mathbf{z}}_j)_{j=1}^{b_2+b_1}\\
				\hline
			\end{array}
			\right)
			\rightarrow
			\begin{array}{|c|c|c}
				\hline
				\multicolumn{3}{|c|}{\mathbf{X}^{\ell-b_2-b_1+1}}\\
				\multicolumn{3}{|c|}{(\mathbf{h}_{a_3+a_2})_{j=b_2+b_1+1}^{a_3+a_2}(\mathbf{x}_j)_{j=b_1+1}^{b_1+a_1}\mathbf{z}_{\ell}}\\
				\cline{3-3}
				\multicolumn{2}{|c|}{\mathbf{x}_{b_1}(\mathbf{z}_j)_{j=b_2+b_1+1}^{\ell-1}} & \\
				\cline{2-2}
				\mathbf{y}\\
				\cline{1-1}
			\end{array}\allowdisplaybreaks\\
			&=
			\begin{array}{|c|c|c}
				\hline
				\multicolumn{3}{|c|}{\mathbf{X}^{\ell+1}}\\
				\multicolumn{3}{|c|}{(\mathbf{h}_j)_{j=1}^{a_3+a_2}(\mathbf{x}_j)_{j=b_1+1}^{b_1+a_1}\mathbf{z}_{\ell}}\\
				\cline{3-3}
				\multicolumn{2}{|c|}{(\widetilde{\mathbf{z}}_j)_{j=1}^{b_2+b_1}\mathbf{x}_{b_1}(\mathbf{z}_j)_{j=b_2+b_1+1}^{\ell-1}} & \\
				\cline{2-2}
				\mathbf{y}\\
				\cline{1-1}
			\end{array}
			=
			\begin{array}{|c|c|c}
				\hline
				\multicolumn{3}{|c|}{\mathbf{X}^{\ell+1}}\\
				\multicolumn{3}{|c|}{\btw^{a_3}\bon^{a_2}(\mathbf{x}_j)_{j=b_1+1}^{b_1+a_1} \mathbf{z}_{\ell}}\\
				\cline{3-3}
				\multicolumn{2}{|c|}{\mathbf{v}(\mathbf{z}_j)_{j=1}^{\ell-1}} & \\
				\cline{2-2}
				\mathbf{y}\\
				\cline{1-1}
			\end{array}
		\end{align*}
		This insertion is the same as in~\eqref{eq:m=rinsbon} or~\eqref{eq:m=rinsx0}.
		Therefore,
		\[
		R\left(
		\begin{array}{|c|}
			\hline
			\mathbf{X}^{\ell}\\
			(\mathbf{h}_j)_{j=1}^{\ell}\\
			(\mathbf{z}_j)_{j=1}^{\ell}\\
			\hline
		\end{array}
		\otimes 
		\begin{array}{|c|}
			\hline
			\mathbf{X}\\
			\mathbf{v} \\
			\mathbf{y}\\
			\hline
		\end{array}
		\right)
		=
		\begin{array}{|c|}
			\hline
			\mathbf{X}\\
			\mathbf{h}_{a_3+a_2}\\
			\mathbf{z}_{\ell}\\
			\hline
		\end{array}\\
		\otimes
		\begin{array}{|c|}
			\hline
			\mathbf{X}^{\ell}\\
			(\mathbf{h}_j)_{j=1}^{a_3+a_2-1}(\mathbf{x}_j)_{j=b_1+1}^{b_1+a_1}\\
			(\widetilde{\mathbf{z}}_j)_{j=1}^{b_2+b_1}\mathbf{y}(\mathbf{z}_j)_{j=b_2+b_1+1}^{\ell-1}\\
			\hline
		\end{array}
		\]
		
		Substituting $ \mathbf{v}=\bon $ and $ (\widetilde{\mathbf{z}}_j)_{j=1}^{b_2+b_1}=\bon^{b_2+1}(\mathbf{x}_j)_{j=1}^{b_1-1} $
		proves~\ref{enum:bona}.
		
		Substituting $ \mathbf{v}=\mathbf{x}_0 $ and $ (\widetilde{\mathbf{z}}_j)_{j=1}^{b_2+b_1}=(\mathbf{x}_j)_{j=0}^{b_1-1} $
		proves~\ref{enum:fb}.
	\end{proof}
	
	\subsection*{Declarations}
	%
	B.~Solomon was partially funded by the AMSI Vacation Research Scholarships 2020--21.
	
	\bibliographystyle{alpha}
	\bibliography{solitons}{}
\end{document}